\documentclass[11pt,twoside]{article}
\usepackage{fancyhdr}
\usepackage{epsfig,graphicx}
\usepackage{afterpage}
\usepackage{amsfonts,amsmath,amssymb,amsthm} 
\usepackage{fullpage}
\usepackage[T1]{fontenc} 
\usepackage{epsf} 
\usepackage{graphics} 
\usepackage[sort,numbers]{natbib} 

\usepackage{psfrag,xspace}
\usepackage{color,etoolbox}
\usepackage{graphicx}
\usepackage{caption}
\usepackage{subcaption}
\usepackage{svg,adjustbox}
\usepackage{xspace}
\usepackage{bm}
\usepackage{bbm}

\usepackage{setspace}

\usepackage{marginnote} 
\theoremstyle{plain}
\usepackage{sansmath}
\usepackage{mathrsfs}

\usepackage[colorlinks,citecolor=blue,urlcolor=blue,linkcolor=blue,bookmarks=false]{hyperref}

\renewcommand{\d}{\,{\mathrm d}}

\newcommand{\Bern}{\mbox{$\mathsf{Bern}$}}

\newcommand{\Geom}{\mbox{$\mathsf{Geom}$}}

\newcommand{\Normal}{\mbox{$\mathsf{N}$}}

\newtheorem{theorem}{Theorem}

\newtheorem{assumption}{Assumption}

\newtheorem{example}{Example}
\newtheorem{proposition}{Proposition}
\newtheorem{corollary}{Corollary}

\makeatletter
\newcommand*{\rom}[1]{\expandafter\@slowromancap\romannumeral #1@}
\makeatother

\newcommand{\tr}[1]{\operatorname{tr}\left({#1}\right)}
\newcommand{\pnorm}[2]{\left\|#2\right\|_{#1}}

\newcommand{\supnorm}[1]{\pnorm{\infty}{#1}}

\renewcommand{\complement}{c}
\let\hat\widehat
\let\tilde\widetilde

\newcommand{\Chat}{\widehat{C}}

\DeclareMathOperator*{\argmax}{arg\,max}
\DeclareMathOperator*{\argmin}{arg\,min}
\DeclareMathOperator*{\arginf}{arg\,inf}
\DeclareMathOperator*{\diam}{diam}

\newcommand{\Truth}{\text{clean}}

\newcommand{\Redi}{\textsc{ReDI}\xspace}
\newcommand{\Wald}{\textsc{wald}\xspace}

\newcommand{\bbE}{\mathbb{E}}
\newcommand{\bbP}{\mathbb{P}}
\newcommand{\bbV}{\mathbb{V}}
\newcommand{\bbR}{\mathbb{R}}

\newcommand{\I}{\mathbb{I}}

\newcommand{\thetahat}{\widehat{\theta}}
\DeclareMathOperator*{\evidence}{ev}
\newcommand{\sgntilde}[1]{\widetilde{\operatorname{sgn}}\left({#1}\right)}

\newcommand{\calH}{\mathcal{H}}
\newcommand{\calP}{\mathcal{P}}
\newcommand{\calF}{\mathcal{F}}
\newcommand{\calA}{\mathcal{A}}
\newcommand{\calL}{\mathcal{L}}
\newcommand{\calS}{\mathcal{S}}
\newcommand{\calX}{\mathcal{X}}
\newcommand{\calY}{\mathcal{Y}}
\newcommand{\calM}{\mathcal{M}}

\newcommand{\calI}{\mathcal{I}}

\newcommand{\scH}{\mathscr{H}}

\newcommand{\KL}{\mathsf{KL}}
\newcommand{\DP}{\mathsf{DP}}
\newcommand{\TV}{\mathsf{TV}}
\newcommand{\HEL}{{\sf H}}
\newcommand{\Was}{\mathsf{W}}
\newcommand{\IPM}{{\sf IPM}}

\newcommand{\truedist}{P^*}
\newcommand{\projdist}{\widetilde{P}}
\newcommand{\trueclass}{\mathcal{Q}}
\newcommand{\model}{\mathcal{P}}
\newcommand{\pilot}{\widehat{P}_1}
\newcommand{\talpha}{t_\alpha}
\newcommand{\Data}{\mathcal{D}}

\newcommand{\projtheta}{\tilde{\theta}}

\begin{document}

\begin{center} {\LARGE{\bf{Robust Universal Inference For Misspecified Models}}}
\\

\vspace*{.3in}

{\large{
\begin{tabular}{ccccc}
Beomjo Park{$^*$}, Sivaraman Balakrishnan{$^{*,\dagger}$}, and Larry Wasserman{$^{*,\dagger}$}\\
\end{tabular}

\vspace*{.1in}

\begin{tabular}{ccc}
$^*$Department of Statistics \& Data Science \\
$^\dagger$Machine Learning Department \\
\end{tabular}

\begin{tabular}{c}
Carnegie Mellon University \\
Pittsburgh, PA 15213.
\end{tabular}

\vspace*{.2in}

\begin{tabular}{c}
{\texttt{beomjop@alumni.cmu.edu}, \texttt{\{siva,larry\}@stat.cmu.edu}}
\end{tabular}
}}

\vspace*{.2in}

\today
\vspace*{.2in}
\begin{abstract}
In statistical inference, it is rarely realistic that the hypothesized statistical model is
well-specified, and consequently, it is important to understand the effects of misspecification on inferential procedures. 
When the hypothesized statistical model is misspecified, the natural target of inference is a projection of the data generating distribution onto the model. 
We present a general method for constructing valid confidence sets for such projections, under weak regularity conditions, despite possible model misspecification.
Our method builds upon the universal inference method
and 
is based on inverting a family of split-sample tests of relative fit. We study settings in which our methods yield either exact or approximate,
finite-sample valid confidence sets for various projection distributions. We study rates at which the resulting confidence sets shrink around their target of inference and complement these results 
with a simulation study and a study of causal discovery using a linear causal model with the \textsc{CausalEffectPairs} dataset.
\end{abstract}
\end{center}

\section{Introduction}
One of the broad goals of statistical inference is to draw conclusions about a
distribution $\truedist$ from a sample of the population. This goal is typically facilitated by the use of a
statistical model $\mathcal{P}$, a collection of distributions, which the statistician hypothesizes will contain a useful approximation to the data generating distribution.
The well-specified case is when $\truedist \in \model$ and the misspecified case is when this does not necessarily hold. 

In the misspecified case, the target of inference is usually
a \emph{projection distribution}. Formally, given a divergence $\rho$ which maps a pair of distributions to $\mathbb{R}_+$, we can define the projection\footnote{We tacitly assume that the projection exists and is unique. When the projection is not unique our inferential guarantees always hold for any (arbitrary) fixed choice of the projection $\projdist$. Characterizing the existence of a projection distribution (for $f$-divergences) has received some attention in past work~\cite{li1999estimation}.} of the distribution 
$\truedist$ onto the statistical model as:
\begin{align}
\label{eqn:target}
\projdist := \arginf_{P \in \model} \rho(\truedist \| P). 
\end{align}
{
For a parametric model
${\model} = \{P_\theta:\ \theta\in\Theta\subseteq \bbR^d \}$, we define the \emph{projection parameter}
$\projtheta\equiv \projtheta(P^*)$ as the element of $\Theta$ that
minimizes
$\rho(\truedist \|P_\theta)$.
} 
The general goal of our paper is to construct \emph{uniformly} valid confidence sets for $\projdist$ assuming only weak regularity conditions
on the distribution $\truedist$ and the statistical model $\model$.
We let $X_1,\ldots, X_{n}$ be an i.i.d sample from a distribution $\truedist \in \trueclass$ 
 where $\mathcal{Q} \supseteq \mathcal{P}$ is a class of distributions 
satisfying weak regularity conditions. 
We wish to construct (honest)
confidence sets, $C_{\alpha}(X_1,\ldots, X_{n})$ 
such that,
\begin{align}
\label{eqn:exact_honest_valid_set}
\inf_{\truedist \in \trueclass} \mathbb{P}_{\truedist}(\projdist \in C_{\alpha}(X_1,\ldots, X_{n})) \geq 1 - \alpha. 
\end{align} 
{In some cases, we construct
\emph{asymptotic} confidence sets which are valid as $n \rightarrow \infty$ even for irregular models.}

{

In parametric models minimum distance methods
\citep{beran_minimum_1977, basu1994}
explicitly aim to estimate {projection parameter} $\projtheta$.
The maximum likelihood estimator can be viewed
as a minimum distance estimator when taking $\rho$ to be the Kullback-Leibler (KL) divergence.
However, absent strong regularity conditions, the KL projection
is highly non-robust to small perturbations in $\truedist$,
even tiny perturbations to $\truedist$ can lead to vastly different projections $\projdist$ \citep{beran_minimum_1977,basu1994}. From a practical standpoint, these instabilities can make the KL projection an
undesirable target, and in these cases it is essential to develop a flexible family of methods
that can target other (more stable) projection distributions. Each choice
of the divergence yields a different target. In some cases, these targets will have
drastically different properties. This in turn poses significant challenges in constructing a
unified framework for statistical inference in the misspecified setting.}

{
Classically the Huber-White approach to inference under model misspecification relies on a Taylor expansion of the log-likelihood to characterize the limiting distribution of the maximum likelihood estimator, together with a ``sandwich estimator'' of the limiting variance~\citep{huber_robust_1965,white_maximum_1982,freedman_so-called_2006}.
The validity of this approach relies
on the model ${\model}$ and the sampling distribution $\truedist$ 
satisfying a number of strong regularity conditions. 
As a simple illustration, consider the following example.}
{
\begin{example} \label{ex:mix_unident}
Suppose we observe
$Y_1,\ldots, Y_n \overset{iid}{\sim} P^*$
where $Y_i\in \{0,1,\ldots \}$, and $P^*$ is a Negative Binomial distribution with mean 4 and variance 18. We want to approximate $\truedist$
with a mixture of {Geometric} distributions
$(1-\pi) \Geom(\lambda_1) + \pi \Geom(\lambda_2)$
where $\pi \in [0,1]$
and $\Geom(\lambda)$ denotes a {Geometric} with mean $\lambda \in (0,1]$.
The parameter is $\theta  = (\pi,\lambda_1,\lambda_2)$.
This model
fails to satisfy the usual regularity conditions and $\theta$ is not identifiable in general. Our choice of $P^*$ is well-approximated by a single
Geometric distribution, and consequently the $L^2$ projection parameter $\projtheta$ is a set of
parameter values $\projtheta_{L_2} = ([0,1], 0.19, 0.19)$.
Standard asymptotic methods for constructing confidence sets for
the projection parameter fail. 
Figure \ref{fig::mix}
shows a valid confidence set
for $L_2$ projection parameter $\projtheta$
computed using the numerical algorithm in Section \ref{section::prac}.
This unusual confidence set
correctly mitigates the two main issues --- the lack of identifiability and model misspecification --- to yield a valid confidence set. Specifically, the confidence set shows that if $\lambda_1\approx 0.2$ then
$\lambda_2$ is essentially unconstrained and vice versa.
\end{example}
}
\begin{figure}[!htb]
\begin{center}
\includegraphics[trim={0 19 0 59},clip,scale=.40]{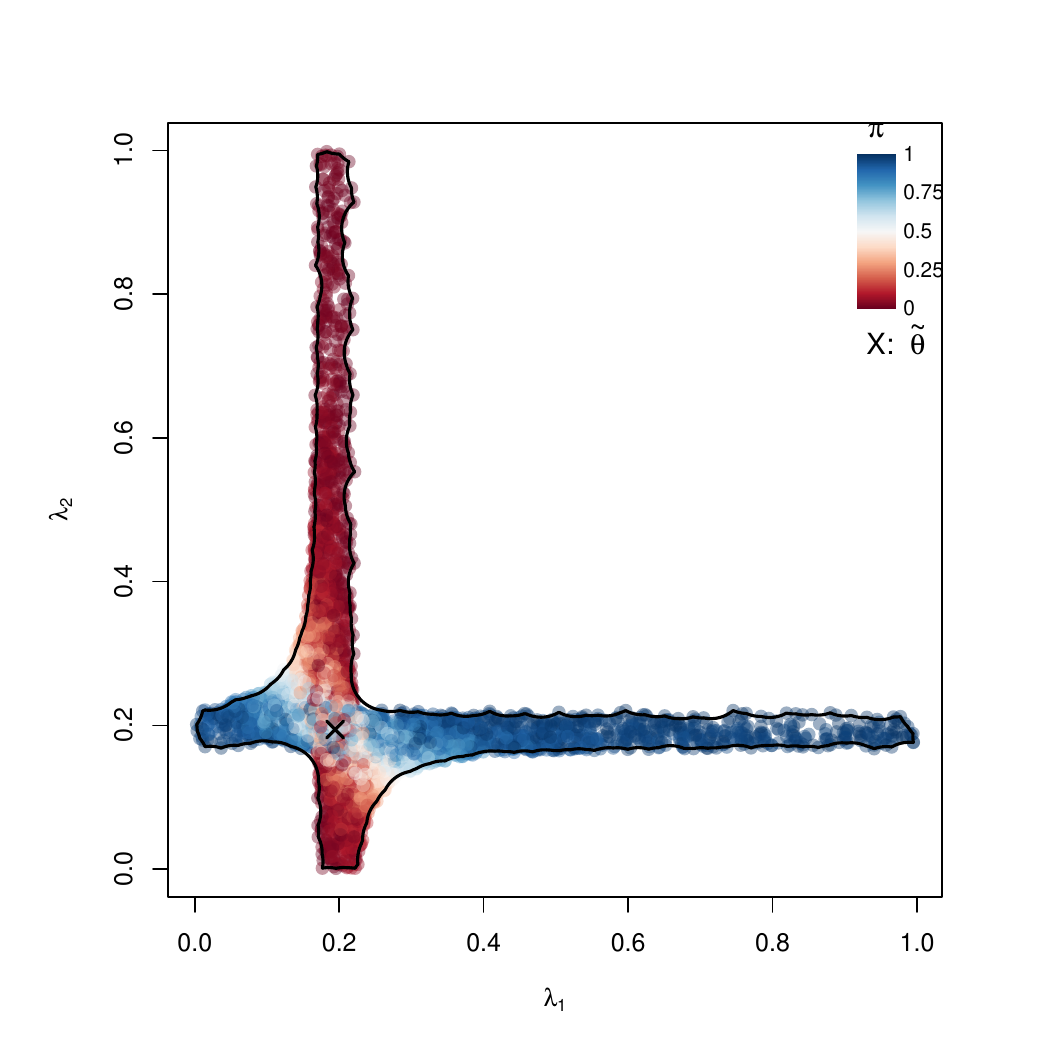}
\end{center}
\caption{{Robust Universal Confidence set for the $L_2$ projection $\projtheta_{L_2} = ([0,1], 0.19, 0.19)$
when 
the true distribution is a Negative Binomial with mean 4 and variance 18.
The Wald confidence set for KL projection $\projtheta_{\KL} = ([0,1], 0.2, 0.2)$ using the Huber-White sandwich estimator is not meaningful since in this example the Fisher information matrix is degenerate.
}
}
\label{fig::mix}
\end{figure}

\noindent {The mixture of Geometric distributions 
is a model that arises in the estimation of microbial species diversity \citep{bunge_estimating_2014, bunge_estimating_2012, hong_predicting_2006}. More broadly, inferential problems with unidentifiable parameters, including those in mixture models, arise in diverse applications \citep{prakasa_rao_identifiability_1992,catchpole_detecting_1997,algeri_searching_2020}. 
See \citet{brazzale_likelihood_2024} for a comprehensive review of irregular problems where likelihood asymptotics fail.}

In recent work~\citep{wasserman_universal_2020}, we introduced a universal inference procedure (described in more detail in Section~\ref{sec:Background})
based on data-splitting to construct uniformly, finite-sample valid likelihood-ratio confidence sets under no regularity conditions. This work showed that, 
in the well-specified setting, sample-splitting can yield practical, finite-sample valid inference, even for irregular statistical models, often at a surprisingly small statistical price. {As highlighted by recent work~\citep{tse_note_2022,dunn2022gaussian,strieder_choice_2022} sample-splitting reduces efficiency
compared to standard methods in regular problems and universal inference is often conservative in these settings. Universal inference, and the methods we propose in this paper, are most useful for irregular problems
where standard methods fail.}

To address the challenges in statistical inference under model misspecification and weak regularity conditions, we develop 
a re-interpretation of the universal inference method \citep{wasserman_universal_2020} as inverting a particular family of pairwise likelihood-ratio tests. This interpretation
brings into focus the key building block of universal inferential methods---pairwise hypothesis tests. Building on this insight
we develop \emph{robust} universal inference procedures by inverting appropriate families of \emph{robust pairwise tests}. 
{Our approach provides a unified framework for constructing robust confidence sets, regardless of
the specific form of the misspecification.} 
We then study 
the design and properties of robust pairwise tests, and relate them to the coverage and size properties of the proposed robust universal inference method.

\subsection{Related Work}
Asymptotic statistical inference, in both the well-specified and misspecified cases, is a topic 
of classical interest. 
Some entry points to the vast literature on this topic include the reference books \cite{huber_robust_1981,vaart_asymptotic_1998}. Results in this literature \citep{kleijn_misspecification_2006,kleijn_bernstein-von-mises_2012}
typically leverage strong regularity conditions to determine the asymptotic distribution of a point estimate (such as the Maximum Likelihood Estimator (MLE)), and use the asymptotic distribution of the estimate 
to construct (asymptotically valid) confidence sets. 
Our work is motivated in part by a recent line of work \citep{rinaldo_bootstrapping_2019,chakravarti_gaussian_2019,wasserman_universal_2020,kuchibhotla_hulc_2021}, and {more classical work~\citep{cox1975} where sample-splitting is used to avoid the strong regularity conditions typically needed for valid statistical inference.}

The focus on statistical inference under weaker regularity conditions, despite model misspecification, is the central theme of work in robust statistics~\citep{tukey_mathematics_1975,huber_robust_1981,hampel_robust_1986}.
One of the best understood methods for constructing robust estimators is to select, from a set of candidates, one which wins a carefully setup tournament -- an idea which goes back to \citet{lecam_convergence_1973,birge_sur_1979,yatracos_rates_1985,donoho_automatic_1988,devroye_combinatorial_2001}, and 
others. At the heart of these tournament estimators are pairwise selectors, which attempt to robustly select one of a pair of candidates, which provide a better \emph{relative} fit to the sampling distribution. 
These robust pairwise tests have been used to great effect in robust estimation, and our work highlights their usefulness in constructing assumption-light confidence sets.

\section{Background} \label{sec:Background}

We let $X_1,\ldots, X_{n}$ be an i.i.d sample from a distribution $\truedist \in \trueclass$
, and 
we let $\model$ denote our working statistical model.
{We assume that the distributions in $\model$ and $\trueclass$ have densities (denoted with lowercase symbols) with respect to a common dominating measure. }
Throughout the paper, the collection of distributions $\trueclass$ will be quite general, typically only satisfying some weak regularity conditions. 

\subsection{Universal Inference}
Our starting point is our prior work~\citep{wasserman_universal_2020} which 
introduced a procedure
based on data-splitting to construct uniformly, finite-sample valid confidence sets under weak regularity conditions. 
Importantly, the validity guarantees of universal inference 
require the statistical model $\model$ to be 
correctly specified. The universal inference procedure is to:
\begin{enumerate}
\item Split the data $\Data := \{X_1,\ldots,X_n\}$ into two sets $\Data_0$ and $\Data_1$. 
\item On the set $\Data_1$ calculate any estimate $\pilot$ (e.g., $\pilot$ could be the MLE in the model $\model$). 
\item 
Let $\mathcal{L}_0(P)$ denote the likelihood of the distribution $P$ evaluated on the samples in $\Data_0$,
$\mathcal{L}_0(P) := \prod_{i \in \Data_0} p(X_i),$
and define $\mathcal{L}_0(\pilot)$ analogously. 
Then construct the confidence set,
\begin{align}
\label{eqn:univ_set}
C_{\alpha}(X_1,\ldots,X_n) = \left\{ P: \frac{\mathcal{L}_0(P)}{\mathcal{L}_0(\pilot)} \geq \alpha \right\}.
\end{align} 
\end{enumerate}
In the well-specified case, \citet{wasserman_universal_2020} show (in their Theorem 1) that,
under no additional regularity conditions, $C_{\alpha}$ is a finite-sample 
valid $1 - \alpha$ confidence set for 
the distribution~$\truedist$.

\vspace{.2cm}

\subsection {A Re-Interpretation of Universal Inference} 
\label{sec:reinterpret_ui}

To motivate the main proposal of this paper, it is useful to re-visit (and generalize)
the procedure described above, via the lens of inverting a family of hypothesis tests. The basic idea is classical, and is sometimes referred
to as the duality between confidence sets and hypothesis tests. 
Formally, given samples $X_1,\ldots, X_n \sim \truedist$,
suppose we have a family of tests 
\begin{align*}
\phi_P: \{X_1,\ldots,X_n\} \mapsto \{0,1\}
\end{align*} 
for testing the null hypothesis $H_0: \truedist = P$.  Here the test function 
 $\phi_P$ takes the value 1 to indicate a rejection of the null hypothesis and takes the value 0 otherwise.
If the family of tests is valid, i.e. they control the Type I error, 
\begin{align}
\label{eqn:typeone}
\bbE_P \left[\phi_P(X_1,\ldots,X_n) \right]\leq \alpha, ~~~\forall~P \in \model,
\end{align}
then the following confidence set,
\begin{align}
\label{eqn:univ_set_two}
C_{\alpha}(X_1,\ldots,X_n) := \left\{ P \in \model: \phi_P = 0 \right\}, 
\end{align}
is uniformly valid when the statistical model is correctly specified, i.e. 
\begin{align*}
\inf_{\truedist \in \model} \mathbb{P}_{\truedist}(\truedist \in C_{\alpha}(X_1,\ldots, X_{n})) \geq 1 - \alpha. 
\end{align*} 
Although this is a general recipe for constructing valid confidence sets, it does not provide the statistician much guidance in designing 
tests which might lead to small confidence sets. 

Universal inference is based on the idea that one can use a separate sample 
to construct an accurate estimate $\pilot$. We can then construct our family of tests, on the remaining samples, to have high power in distinguishing 
the sampling distribution from this pilot estimate. Formally, we could choose our family of tests 
to have high power to distinguish the hypotheses:
\begin{align*}
H_0: ~\truedist = P,  \qquad \text{versus} \qquad
H_1: ~\truedist = \pilot.
\end{align*}
This use of a separate sample to construct a pilot estimate, 
simplifies the design of the tests to invert considerably since now we can focus on 
tests that have strong guarantees for distinguishing this simple null-versus-simple alternative. Indeed, universal inference uses 
the likelihood-ratio test for distinguishing these hypotheses, resulting in tests $\phi_P$ of the form:
\begin{align*}
\phi_P = \mathbb{I}\left[ \frac{\mathcal{L}_0(\pilot)}{ \mathcal{L}_0(P)} > \talpha(P, \pilot) \right],
\end{align*}
for a choice of the threshold $\talpha(P, \pilot)$ which ensures that the condition in~\eqref{eqn:typeone} is satisfied. Although it is possible in principle to determine \emph{optimal} thresholds in the likelihood-ratio tests above, this can be practically cumbersome since these thresholds depend on both the pilot estimate $\pilot$ and the null
hypothesis $P$ under consideration. The work of \citet{wasserman_universal_2020} further shows that a universal threshold $\talpha = 1/\alpha$ suffices to ensure the condition in~\eqref{eqn:typeone}. To summarize,
one can view the universal inference confidence set~\eqref{eqn:univ_set} as arising by inverting a family of likelihood-ratio tests designed to distinguish
each candidate distribution $P$ from a pilot estimate $\pilot$.  

We emphasize that the universal inference procedure, and its reinterpretation described above rely crucially on correct model-specification to ensure validity.
For instance, inverting a family of tests that satisfies \eqref{eqn:typeone} is no longer meaningful when the model is misspecified. 
However, the testing interpretation suggests that one might develop novel variants of the universal inference procedure which 
are useful despite model-misspecification, by formulating appropriate robust hypotheses and designing robust tests for distinguishing them. We make these ideas
precise in Section~\ref{sec:rui}.

\subsection{Divergences}
Throughout this paper, we make frequent use of different divergences between pairs of 
probability distributions. We briefly introduce them here. We let $P$ and $Q$ be distributions 
with densities $p$ and $q$ with respect to a common dominating measure $\lambda$. 

The Hellinger distance is defined as:
$\HEL(P,Q) = \frac{1}{\sqrt{2}} \left( \int (\sqrt{p} - \sqrt{q})^2 d\lambda \right)^{1/2}, $
and the Kullback-Leibler (KL) divergence is defined as:
\begin{align*}
\KL(P \| Q) = 
\begin{cases} 
\int \left(\log \frac{p}{q} \right) \d P, & \text{if}~P~\text{is dominated by}~Q, \\
\infty, & \text{otherwise}.
\end{cases}
\end{align*}
The family of density power divergences \citep{basu_robust_1998}, are defined for a parameter $\beta \geq 0$ as,
\begin{align*}
  \DP_\beta (P \| Q) &=
    \int \left\{ q^{1+\beta} - \left( 1+\frac{1}{\beta} \right) q^{\beta} p + \frac{1}{\beta} p^{1+\beta} \right\} d \lambda,
      & \beta > 0
\end{align*}
and $\DP_0 = \KL$ is defined by taking the limit of $\beta \to 0$.
{When $\beta=1$,
the density power divergence is the $L_2$ distance
$\int (p - q)^2 \d\lambda$.}
Finally, the family of Integral Probability Metrics \citep[IPMs,][]{muller_integral_1997}, 
are defined as 
\begin{align*}
    \IPM_{\mathcal{F}}(P, Q) = \sup_{f \in \mathcal{F}} \left| \bbE_{P} (f) - \bbE_{Q} (f) \right|
\end{align*}
where $\mathcal{F}$ is a symmetric class (i.e., $f \in \mathcal{F} \implies - f \in \mathcal{F}$) of real-valued bounded measurable functions on the domain of $P$ and $Q$. 
Important special cases of IPMs include the Total Variation distance (TV, where $\mathcal{F}$ is the collection of functions with sup-norm at most 1), the Wasserstein-1 distance (where $\mathcal{F}$ is the collection of 1-Lipschitz functions) and 
the Maximum Mean Discrepancy (where $\mathcal{F}$ is the unit ball of a Reproducing Kernel Hilbert Space).

\section{{Model Misspecification and Unstable Projection Parameters}}
\label{sec:fail}

To provide some motivation and intuition for the methods we propose in this paper, it is useful to understand some of the failures of the universal inference framework 
when the statistical model is misspecified, and the target of inference is the KL projection. {All proofs accompanying the examples given in this section are in Supplementary Material~\ref{app:examples}.}

\subsection{Unbounded Likelihood-Ratios}
The behavior of likelihood-ratio based methods can be sensitive to the tail behavior of likelihood-ratios.
The following simple example illustrates that under model misspecification, universal inference can fail to cover the KL projection parameter. These pathologies
do not arise when the statistical model is correctly specified, and the challenges in this example arise due to an interplay between poorly behaved likelihood-ratios and model misspecification.
This example also serves to highlight the fact that the KL projection parameter can in some cases be an undesirable inferential target. We let $\Bern(p)$ denote the Bernoulli distribution with parameter~$p$.

\begin{example}\label{ex: instability_lr}
Suppose we observe 
$X_1,\ldots,X_n \sim \truedist := \Bern (\epsilon_n)$ for some non-negative $0 < \epsilon_n < (1-\alpha)/n$. 
We use the statistical model $\model = \left\{ \Bern(p) : p \in \left\{0, 1/2 \right\} \right\}$. 
Suppose we consider the pilot estimator to be the MLE,
\begin{align}\label{eq:pilot_MLE}
\pilot = \argmax_{p \in \model} \mathcal{L}_1(p).
\end{align}
Then, for all sufficiently large $n \geq n_0$ where $n_0$ only depends on $\alpha$ the split LRT confidence set in~\eqref{eqn:univ_set}, with an equal sized split into $\Data_0$ and $\Data_1$, 
fails to cover the KL projection $\projdist$ at the nominal level.
\end{example}
\noindent In this example the KL projection distribution $\projdist$ is $\Bern(1/2)$. 
For $\epsilon_n \ll 1/n$, with high probability the samples $X_1,\ldots,X_n$ are all $0$. Consequently, the MLE $\pilot$ with high-probability 
will be $\Bern(0)$. Furthermore, the split sample likelihood $\mathcal{L}_0$ will be much higher for $\Bern(0)$ than $\Bern(1/2)$, and consequently $\Bern(1/2)$ 
will not be included in the universal set.

In this example likelihood-ratios are unbounded and as a consequence the KL divergence is an unstable function of the model parameters, i.e. 
when $\epsilon_n = 0$, $\KL(\Bern(\epsilon_n) \| \Bern(0))$ is 0, but is $\infty$ for any $\epsilon_n > 0$. In such cases, 
the finite-sample (log)-likelihood-ratio is a poor estimate of the population KL divergence, and this poses significant challenges for 
finite-sample valid inference. From a practical standpoint, a more reasonable inferential target could be a different, stabler projection distribution 
(e.g., the Hellinger or TV projection distribution) and we address this in Sections~\ref{sec:hellinger} and \ref{sec:IPM}.

{
Huber's $\epsilon$-contamination model yields another canonical illustration of the non-robustness of the KL projection. The following example shows that universal inference can fail in this setting.
\begin{example}\label{ex:normal_contam}
Suppose we observe 
$X_1,\ldots,X_n \sim \truedist := (1-\epsilon_n) N(0,1) + \epsilon_n Q_{a_n}$ for some small positive $0 <\epsilon_n \leq c_1/n$, $|a_n| \geq c_2 n^2$ where $c_1, c_2 >0$ are small positive constants, and $Q_{a}$ is uniform on
$[a-\delta, a+\delta]$ for some $\delta>0$.
Consider a Gaussian location model $\model = \left\{ N(\theta, 1) : \theta \in \bbR \right\}$. 
Suppose we take the pilot estimator to be the MLE~\eqref{eq:pilot_MLE}.
Then, for all sufficiently large $n \geq n_0$ where $n_0$ only depends on $\alpha$ the split LRT confidence set in~\eqref{eqn:univ_set}, with an equal sized split into $\Data_0$ and $\Data_1$, 
fails to cover the KL projection $\projdist$ at the nominal level.
\end{example}
\noindent The KL projection parameter in this example is $\projtheta=\epsilon_n a_n$ which tends to $\infty$ as $n \to\infty$ despite $\epsilon_n \ll 1/n$. With high probability, all samples are nearly indistinguishable from those of $\Normal(0,1)$. The MLE $\thetahat_1=\overline{X}_{n_1}$ will be much closer to zero than $\projtheta$, resulting in a higher split sample likelihood $\calL_0$ for $\thetahat_1$ than $\projtheta$. Consequently $\projtheta$ will not be included in the universal set.
The fundamental issue again lies in the instability of the KL projection, which arises from the sensitivity of the expected log-likelihood ratio to its tail behavior.
In contrast to the KL projection, the Hellinger projection, for example, tends to the uncontaminated distribution $N(0,1)$ \citep{beran_minimum_1977} as $\epsilon_n \rightarrow 0$ and is a better inferential target. In Section~\ref{sec:Empirical_analysis}, we empirically demonstrate this issue and show that targeting a more stable projection alleviates the problem. 
}

\subsection{Failure Despite Bounded Likelihood-Ratios}

In the previous example it is clear that unbounded likelihood-ratios can result in pathologies that are challenging to address with finite-sample valid inference.
However, even when all likelihood-ratios in the model are well-behaved, universal inference can fail to cover the KL projection parameter. It is important to note that except under the stringent
condition that the underlying model is convex (see Section~6 of \citep{wasserman_universal_2020}), universal inference has no guaranteed coverage when the model is misspecified.

\begin{example}\label{ex: sLRT_fail}
Suppose we obtain $X_1,\ldots,X_n \sim \truedist := \Bern (0.5 + \epsilon_n)$ for some small, positive $0 <\epsilon_n \leq c/n$, where $c > 0$ is a small positive universal constant.
Our hypothesized model consists of two distributions,  $\model = \left\{ \Bern(p) : p \in \left\{1/4, 3/4 \right\} \right\}$. 
Suppose we take the pilot estimator to be the MLE~\eqref{eq:pilot_MLE}.
Then, for all sufficiently large $n$ (depending only on $\alpha$) the split LRT confidence set in~\eqref{eqn:univ_set}, with an equal sized split into $\Data_0$ and $\Data_1$, 
fails to cover the KL projection $\projdist$ at the nominal level.
\end{example}

The KL projection distribution is $\Bern(3/4)$. We show that the pilot estimate $\pilot$ with probability near $1/2$ will
be the distribution $\Bern(1/4)$, and further with probability near $1/2$ the KL projection $\Bern(3/4)$ will have a much smaller split sample likelihood than $\pilot$. As a direct consequence,
universal inference will fail to cover the projection distribution $\Bern(3/4)$. 

In contrast to the previous example, this example is much less pathological. All the relevant likelihood-ratios are bounded, and the log-likelihood is a consistent
estimate of the KL divergence. However, even in this relatively benign example universal inference fails.
We show in Section~\ref{sec:KL} that a simple modification to the universal inference procedure fixes this issue when the relevant likelihood-ratios are bounded, and ensures correct coverage.

In order to focus on the main issues, we have illustrated the failure of universal inference when the pilot estimator is the MLE. Indeed, part of the appeal of universal inference is that its 
coverage guarantees hold, in the well-specified case for any pilot estimate (including the MLE). 
Though we do not pursue this here, it is straightforward to extend these examples to show that both failures persist irrespective of 
how the pilot is chosen, i.e. the failures of universal inference that we highlight are driven by the second stage (of constructing the confidence set) and not by the first stage (of constructing a reasonable pilot estimate).

These examples set the stage for the methodological development of the rest of the paper. To address problems of the first type we recommend targeting a different
projection parameter (for instance, the TV or Hellinger projection, in Sections~\ref{sec:hellinger} and \ref{sec:IPM}), and to address problems of the second type we develop methods which guarantee coverage
of the KL projection parameter when the likelihood-ratios are uniformly upper bounded or more generally have finite $2 + \xi$ moments for some $\xi > 0$ (see Section~\ref{sec:KL}).

\section{Robust Universal Inference} \label{sec:rui}

In this section, we present a simple but powerful pair of general results which yield exact and approximate universal confidence sets. The workhorse of these results 
are tests of \emph{relative fit} which we first briefly introduce before showing how these tests can be inverted to derive robust confidence sets.

\subsection{Tests of Relative Fit}
\label{sec:Test_RF}

Suppose that we are given samples $X_1,\ldots, X_n \sim \truedist$, together with a pair of candidate distributions $(P_0, P_1) \in \model^2$, and a divergence measure $\rho$. 
With this setup in place, we now consider a family of tests $\phi_{P_0, P_1}$ to distinguish the hypotheses:
\begin{align}
\label{eq:exact_relative_test}
H_0: &~\rho(\truedist \| P_0) \leq \rho(\truedist \| P_1), \qquad \text{versus} \qquad
H_1: ~\rho(\truedist \| P_0) > \rho(\truedist \| P_1).
\end{align}

We refer to the tests $\phi_{P_0, P_1}$ as {\em tests of relative fit}.
In contrast to the classical setting, where we hypothesize 
that one of the distributions $(P_0, P_1)$ truly generated 
the samples, in the misspecified setup 
this assumption is no longer tenable.
Instead, we hypothesize that one of the distributions $(P_0, P_1)$ is \emph{closer} to the data generating distribution.
In general, the two hypotheses are no longer simple hypotheses and we need to take some care in designing the family 
of tests $\phi_{P_0, P_1}$. The design of tests of relative fit (and related variants) has a rich history and forms the basis for a class of tournament-based robust estimators 
\citep{birge_sur_1979, devroye_combinatorial_2001, lecam_convergence_1973, lugosi_risk_2019, yatracos_rates_1985}.

For divergences like the Total Variation and the Hellinger distance, designing exact tests of relative fit can require strong regularity conditions akin to those that would be required
to estimate these divergences. Surprisingly, in these cases, it is still possible to design \emph{approximate} tests of relative fit
under weak regularity conditions. More formally, 
suppose that for some $\nu \geq 1$, we can design a test for the following null hypothesis:
\begin{align}
\label{eq:approx_relative_test}
H_0: &~\nu\rho(\truedist \| P_0) \leq \rho(\truedist \| P_1).
\end{align}
We refer to tests for this hypothesis as \emph{approximate} tests of relative fit $\phi_{P_0, P_1,\nu}$. Under the null hypothesis, the distribution $P_0$ is closer than $P_1$ to $\truedist$ by a factor $\nu \geq 1$, which can 
ease the design of valid tests for this hypothesis.

Robust tests for null hypotheses of the form in~\eqref{eq:approx_relative_test} (for the Hellinger distance) were introduced by \citet{lecam_convergence_1973} and are discussed in detail in the work of~\citet{baraud_rho-estimators_2018}. In the context
of estimation these approximate tests yield what are known as non-sharp oracle inequalities. In the context 
of inference, as we explore further in Section~\ref{sec:RUCset}, inverting approximate relative fit tests will yield weaker guarantees.  
In Section~\ref{sec:tests} we consider the design of tests of relative fit in concrete settings, but now proceed to study the implications of designing such tests 
for the construction of robust confidence sets.

\subsection{Exact Robust Universal Confidence Sets} 
\label{sec:RUCset}
We now propose to construct a confidence set by inverting a family of tests of relative fit. This is similar in spirit to 
the procedure described in Section~\ref{sec:reinterpret_ui}.

Suppose, for every $\truedist\in\trueclass$, the family of tests of relative fit $\phi_{P_0, P_1}$ is valid, i.e. it controls the Type I error:
\begin{align}
\label{eqn:typeone_robust}
\bbE_{\truedist} \left[\phi_{P_0, P_1}(X_1,\ldots,X_n) \right]\leq \alpha, ~~~\forall~(P_0, P_1) \in \calS_0
\end{align}
where $\calS_0 = \{ (P_0, P_1) \in \model^2: \rho(\truedist \| P_0) \leq \rho(\truedist \| P_1)\}.$ 
Then, for any fixed $P_1 \in \model$, the confidence set we construct is the set of candidates $P_0$ which we fail to reject:
\begin{align*}
C_{\alpha,n} \equiv C_{\alpha}(X_1,\ldots,X_n) := \left\{ P_0 \in \model: \phi_{P_0, P_1} (X_1,\ldots,X_n) = 0  \right\}. 
\end{align*}
The following result shows that irrespective of the choice of $P_1$ the above construction yields a valid confidence set for the projection distribution:
\begin{proposition}\label{prop:honest_exact}
For any fixed $P_1 \in \model$, $C_{\alpha,n}$ is a uniformly valid $(1-\alpha)$ honest confidence set for the projection $\projdist$.
\end{proposition}
\begin{proof}
For any $\truedist \in \trueclass$,
$\bbP_{\truedist}(\projdist \notin C_{\alpha,n} ) 
    =   \bbP_{\truedist}( \phi_{\projdist, P_1} = 1 )
    =   \bbE_{\truedist}( \phi_{\projdist, P_1} ) \le \alpha$
using~\eqref{eqn:typeone_robust} since $(\projdist, P_1) \in \calS_0$ for any choice of $P_1 \in \model$.
\end{proof}

As in the well-specified case discussed earlier, this general result does not provide any guidance on how to choose $P_1$. 
We follow the idea of universal inference and first 
construct an accurate estimate $\pilot$ of $\projdist$ from a 
separate sample $\Data_1$ and then construct 
the family of \emph{split} tests of relative fit $\phi_{P_0, \pilot}$ 
from the remaining samples $\Data_0$. We call the 
resulting confidence set the \emph{exact Robust Universal Confidence set}.
\begin{align}
\label{eq:exact_rob_uni_set}
  C_{\alpha,n} \equiv C_{\alpha} (X_1,\dots, X_n) := \{P_0\in \calP: \phi_{P_0, \pilot} (\Data_0) = 0\}.
\end{align}

\begin{theorem}
\label{thm:honest_exact}
 Let $\pilot \in \model$ be any estimate of $\projdist$ based on $\Data_1$. Then, the exact robust universal confidence set $C_{\alpha,n}$ is a uniformly valid confidence set for $\projdist$, meaning that
\begin{align*}
  \inf_{\truedist \in \trueclass}  \bbP_{\truedist} (\projdist \in  C_{\alpha, n}) \ge 1 - \alpha. \qedhere
\end{align*}
\end{theorem}

\begin{proof}
The proof is straightforward noticing that conditional on $\Data_1$, the claim reduces to the claim of Proposition~\ref{prop:honest_exact}. Concretely, for any $\truedist \in \trueclass$,
  \begin{align*}
    \bbP_{\truedist} (\projdist \notin C_{\alpha, n}) 
      =\bbE_{\truedist} (\phi_{\projdist, \pilot})  
      =&  \bbE_{\Data_1} \left[ \bbE_{\Data_0}(\phi_{\projdist, \pilot} (\Data_0) \;|\; \Data_1) \right] 
      \le \bbE_{\Data_1} (\alpha) = \alpha.  \qedhere
  \end{align*}
\end{proof}

The robust confidence set will often contain both the pilot estimate $\pilot$ as well as the projection distribution $\projdist$ 
(see Proposition~\ref{prop:min_diam_exact} in Supplementary Material~\ref{app:secfour} for a formal statement). This is similar to the classical universal inference procedure which in the well-specified case will often contain both the 
pilot estimate and the true sampling distribution. In universal inference this suggests that in order to obtain small confidence sets, 
we should aim to design $\pilot$ to be a good estimate of the true sampling distribution $\truedist$. On
the other hand in the misspecified case, this suggests that we should design $\pilot$ to be a good estimate of the projection $\projdist$. Specifically, our pilot estimate should 
be tailored to the divergence measure $\rho$. We investigate the choice of $\pilot$ and its effect on the size of the resulting confidence set further in Section~\ref{sec:size_set}.

\subsection{Approximate Robust Universal Confidence Sets}
In some cases, e.g., for the Hellinger distance and the TV distance, designing exact robust tests will require some (potentially strong) regularity conditions.
However, in these cases one can design approximate tests of relative fit straightforwardly.

Suppose, for any $\truedist\in \trueclass$, the family of approximate tests of relative fit $\phi_{P_0, P_1, \nu}$ which controls the Type~\rom{1} error satisfies \eqref{eqn:typeone_robust} with $\calS_0 = \{ (P_0, P_1) \in \model^2 : \nu \rho(\truedist \| P_0) \leq \rho(\truedist \| P_1)\}$ for some $\nu \geq 1$. We will additionally make the mild assumption that our tests of relative fit do not reject (with probability at least $1-\alpha$) when comparing the relative fit of a distribution to itself, i.e.: 
\begin{align}
\label{ass:powerless}
\sup_{\truedist\in\trueclass}\bbE_{\truedist} [\phi_{P,P,\nu}] \le \alpha~\text{for any fixed}~ P\in\model.  
\end{align}
This condition will be true for all the tests we introduce in Section~\ref{sec:tests}.

Let $\pilot$ be any estimate of $\projdist$ from $\Data_1$.
Then, the approximate robust universal confidence set, akin to \eqref{eq:exact_rob_uni_set}, is obtained by inverting the family of valid split tests $\phi_{P_0, \pilot, \nu}$ constructed from the remaining samples $\Data_0$:
\begin{align}
\label{eqn:robust_approx_set}
C_{\nu,\alpha,n} \equiv C_{\nu,\alpha}(X_1,\ldots,X_n) := \left\{ P_0 \in \model: \phi_{P_0, \pilot, \nu} (\Data_0) = 0  \right\}. 
\end{align}

This confidence set may not cover the projection distribution $\projdist$. We will relax our goal to instead be to 
cover an approximate projection distribution.
More formally, we relax the target of inference to be the \emph{$\nu$-approximate projection set} $\projdist_\nu$ defined as
\begin{align}
\label{eqn:pnu}
   \projdist_\nu = \{P\in \model: \rho(\truedist \| P) \le \nu \rho(\truedist \| \projdist) \}.
\end{align}
If a set $C$ is a $\nu$-approximate confidence set, we define its coverage by
\begin{align*}
  \bbP_{\truedist}\left(Q\in C\ {\rm for\ some\ }Q \in \projdist_\nu\right) =
  \bbP_{\truedist}\left(\projdist_\nu \cap C \neq \emptyset\right). 
\end{align*}
Figure~\ref{fig:schematic_approx_set} shows a schematic diagram to illustrate the notion of approximate coverage. When $\nu = 1$, i.e. we invert an exact test, we guarantee that with probability at least $1 - \alpha$, the set 
$C_{\nu,\alpha,n}$ contains $\projdist$. On the other hand, when $\nu > 1$ we only guarantee that the intersection of $C_{\nu,\alpha,n}$ with the collection of $\nu$-approximate projections (in cyan) is non-empty. 
\begin{figure}[!htb]
  \fontsize{12}{12}\selectfont
  \centering
  \adjustbox{trim=2cm 8cm 6cm 6cm} {
  \includesvg[width=0.8\textwidth,inkscapearea=page]{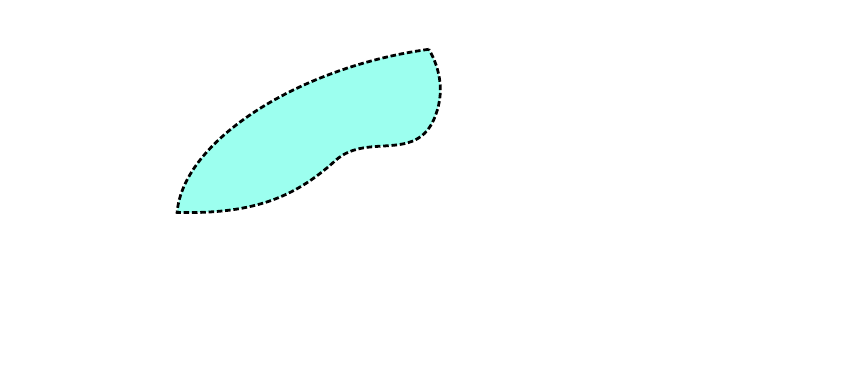}
  }
  \caption{An exact confidence set $C_{\nu,\alpha,n}$ (in blue) targets the distribution $\projdist$ (in red star), and an approximate confidence set targets the set $\projdist_\nu$ (in cyan). 
   }
  \label{fig:schematic_approx_set}
\end{figure}

The set $\projdist_\nu$ is a collection of distributions that are as close to $\truedist$ as $\projdist$ (up to a factor $\nu$). 
The approximate confidence set guarantee
is most meaningful when $\nu$ is close to 1, or when the model misspecification is not too extreme, i.e. $\rho(\truedist \| \projdist)$ is small. 
{
The approximate confidence set requires some care to interpret. In practice we are sometimes interested in constructing a confidence set for a \emph{functional} $\psi: \mathcal{P} \mapsto \mathbb{R}^d$ of the projection distribution (say, its mean). A $\nu$-approximate confidence set for $\widetilde{P}$ does not in general yield an approximate confidence set for $\psi(\widetilde{P})$ except under strong regularity conditions which allow us to relate the divergence $\rho$ between two distributions to a distance between their induced functional values (for instance, differentiability in quadratic mean \citep{vaart_asymptotic_1998}). }

\begin{theorem}
\label{thm:honest_approx}
Let $\pilot \in \model$ be any estimate of $\projdist$ based on $\Data_1$. Suppose that our approximate relative fit tests are valid, and satisfy the condition in~\eqref{ass:powerless}. 
Then, the approximate robust universal confidence set $C_{\nu,\alpha,n}$ is a uniformly valid $\nu$-approximate confidence set for $\projdist$:
\begin{align*}
  \inf_{\truedist \in \trueclass}  \bbP_{\truedist} (\projdist_\nu \cap C_{\nu,\alpha, n} \ne \emptyset) \ge 1 - \alpha.
\end{align*}
\end{theorem}
\begin{proof}
Fix any $\truedist \in \trueclass$. Let the event $E = \{\pilot \in \projdist_\nu\}$. On the event $E$,~\eqref{ass:powerless} implies 
\begin{align*}
  \bbP_{\truedist} (\pilot \notin C_{\nu,\alpha,n} ~|~ E) = \bbE_{\Data_1} (\bbE_{\Data_0} (\phi_{\pilot, \pilot,\nu} (\Data_0) ~|~ \Data_1, E) ~|~ E) \le \alpha.
\end{align*}
On the complement of $E$, i.e., $\pilot \notin \projdist_\nu$, $\calS_0$ contains $(\projdist, \pilot)$. Thus, an analogous argument to that in the proof of Theorem~\ref{thm:honest_exact} can be used.
Combining the two results, we obtain that, for all $\truedist \in \trueclass$,
\begin{align*}
  \bbP_{\truedist} (\projdist_\nu \cap C_{\nu,\alpha, n} = \emptyset) 
   &\leq  \bbP_{\truedist} (\pilot \notin C_{\nu,\alpha, n} ~|~ E) \bbP(E) + \bbP_{\truedist} (\projdist \notin C_{\nu,\alpha, n} ~|~ E^\complement) \bbP(E^\complement)
   \leq  \alpha.   \qedhere
\end{align*}
\end{proof}

As in the construction of the exact robust universal confidence set, one should aim to choose the pilot estimate $\pilot$ as close as possible to $\projdist$. 
In the exact setting, the choice of the pilot estimate does not affect the validity of the resulting set and only affects its size. However, in 
constructing an approximate robust universal set, 
if we can ensure the pilot is accurate, then our approximate validity guarantees improve. Concretely, for some sequence $\kappa_n$ we define:
\begin{align}
\label{eqn:smaller}
\calS(\kappa_{n}) := \{P\in\model : \rho(\truedist \| P) \le \rho(\truedist \| \projdist) + \kappa_{n} \}.
\end{align}
If we can ensure that the pilot estimate is contained in $\calS(\kappa_{n})$ with probability at least $1 - \beta$ for some sequence $\kappa_n$, 
then the constructed confidence set $C_{\nu, \alpha,n}$ will intersect $\calS(\kappa_{n})$ with high probability. For instance, if $\kappa_n \rightarrow 0$ 
as $n$ grows, then rather than simply intersecting the set of approximate projections $\projdist_\nu$, we can now show that $C_{\nu,\alpha,n}$ intersects a shrinking neighborhood
around $\projdist$. More formally we have the following result (we omit its proof since it follows the same arguments as in Theorem~\ref{thm:honest_approx}):

\begin{proposition} \label{cor:small_target}
Let $\calS(\kappa_{n_1})$ be defined as in~\eqref{eqn:smaller}, and suppose that our pilot is accurate, i.e. we can ensure that with probability at least $1 - \beta$, $\pilot \in \calS(\kappa_{n_1})$. 
Suppose further that our approximate relative fit tests are valid, and satisfy the condition in~\eqref{ass:powerless}. Then:
\begin{align*}
  \inf_{\truedist \in \trueclass}  \bbP_{\truedist} \left(\calS(\kappa_{n_1}) \cap C_{\nu,\alpha, n} \ne \emptyset \right) \ge 1 - \alpha - \beta.
\end{align*}
\end{proposition}
{
The critical shortcoming of a $\nu$-approximate confidence set is that its interpretation is challenging as it does not guarantee coverage of the projection distribution $\widetilde{P}$. In some cases, for example, when the model misspecification is a consequence of data contamination, the statistician might be able to provide a reasonable upper bound on $\rho(P^* \| \widetilde{P})$. As an example when the distribution $P^*$ arises from an $\epsilon$-contamination of a distribution $P \in \mathcal{P}$ then we have the upper bound $\text{TV}(P^*, \widetilde{P}) \leq \epsilon$. In this case one can construct a valid (classical) confidence set by enlarging the $\nu$-approximate confidence set. For a given upper bound $\epsilon$ on $\rho(P^* \| \widetilde{P})$ define:
\begin{align}
\label{eqn:bonafide}
    C_{\epsilon,\alpha,n} := \left\{ P \in \mathcal{P}: \inf_{Q \in C_{\nu,\alpha,n}} \rho(P\| Q) \leq (\nu+1) \epsilon \right\}.
\end{align}
Then we have the following Corollary:
\begin{corollary}
    Suppose that $\rho$ is symmetric and satisfies the triangle inequality, then under the conditions of Theorem~\ref{thm:honest_approx} for $C_{\epsilon,\alpha,n}$ in~\eqref{eqn:bonafide}
\begin{align*}
  \inf_{\truedist \in \trueclass}  \bbP_{\truedist} \left(\widetilde{P} \in C_{\epsilon,\alpha, n} \right) \ge 1 - \alpha.
\end{align*}
\end{corollary}}
In this section, we have shown that inverting exact or approximate tests of relative fit yields robust exact or approximate confidence sets despite model-misspecification. We now turn our attention
to the design and analysis of these tests. 

\section{Designing Tests of Relative Fit}
\label{sec:tests}
In this section,
we design valid exact tests of relative fit in KL and the density power divergences, and design valid approximate tests
for the Hellinger and IPM-based divergences. 

\subsection{Kullback-Leibler  Divergence}
\label{sec:KL}

To design an exact test of relative fit for the KL divergence we make a simple observation
that there is a natural plug-in estimator of the difference in KL divergences. We can rewrite
the difference in KL divergences as:
\begin{align*}
  \KL(\truedist \| P) - \KL(\truedist \| \pilot) = 
\int \log \frac{\widehat{p}_1}{p} \d \truedist
\end{align*}
where $p$ and $\widehat{p}_1$ are the density of $P$ and $\pilot$ with respect to a common dominating measure. When
we obtain samples from $\truedist$ this suggests the following
log split likelihood ratio test: 
\begin{align}\label{eq: def_RIFT}
  \phi_{P} = \I \left[ \frac{1}{n_0} \sum_{i\in \calI_0} T_i (P,\pilot) > t_\alpha (P, \pilot) \right],
  \qquad
  {T_i(P, \pilot) \equiv T(X_i; P, \pilot)
  =  \log \frac{\widehat{p}_1 (X_i)}{p (X_i)}},
\end{align}
where $\calI_0$ is an index set of $\Data_0$ and $t_\alpha (P, \pilot)$ is chosen to ensure validity. This test 
was called the relative information fit test (RIFT) and was used to select between two candidate estimates by \citet{chakravarti_gaussian_2019}. In our paper, 
we invert the same test in order to construct a robust universal confidence set.

When the variance of $T_i(P, \pilot)$ (conditional on $\mathcal{D}_1$) 
is finite and positive, we can derive the asymptotic 
distribution (conditional on $\mathcal{D}_1$) 
of the log split likelihood ratio via the CLT.
Let $\overline{T}_{n_0} (P,\pilot) = \sum_{i\in \calI_0} T_i(P, \pilot) / n_0$.
Conditional on $\Data_1$ and assuming that the variance $\bbV_{\truedist} [T(P_0, P_1)] < \infty$, for any~$(P_0,P_1) \in\model^2$,
\begin{align*}
  \sqrt{n_0} \left( \overline{T}_{n_0} (P,\pilot) - \bbE_{\truedist} T (P, \pilot) \right) \overset{\truedist}{\rightsquigarrow}\Normal \left(0, s_P^2 \right)
  \qquad \text{ as } n_0 \to \infty
\end{align*}
where $s_P^2 \equiv s_P^2 (\Data_1) = \bbE_{\truedist} [T_1^2] - \bbE_{\truedist} T_1^2$ can be estimated by $\hat{s}_P^2 = \frac{1}{n_0} \sum_{i\in\calI_0} (T_i(P, \pilot) - \overline{T}_{n_0})^2$,
and $\rightsquigarrow$ denotes convergence in distribution (conditional on $\mathcal{D}_1$). 

When assessing distributions $P$ that are very similar to the pilot $\pilot$, it might be the case that $s_P^2$ is vanishingly small. Consequently, it is possible that $\widehat{s}_P/s_P$ does
not converge in probability to 1, and the CLT with estimated variance $\widehat{s}_P^2$ need not hold. Following \citet{chakravarti_gaussian_2019} 
we modify each $T_i(P,\pilot)$ by adding a small amount of independent Gaussian noise, i.e. we replace each $T_i(P, \pilot)$ above by $T_i(P, \pilot) + \delta Z_i$ where $Z_1,\ldots,Z_{n_0} \sim N(0,1)$,
for some small positive constant $\delta > 0$. 
{Large values of $\delta$ will lead to slightly larger intervals, while smaller $\delta$ will require larger sample sizes
for the asymptotics to be valid. Optimizing this choice is an open problem, but this modification for small values of $\delta$ (e.g., taking $\delta=0.01$) has no practical effect and simply eases the theoretical analysis.} 
We denote the resulting
statistic by $\overline{T}_{n_0,\delta}(P, \pilot)$ and the corresponding empirical standard deviation by $\widehat{s}_{P,\delta}$. 

Then, we define the KL \emph{Relative Divergence Fit} (\Redi) set as 
\begin{align}\label{eq: ConfSet_Rift}
  \Chat_{\KL\Redi, n} \equiv \Chat_{\alpha, n} (\Data) = \left\{P\in \model : \overline{T}_{n_0, \delta}(P, \pilot) \le  \frac{z_{\alpha}  \hat{s}_{P,\delta}}{\sqrt{n_0}} \right\}
\end{align}
where $z_\alpha$ is a $1-\alpha$ quantile of standard normal distribution. The following result provides asymptotic and non-asymptotic guarantees for the set $\Chat_{\KL\Redi, n}$.
\begin{theorem}
\label{thm:Rift_validity}
Suppose that $\mathcal{Q}$ is such that for some $0 < \xi \leq 1$ the $2+\xi$ moments $\bbE_{\truedist} |T(X; P_0, P_1) - \bbE_{\truedist}T(X; P_0, P_1)|^{2+\xi} \leq M < \infty$ are finite, for any~$(P_0,P_1) \in\model^2$, then 
  \begin{align}\label{eq:Rift_coverage_asymp}
     \inf_{\truedist \in \trueclass} \bbP_{\truedist} (\projdist \in \Chat_{\KL\Redi, n}) \geq 1 - \alpha - C n^{-\xi/2},
  \end{align}
where $C < C' (1 + M) /\delta^{(2+\xi)}$ for a universal constant $C'$.
\end{theorem}

We give a formal proof in Supplementary Material~\ref{sec:proof_section_test}. The claim follows as a consequence of the Berry-Esseen bound for the studentized statistic~\citep{bentkus_berry-esseen_1996,petrov_sums_1975}. Some care is required 
to handle the degeneracy (discussed above) {and the randomness in the pilot estimate $\pilot$.}

We can now revisit the failures of universal inference discussed in Section~\ref{sec:fail}. Recall that Example~\ref{ex: instability_lr} illustrates the instability of the KL projection because likelihood ratios may not be bounded. 
The KL \Redi set does not resolve this weakness since the KL \Redi set uses the same split likelihood ratio statistic as for the universal confidence set \citep{wasserman_universal_2020} and its $2 + \xi$ 
moment is not uniformly bounded in Example~\ref{ex: instability_lr}. However, the KL \Redi set does resolve the failure highlighted in Example~\ref{ex: sLRT_fail}.

\begin{corollary}
\label{prop:RIFT_sucess}
Assume the same model as in 
Example~\ref{ex: sLRT_fail}. Suppose we take the pilot estimator to 
be the MLE. The KL \Redi set~\eqref{eq: ConfSet_Rift}, with an equal sized split into $\Data_0$ and $\Data_1$, covers the KL projection $\projdist$ at the nominal level asymptotically.
\end{corollary}
This result follows directly from Theorem~\ref{thm:Rift_validity}, since in this example all of the relevant log likelihood ratios are uniformly upper bounded. 
It is worth noting that both the standard universal set, and the set $\Chat_{\KL\Redi, n}$ are based on essentially the same split likelihood ratio statistic, 
and it is perhaps surprising that the standard universal set fails but $\Chat_{\KL\Redi, n}$ succeeds in guaranteeing coverage. 
Despite being based on the same statistic, the two sets use very different thresholds. It is easy to see that one can rewrite the split
LRT confidence set in universal inference \citep{wasserman_universal_2020} as:
\begin{align*}
  \Chat_{sLRT}= \left\{P\in \model : \overline{T}_{n_0} (P,\pilot) \le \frac{\log (1/\alpha)}{n_0}  \right\}.
\end{align*}
The threshold used in (non-robust) universal inference decays at the fast rate of order $O(1/n_0)$ compared to that of the robust universal confidence set $\Chat_{\KL\Redi, n}$ 
whose threshold decays at the rate $O(1/\sqrt{n_0})$. When the model is misspecified the (non-robust) universal set shrinks too rapidly leading to the failure highlighted in Example~\ref{ex: sLRT_fail}.

The confidence set $\Chat_{\KL\Redi, n}$ is constructed by approximating the distribution of the test statistic in~\eqref{eq: def_RIFT}. 
When likelihood ratios are uniformly upper bounded it is straightforward to construct finite-sample valid sets via an exponential tail bound. 
For example, the finite-sample exact robust universal confidence set based on the Hoeffding bound is:
\begin{align}
\label{eqn:hoef}
  \Chat_{\textsc{HF},B,n} = \left\{P\in \model : \overline{T}_{n_0} (P,\pilot) \le B\sqrt{\frac{ \log \left(1 / \alpha\right)}{2n_0} } \right\},
\end{align}
where $B$ is such that $|T_i (P_0, P_1) - \mathbb{E}_{\truedist} T(P_0,P_1)| \le B$ for all $(P_0,P_1)\in\model^2$. In this case we assume that the 
upper bound $B$ is known. 

{
One practical benefit of the set $\Chat_{\textsc{HF},B,n}$ over the set $\Chat_{\KL\Redi, n}$ is that the former set is based on a test whose threshold is independent of the distribution $P$ under consideration. This fact has an important implication for hypothesis testing.
Concretely, given two collections $\mathcal{P}_0$ and $\mathcal{P}_1$ of distributions, and samples $X_1,\ldots,X_n \sim P^*$ suppose that we wish to test the composite null hypothesis:
\begin{align}
\label{eqn:null}
    H_0: \inf_{P \in \mathcal{P}_0} \KL(P^* \| P) \leq \inf_{Q \in \mathcal{P}_1} \KL(P^* \| Q).
\end{align}
This hypothesis generalizes a goodness-of-fit test to settings with model misspecification, the primary focus of our paper.
Consider the following procedure: 
\begin{enumerate}
    \item Let $\widehat{P}_1 \in \mathcal{P}_1$ be \emph{any} estimate constructed using $\mathcal{D}_1$.
    \item Let $\widehat{P}_0 \in \mathcal{P}_0$ denote the MLE (over $\mathcal{P}_0$) computed on $\mathcal{D}_0$.
    \item Reject the null if:
    \begin{align*}
        \overline{T}_{n_0}(\widehat{P}_0, \widehat{P}_1) > B \sqrt{\frac{ \log(1/\alpha)}{2n_0}}. 
    \end{align*}
\end{enumerate}
The following proposition shows that this test controls the Type I error at $\alpha$. 
\begin{proposition}
    Suppose that all likelihood ratios between distributions in $\mathcal{P}_0$ and $\mathcal{P}_1$ are uniformly upper bounded by $B$. Then the test described above controls the Type I error for testing $H_0$ in~\eqref{eqn:null} at $\alpha.$
\end{proposition}
This test serves as an analogue of the split-LRT of \citet{wasserman_universal_2020}, with the key distinction that it allows for model misspecification. Importantly, this test inherits one of the main benefits of the split-LRT --- if we can maximize the likelihood under $\mathcal{P}_0$ then it is possible to construct a sensible test for the null hypothesis~\eqref{eqn:null}.
}

One can generalize this construction in various ways. 
When the statistic is assumed to only have finite
variance one can use Chebyshev's inequality to construct a finite-sample valid set. 
When in addition to boundedness the statistic might have small variance
one can use empirical Bernstein-type inequalities to construct finite-sample valid confidence sets. 
We explore these further and compare the empirical performance of  $\Chat_{\KL\Redi, n}$ and these finite-sample valid sets in Supplementary Material~\ref{sec:finite_set}. 

\subsection{Density Power (DP) Divergences}
\label{sec:dpd}

We can construct an exact test of relative fit for the family
of DP divergences following the same strategy as in KL case. 
Let $  \overline{T}_{n_0}(P, \pilot)$ be a plug-in estimator of the difference in DP divergences:
\begin{align}
\label{eqn:ha}
  \overline{T}_{n_0}(P, \pilot) 
    &= \int \left\{ p^{1+\beta} - \widehat{p}_1^{1+\beta}  \right\} \d \lambda
     - \left( 1+\frac{1}{\beta} \right) \frac{1}{n_0} \sum_{i\in \calI_0} \left[  p^{\beta} - \widehat{p}_1^{\beta} \right] (X_i).
\end{align}
The split statistics $T_i(P, \pilot)$ encode the difference in average $\beta$-powered densities (penalized with $L_{1+\beta}$ norm) rather than the log-likelihood ratio evaluated on the sample $\Data_0$ when $\beta > 0$.
Then, conditional on $\mathcal{D}_1$, 
$\bbE_{\truedist} T(P,\pilot) = 
\DP_\beta (\truedist \| P) -  \DP_\beta (\truedist \| \pilot)$. 
We define the DP \Redi set $\Chat_{\DP\Redi,n}$ 
exactly as in~\eqref{eq: ConfSet_Rift},
and observe that the analogue of 
Theorem~\ref{thm:Rift_validity} holds (with an identical proof) for  $\Chat_{\DP\Redi,n}$.

Recall that KL \Redi set was unable to resolve the instability problem
in Example~\ref{ex: instability_lr}. This is because the likelihood ratios in this model can blow up. On the other hand the DP set relies on the statistics in~\eqref{eqn:ha}, which are bounded for any $\beta > 0$, provided the relevant
densities are well-defined.
Formally, we have the following result:
\begin{proposition}
\label{prop:DP_ex1}
   Suppose we have the same model as in 
   Example~\ref{ex: instability_lr}.
For sufficiently large $n$, for any pilot estimator $\pilot$, the DP \Redi set~$\Chat_{\DP\Redi,B,n}$ defined as in~\eqref{eqn:hoef} with $B=1 + 1/\beta$, with an equal sized split into $\Data_0$ and $\Data_1$, covers the DP projection $\projdist$ at the nominal level.
\end{proposition}

A formal proof can be found in Supplementary Material~\ref{sec:proof_section_test}. The key observation is that the DP projection is $\Bern(0)$ for a sufficiently large sample size for any fixed $\beta > 0$. The DP projection in this example is more stable than the KL projection $\Bern(1/2)$, considering that $\epsilon_n$ is much closer to $0$ than $1/2$.
Consequently, we show that the DP \Redi set will cover the target of inference $\Bern(0)$ with high probability. We emphasize that the MLE is also $\Bern(0)$ with high probability, yet both universal split LRT and KL \Redi set based on the MLE fail to cover the KL projection due to the instability of the population projection distribution.

\subsection{Hellinger Distance}
\label{sec:hellinger}

The Hellinger distance (or the difference in Hellinger distances) does not lend itself to a natural plug-in estimator. The usual method of estimating the Hellinger distance proceeds instead via some type of non-parametric density estimation, which in turn 
requires additional smoothness assumptions. Since our goal in this paper is to design assumption-light methods, we instead relax the target of inference. This in turn opens the door for designing approximate tests of relative fit.

Our strategy will be to modify the $\rho$-estimator\footnote{The name ``$\rho$-estimator'' comes from the standard symbol used for the Hellinger affinity.} \citep{baraud_estimator_2011,baraud_new_2017,baraud_rho-estimators_2018} 
which is a density estimator tailored to the Hellinger loss.
Define the \emph{split $\rho$-test statistic}
 \begin{align*}
  \overline{T}_{n_0} (P, \pilot) := \Delta (P, \pilot) + \frac{1}{n_0} \sum_{i\in\calI_0} \psi \left( \sqrt{\frac{\widehat{p}_1}{p}} (X_i) \right),
\end{align*}
where $\Delta (P_0, \pilot) =  
\left[\HEL^2(P_0, \overline{P}) - \HEL^2(\pilot, \overline{P}) \right] / \sqrt{2}$, $\overline{P} = (P_0 + \pilot) / 2$ and 
$\psi: [0,\infty] \mapsto [-1,1]$ is a non-decreasing Lipschitz 
function satisfying $\psi (x) = - \psi (1/x)$. 
The choice of $\psi$ we adopt throughout this paper, is to 
take $\psi(u) = (u-1)/\sqrt{1+u^2}$ which
comes from work on the $\rho$-estimator \citep{baraud_estimator_2011,baraud_new_2017,baraud_rho-estimators_2018}.
The function
$\psi$ is a bounded transformation of the likelihood ratio, and due to this boundedness the split $\rho$-test statistic is tightly concentrated around its expectation.
The following proposition, which follows directly from Proposition 11 of \citet{baraud_tests_2020}, characterizes the expectation of the split $\rho$-statistic. 
\begin{proposition}
\label{prop:hel_ub}
For any $P^*, P_0, P_1$, 
\begin{align*}
  \left(2 + \sqrt{2}\right) \bbE_{\truedist} \overline{T}_{n_0} (P_0,P_1) \le \left(3 + 2\sqrt{2}\right) \HEL^2 (\truedist, P_0) - \HEL^2 (\truedist, P_1).
\end{align*}  
\end{proposition}

This proposition ensures that $\bbE_{\truedist} \overline{T}_{n_0}(P_0, P_1)$ is negative for any $\truedist \in \trueclass$ when the null hypothesis $H_0 : (3+2\sqrt{2}) \HEL^2 (\truedist, P_0) \le \HEL^2 (\truedist, P_1)$ is true. 
Thus $\overline{T}_{n_0}(P_0, \pilot)$ could be a useful statistic for designing an approximate test of relative fit in the Hellinger distance with $\nu = \sqrt{3+2\sqrt{2}}.$

We define the Hellinger Relative Distance fit (\HEL \Redi) set $\Chat_{\HEL\Redi,n}$ exactly analogous to the 
KL \Redi set~\eqref{eq: ConfSet_Rift} (obtained from a $\delta$-corrupted version of the statistics $\overline{T}_{n_0}(P, \pilot)$).  
The following result follows by combining Theorems~\ref{thm:honest_approx} and \ref{thm:Rift_validity}, and noticing that the split statistic is upper bounded.

\begin{corollary}
\label{thm:HELRedi_validity}
Let $\nu = \sqrt{3 + 2\sqrt{2}}$. For any $\mathcal{Q}$,  
  \begin{align}
      \inf_{\truedist \in \trueclass} \bbP_{\truedist} (\projdist_\nu \cap \Chat_{\HEL\Redi, n} \ne \emptyset) \ge 1 - \alpha - C/\sqrt{n},
  \end{align}
   where $C < C'/\delta^3$ for a universal constant $C'$.
\end{corollary}

We are now in a position to revisit  Example~\ref{ex: instability_lr}. In Proposition~\ref{prop:DP_ex1}, we showed that changing the target of inference to DP projection could address the failure of universal inference. 
In a similar vein, targeting the Hellinger projection resolves the failure, but interpreting the resulting guarantee requires some nuance as \HEL\Redi set may not cover the exact Hellinger projection, and is only guaranteed to cover 
a $\nu$-approximate projection. 
In the case of Example~\ref{ex: instability_lr}, it will turn out for sufficiently small values $\epsilon$ the $\nu$-approximate Hellinger projection set is a singleton (and equal to the exact Hellinger projection). As highlighted earlier, when 
the amount of model-misspecification is not too large the distinction between the $\nu$-approximate projection set and the exact projection can be small.

\begin{proposition}
\label{prop:HEL_ex1}
    Assume the same model as in Example~\ref{ex: instability_lr}. Suppose we take the pilot estimator to be the Minimum Hellinger Distance estimator \citep[MHDE,][]{beran_minimum_1977},
  $\pilot = \argmin_{P \in \model} \HEL (\bbP_{n_1} \| P)$.
For sufficiently large $n (> 20)$, the Hellinger \Redi set~$\Chat_{\HEL\Redi,n}$, with an equal sized split into $\Data_0$ and $\Data_1$, covers the  Hellinger projection $\projdist \equiv \Bern(0)$ at the nominal level asymptotically.
\end{proposition}

A formal proof is provided in Supplementary Material~\ref{sec:H_ex1_proof}. In this example,
the $\nu$-approximate Hellinger projection is exactly the Hellinger projection when $\epsilon \le 0.05$, and is the entire model $\model$, otherwise. This means that for larger values of $\epsilon$, approximate validity is trivial, yet vacuous, as the target of inference can be any distribution in $\model$. This highlights the downside of targeting the $\nu$-approximate projection set: when the model-misspecification is severe the resulting guarantees might be vacuous.

\subsection{Integral Probability Metrics (IPMs)}
\label{sec:IPM}

Our proposal for a $\nu$-approximate test of relative fit for IPMs is inspired by the work of \citet{yatracos_rates_1985} and \citet{devroye_combinatorial_2001}, where 
a similar idea was used to design robust density estimates. Recall the definition of the IPM,
  $\IPM_{\calF}(P_0, P_1) = \sup_{f \in \calF} \left( \bbE_{P_0} (f) - \bbE_{P_1} (f) \right)$.
 Associated with any pair of distributions is a so-called witness function $f^*_{(P_0,P_1)} = {\arg\!\sup}_{f \in \calF} ( \bbE_{P_0} (f) - \bbE_{P_1} (f) )$, which 
 witnesses the largest mean discrepancy between the two distributions. 
 The split test statistic is then defined by:
\begin{align}\label{eq: IPM_split_stat}
  \overline{T}_{n_0} (P, \pilot) =  \int f^*_{(P, \pilot)} \frac{\d P + \d \pilot}{2} - \frac{1}{n_0} \sum_{i \in \calI_0} f^*_{(P, \pilot)} (X_i).
\end{align}
 The usefulness of this statistic is highlighted by the following characterization of its expectation. 
\begin{proposition}\label{prop: IPM_gamma} 
For any $P^*, P_0, P_1$, 
\begin{align*}
  2 \bbE_{\truedist} T (P_0,P_1) \le 3 \, \IPM (\truedist, P_0) -  \IPM (\truedist, P_1).
\end{align*}  
\end{proposition}
See Supplementary Material~\ref{sec:proof_section_test} for a formal proof. For the TV IPM this result appears in the work of \citet{yatracos_rates_1985} and \citet{devroye_combinatorial_2001}, and our result generalizes their argument
to other IPMs. Proposition~\ref{prop: IPM_gamma} ensures that $\bbE_{\truedist} T(P,Q)$ is negative for all $\truedist \in \trueclass$ under the null hypothesis in~\eqref{eq:approx_relative_test} with $\nu=3$. 
We can construct $\Chat_{\IPM\Redi}$ by inverting the IPM approximate relative fit test, to obtain an identical guarantee to the one in Corollary~\ref{thm:HELRedi_validity} (now with $\nu = 3$).

To further illustrate the construction of IPM approximate \Redi sets we consider three widely used IPMs---total variation distance, Wasserstein distance, and maximum mean discrepancy---where the witness functions are more explicit in Supplementary Material~\ref{sec:IPM_Redi_example}.

\subsection{Unified Sufficient Conditions for any Divergence Measure}
In this section we unify some of the treatment of the previous sections by giving conditions on split test statistics which ensure 
the exact and approximate validity of the resulting confidence sets.
Given data $\Data$, we consider tests of the form:
\begin{align*}
\phi_{P_0, P_1, \nu} = \I( \overline{T}_n(P_0,P_1) > t_\alpha(P_0,P_1)).
\end{align*}
We assume that the test statistic satisfies the following two additional conditions:
\begin{assumption} 
\label{ass:T_antisym}
   $T$ is anti-symmetric, i.e.,
   $T(X; P_0, P_1) = - T(X; P_1, P_0)$ for all $P_0, P_1 \in \model$.
\end{assumption}
\begin{assumption}
\label{ass:T_ub}
  There exists some fixed, positive numbers {$c_1 > 0$ and $\nu\geq 1$} such that for all $\truedist \in \trueclass$, and any fixed $P_0, P_1 \in \model$,
  \begin{align*}
    c_1 \bbE_{\truedist} T (\Data; P_0, P_1) \le \nu \rho (\truedist \| P_0) - \rho (\truedist \| P_1).    
  \end{align*}
\end{assumption}
{Assumption~\ref{ass:T_antisym} implies that the one-sided test statistic represents and preserves the directional difference from $\truedist$ to $P_0$ and $P_1$.}
Assumption~\ref{ass:T_ub} ensures that $\bbE_{\truedist} T (\Data; P_0, P_1)$ is always negative for all $\truedist\in\trueclass$ when the null hypothesis~\eqref{eq:approx_relative_test} is true. For instance, Propositions~\ref{prop:hel_ub} and \ref{prop: IPM_gamma} establish the analogue of Assumption~\ref{ass:T_ub} for Hellinger and IPM projection, respectively. 

Now, we may define $\rho$-\Redi set $\Chat_{\rho\Redi,n}$ as in KL \Redi set~\eqref{eq: ConfSet_Rift} by inverting the test based on (a $\delta$ corrupted version of) the statistic $T$: 
\begin{align}\label{eq: ConfSet_gen}
  \Chat_{\rho\Redi, n} := \left\{P\in \model : \overline{T}_{n_0, \delta}(P, \pilot) \le  \frac{z_{\alpha}  \hat{s}_{P,\delta}}{\sqrt{n_0}} \right\}
\end{align}

If the test statistic is bounded, i.e. $T(X;P_0,P_1) \leq B$ for any pair of distributions $P_0,P_1 \in \mathcal{P}^2$ then 
we can define the finite-sample $\rho$-\Redi set as in~\eqref{eqn:hoef}:
\begin{align}
\label{eqn:hoef_gen}
  \Chat_{\rho\Redi,B,n} = \left\{P\in \model : \overline{T}_{n_0} (P, \pilot) \le B\sqrt{\frac{ \log \left(1 / \alpha\right)}{2n_0} } \right\}
\end{align}

The following general result holds: 
\begin{theorem}
\label{thm:general}
Suppose that the test statistic satisfies Assumptions~\ref{ass:T_antisym} and~\ref{ass:T_ub}. 
\begin{enumerate}
\item Suppose that $\mathcal{Q}$ is such that for some $0 < \xi \leq 1$ the $2+\xi$ moments $\bbE_{\truedist} |T(X; P, Q) - \bbE_{\truedist}T(X; P, Q)|^{2+\xi} \leq M < \infty$ are finite 
for any~$(P,Q) \in\model^2$. 
Then, 
  \begin{align*}
     \inf_{\truedist \in \trueclass} \bbP_{\truedist} (\projdist \in \Chat_{\rho\Redi, n}) \geq 1 - \alpha - C n^{-\xi/2},
  \end{align*}
where $C < C' (1 + M) /\delta^{(2+\xi)}$ for a universal constant $C'$.
\item Suppose that $T(X; P,Q) \leq B$, then:
\begin{align*}
 \inf_{\truedist \in \trueclass} \bbP_{\truedist} (\projdist_\nu \cap \Chat_{\rho\Redi,B, n} \ne \emptyset) \geq 1 - \alpha.
\end{align*}
\end{enumerate}
\end{theorem}
The proof of the validity claims follows the same structure as the proof of Theorem~\ref{thm:Rift_validity}. The crucial Assumption~\ref{ass:T_ub} distills out the key property of the test statistics that is useful in ensuring asymptotic or 
finite-sample validity. With these general validity results in place, we now turn our attention to studying the size of the resulting robust universal sets. 


\section{Size of Robust Universal Confidence Sets}
\label{sec:size_set}
In the well-specified setting, for statistical models which satisfy classical regularity conditions, \citet{wasserman_universal_2020} showed that the Hellinger diameter of the split LRT confidence set depends on
two factors: the size of $\model$ determined by its (local) Hellinger bracketing entropy, and the closeness of $\pilot$ to $\truedist$ in the Hellinger distance. In a similar vein, in this section we show that 
the size of the universal sets, under certain regularity conditions, can be upper bounded by two factors: roughly, measuring the quality of the pilot estimate, and the size of statistical model. 

In the misspecified setting, we would like the robust universal set to shrink around its target at a fast rate. 
To measure the (directed) divergence between two sets measured in a divergence $\rho$ and with respect to $\truedist$ outside of $\model$, we define the \emph{$\rho_{\scH}^{\truedist}$-divergence} motivated by the directed Hausdorff distance. 
For a given divergence $\rho$ and a collection of distributions $S_1 \subset \model$, we define an $\epsilon$-fattening of $S_1$ by:
\begin{align*}
S_1 \oplus\epsilon := \cup_{Q \in S_1} \{P \in \model : \rho (\truedist \| P) \le \rho (\truedist \| Q) + \epsilon\}.
\end{align*}
Now given two collections of distributions $S_0, S_1 \subset \model$, we define the \emph{$\rho_{\scH}^{\truedist}$-divergence} by
\begin{align*}
  \rho^{\truedist}_{\scH} (S_0, S_1) = \inf \{ \epsilon \ge 0 : S_0 \subseteq S_1 \oplus \epsilon \}. \qquad
\end{align*}
$\rho^{\truedist}_{\scH} (S_0, S_1)$ is the minimum $\epsilon$-fattening of $S_1$ with reference to $\truedist$ containing $S_0$. 

To express the rate at which the robust universal sets shrink,
we use the Rademacher complexity of $\mathcal{F}_{T, \mathcal{P}}$, a function class which depends on the test statistic of choice, and the statistical model $\mathcal{P}$. Concretely, we define,
\begin{align*}
\mathcal{F}_{T, \mathcal{P}} := \left\{f: f(x) := T(x; P,Q),~~P,Q \in \mathcal{P} \right\}.
\end{align*}
We denote the Rademacher complexity of this class by $\mathfrak{R}_n(\mathcal{F}_{T, \mathcal{P}})$:
\begin{align*}
\mathfrak{R}_n(\mathcal{F}_{T, \mathcal{P}}) := \mathbb{E}\left[ \sup_{f \in \mathcal{F}_{T, \mathcal{P}}} \frac{1}{n} \sum_{i=1}^n R_i f(X_i)\right],
\end{align*}
where $R_i$ are i.i.d. Rademacher random variables.
In some of the cases we have considered in this paper, under additional regularity conditions the complexity measure 
$\mathfrak{R}_n(\mathcal{F}_{T, \mathcal{P}})$, can be related to a complexity measure of the underlying model $\mathcal{P}$ using a standard contraction argument \citep{ledoux_probability_1991}:
\begin{example} 
\label{lem:F_P_tranform} 
Suppose that $\calP$, and the pilot estimate $\pilot$ are distributions supported on some compact set $\mathcal{C}$, with density with respect to the Lebesgue measure which are upper and lower bounded by constants. 
Then, for the test statistics introduced in Sections~\ref{sec:KL},\ref{sec:dpd} and \ref{sec:hellinger}, $\mathfrak{R}_n(\mathcal{F}_{T, \mathcal{P}}) \lesssim \mathfrak{R}_n(\mathcal{P})$. 
\end{example}

Finally, to characterize the quality of the pilot estimator $\pilot$, we say that the 
$\pilot$ is an $\eta_n$-consistent estimator if
\begin{align}\label{eq:pilot_close}
  \rho (\truedist \| \pilot) - \rho(\truedist \| \projdist) = O_{\truedist} (\eta_n),
\end{align}
where we use the standard big O in probability notation to indicate stochastic boundedness.

With these preliminaries in place, we have the following result for the size
of the $\rho$-\Redi set obtained by inverting a finite-sample valid relative fit test. The proof will be given in Supplementary Material~\ref{sec: proof_delta_bound}. 
\begin{theorem} \label{thm:exact_optimal_shrink}
Suppose that \eqref{eq:pilot_close} holds and $\sup_{(P, Q)\in\model^2} |T(P, Q) - \mathbb{E} T(P,Q)| \le B$. 
Fix any projection distribution $\projdist$, and recall the collection $\projdist_\nu$ in~\eqref{eqn:pnu}.
Then the robust universal confidence set $\Chat_{\rho\Redi,B,n}$ in~\eqref{eqn:hoef_gen}, for an equal sized split into $\mathcal{D}_0$ and $\mathcal{D}_1$, 
satisfies for any $\truedist\in \trueclass$,
\begin{align*}
  \rho_{\scH}^{\truedist} \left(  \Chat_{\rho\Redi,B,n}, \projdist_\nu\right) \le O_{\truedist} \left( \eta_n + \mathfrak{R}_n(\mathcal{F}_{T, \mathcal{P}}) + B\sqrt{\frac{\log(1/\alpha)}{n}}\right).
\end{align*} 
\end{theorem}

Theorem~\ref{thm:exact_optimal_shrink} states that the directed $\rho_{\scH}^{\truedist}$-divergence between the \emph{exact} robust universal confidence set and its target $\projdist$ shrinks to zero at the prescribed rate, since $\projdist_\nu$ is a singleton $\{\projdist\}$ when $\nu = 1$. One can no longer show such a result for the $\nu$-approximate robust universal confidence set even with an infinite number of observations. This is because, conditional on $\Data_1$, the split test $\phi_{P, \pilot, \nu}$ is guaranteed to achieve (exponentially) small Type~\rom{2} error uniformly over $\truedist\in\trueclass$ only for distributions $P$ which are at least $\nu \rho(\truedist \| \pilot)$ away from $\truedist$.
Nevertheless, Theorem~\ref{thm:exact_optimal_shrink} characterizes the rate at which $\Chat_{\rho\Redi,B,n}$ shrinks to $\projdist_\nu$.

Theorem~\ref{thm:exact_optimal_shrink} also shows how the size of the set depends on the choice of $\pilot$. When possible we should 
choose a pilot estimate $\pilot$ which converges to the target $\projdist$ at a fast rate to ensure that the term $\eta_n$ is sufficiently small. A sensible choice is often a minimum distance estimator \citep[MDE,][]{parr_minimum_1980} which is not only a consistent estimator of $\projdist$ under some regularity conditions but is also robust to some misspecification in its corresponding distance \citep{millar_robust_1981,donoho_automatic_1988}.

\section{{Practical Implementation}}
\label{section::prac}

{
In this section we discuss practical implementation details
and illustrate the method on the setting of Example~\ref{ex:mix_unident}.
Our confidence sets have the form
$C_\alpha = \{\theta\in\Theta:\ \overline{T}_{n_0}(\theta,\thetahat_1) \leq t_\alpha(\theta,\thetahat_1)\}$
for some statistic 
$\overline{T}_{n_0}(\theta,\thetahat_1)$ and critical value $t_\alpha(\theta,\thetahat_1)$.
One approach is 
to evaluate
$\overline{T}_{n_0}(\theta,\thetahat_1)$ and $t_\alpha(\theta,\thetahat_1)$ on a regular grid in $\Theta$, but this can be expensive.
Lately,
researchers in
simulation based inference,
where
confidence sets require numerical evaluations across many $\theta$ values,
have taken a different approach.
These methods draw a large sample of
$\theta$'s from any convenient distribution and then apply nonparametric smoothing
over $\Theta$.
See e.g.,
\citet{dalmasso_likelihood-free_2024}.
We propose a similar approach here.
Assuming that $\Theta$
is a compact subset of $\mathbb{R}^d$, 
the steps are as follows:
\begin{enumerate}
    \item Draw a sample $\theta_1,\ldots, \theta_B$ from a fully supported
distribution $\pi$.
    \item For $j=1,\ldots, B$, let
$$
Z_j = \I \left\{ \overline{T}_{n_0}(\theta_j,\thetahat_1) > t(\theta_j,\thetahat_1) \right\}.
$$
    \item Perform a nonparametric regression
of
$(Z_1,\ldots, Z_B)$ on
$(\theta_1,\ldots, \theta_B)$
to estimate the $p$-value
$pv(\theta_j) = \mathbb{E}[Z_j|\theta_j]$.
Denote the estimated function by
$\hat{pv}(\theta)$.
    \item Return
$C_\alpha = \{\theta:\ \hat{pv}(\theta) \geq \alpha\}$.
\end{enumerate}
}

\noindent{
The set $C_\alpha$ approaches the true confidence set as $B$ increases.
Formally we have the following from standard nonparametric regression results:
\begin{proposition}
Let $pv(\theta) = \bbP (T(\theta,\thetahat_1) \geq t(\theta,\thetahat_1))$.
Assume that,
conditional on $\Data_1$,
$\sqrt{n}(T(\theta,\thetahat_1) - t(\theta,\thetahat_1))
 \overset{\truedist}{\rightsquigarrow} N(0,\sigma^2(\theta,\thetahat_1))$.
Further, assume that
$\sup| \hat{pv}(\theta) - pv(\theta)| = O_{\truedist}( (\log B/B)^\gamma)$
for some $\gamma>0$.
Then 
$$
\bbP_{\truedist}(\projtheta(P) \in C_\alpha) = 1-\alpha + O_{\truedist}(n^{-1/2}) + O_{\truedist}( (\log B/B)^\gamma).
$$
\end{proposition}
}

{
An active area of research in simulation based inference is on devising methods for choosing the $\theta_j$'s sequentially.
As we get information about the function
$pv(\theta)$ we adapt the distribution $\pi$ to sample
near the boundary of the confidence set.
These methods have great promise for dealing with high-dimensional parameter spaces.}

\section{Illustration: Contaminated Data Generating Process}
\label{sec:Empirical_analysis}

{In this section, we evaluate our proposed exact and approximate robust universal confidence sets in contamination setups and demonstrate the advantages of our proposed methods. Two additional illustrations---overdispersion and irregular Bernoulli trials corresponding to Example~\ref{ex: instability_lr} and \ref{ex: sLRT_fail}---are provided in the Supplementary Material~\ref{sec:empirical_bernoulli} and \ref{sec:overdisp}} to demonstrate the broader application of the method.

Consider the following contaminated data generating distributions
 which are mixtures of Gaussians. This simulation setup is used in the work of \citet{catoni_challenging_2012}.
\begin{align*}
  \truedist_1 &= 0.99 N(0, 1) + 0.01 N(0, 30^2) &\text{(Symmetric)}\\
  \truedist_2 &= 0.94 N(0, 1) + 0.01 N(20, 20^2) + 0.05 N(-30, 20^2)  &\text{(Asymmetric)}\\
  \truedist_3 &= 0.7 N(2, 1) + 0.2 N(-2, 1) + 0.1 N(0, 30^2) &\text{(Heavily Asymmetric)}
\end{align*}

Consider a location-scale family of Gaussian distribution $\mathcal{P}_\Theta = \{N(\mu, \sigma^2 ) : (\mu, \sigma)\in\Theta\}$ as a working model. This working model is a sensible choice, particularly for $\truedist_1$ and $\truedist_2,$ as both distributions are rooted in the standard Gaussian distribution. As $\truedist_3$ originates from a bimodal distribution, the working model may not be optimal; nonetheless, it illustrates how our method can offer a valid inference even under such severe misspecification. (See Supplementary Material~\ref{sec:vis_dgp} and \ref{app: contam_loc} for distribution plots and additional simulations for a location family with fixed scale.) Our goal is to evaluate the empirical performance---coverage and size---of robust universal confidence sets for the (approximate) projection of the various contaminated distributions onto $\mathcal{P}$. 

Figure~\ref{fig: contam_proj} shows the mean and standard deviation of the projection distribution with respect to the KL, DP, Hellinger and TV distances along with the mean and standard deviation of the contaminated and uncontaminated distributions. The KL projection parameter is the same as the parameters of contaminated distribution in all cases suggesting the target instability to the contamination.
The DP projection parameters, get closer to uncontaminated parameters as the $\beta$ parameter increases. 
The Hellinger projection is the closest to the uncontaminated parameters among all projections we considered, however, the size of $\projdist_\nu$ is much larger than that of approximate TV projection. The set $\projdist_\nu$ for both Hellinger and TV distance is quite large for the heavily misspecified case (Case 3).
Practically, we recommend targeting DP projection with a reasonable choice of $\beta (> 0.05)$ for this heavily misspecified case.
\begin{figure}[!htb]
\centering
    \begin{subfigure}{.32\textwidth}
        \centering
        \caption{Case 1: Symmetric}        
        \includegraphics[trim={15 20 25 20},clip,width=\textwidth,height=1.55in]{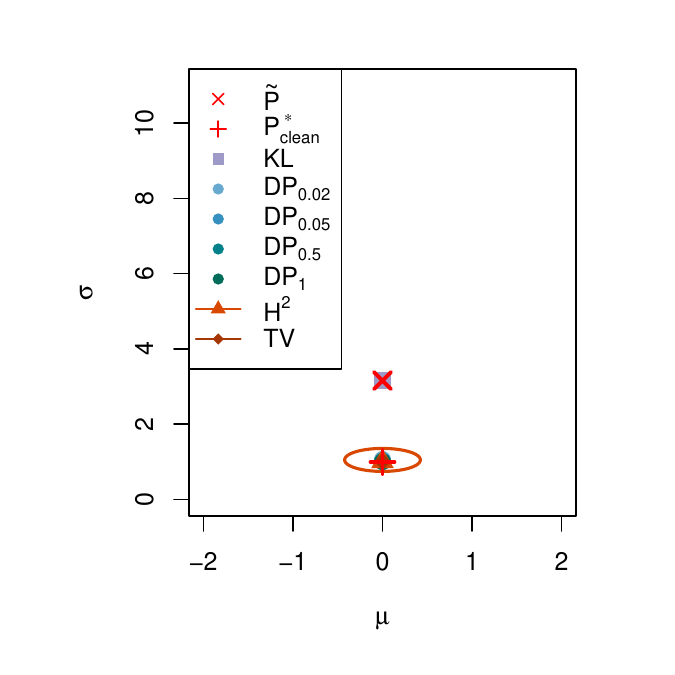}
    \end{subfigure}
    \begin{subfigure}{.32\textwidth}
        \centering
        \caption{Case 2: Asymmetric}
        \includegraphics[trim={15 20 25 20},clip,width=\textwidth,height=1.55in]{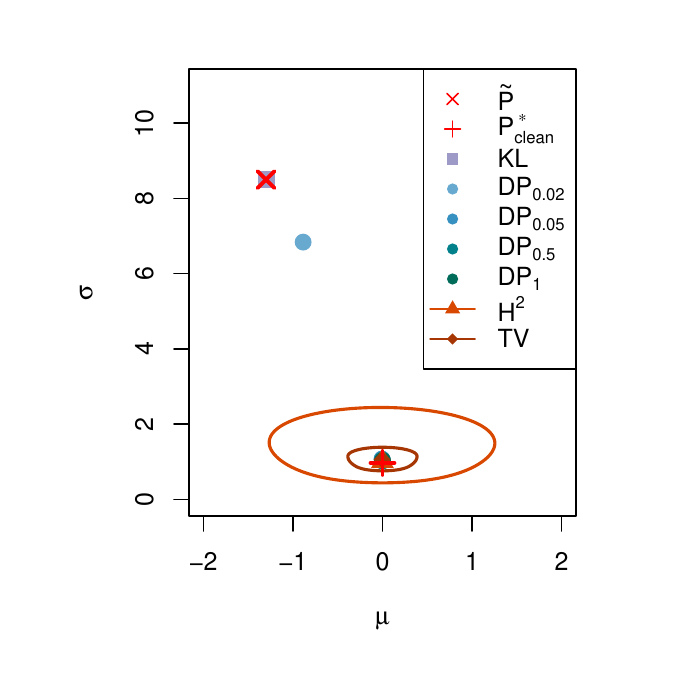}
    \end{subfigure}
    \begin{subfigure}{.34\textwidth}
        \centering
        \caption{Case 3: Heavily asymmetric}
        \includegraphics[trim={15 20 25 20},clip,width=\textwidth,height=1.55in]{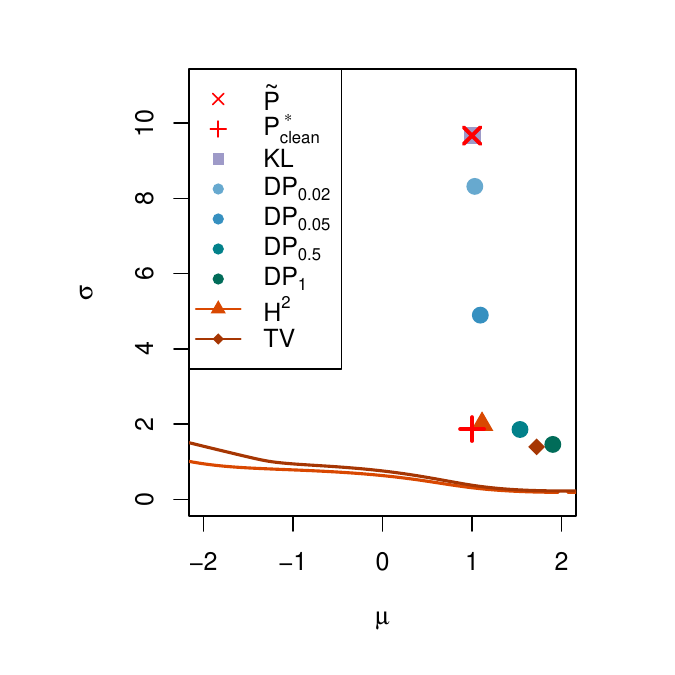}
    \end{subfigure}    
    \caption{Projection parameters ($\KL$, $\DP_\beta$, $\HEL$, $\TV$), contaminated distribution ($\projdist$), and uncontaminated distribution $(\truedist_{\Truth})$. Level set of parameters of $\projdist_\nu$ for Hellinger and TV distance is given in contours.
    }\label{fig: contam_proj}
\end{figure}

Figure~\ref{fig: contam_sLRT_Rift} illustrates the empirical coverage and size of split LRT and ($\KL$ and $\DP$) \Redi sets based on 1000 replications. For split LRT and KL \Redi sets, we choose $\hat{\theta}_1$ to be the quasi-MLE {of the first half samples}, whereas, for the
DP \Redi set, we use the minimum DP divergence estimator. The {non-robust} split LRT fails to cover KL projection in all cases whereas \Redi sets achieve the nominal coverage with large enough sample sizes. The DP \Redi sets show superior coverage than KL \Redi set across all sample sizes. Such a target coverage improvement is more evident in the smaller sample sizes below 200, and as $\beta$ gets larger, i.e., the DP \Redi set targets a more stable projection. Regardless of what divergence measure $\rho$ is of interest, the size of the confidence set with reference to $\rho$ shrinks to zero as the sample size increases. Again, the extremal values of $\KL_\scH^{\truedist} (\Chat_{\Redi}, \projdist)$ for sample sizes below 500 highlight the instability of KL projection. 

\begin{figure}[!htb]
\centering
\begin{subfigure}{.32\textwidth}
\centering
\caption{Case 1: Symmetric}
\includegraphics[trim={0 20 40 20},clip,width=\textwidth,height=1.55in]{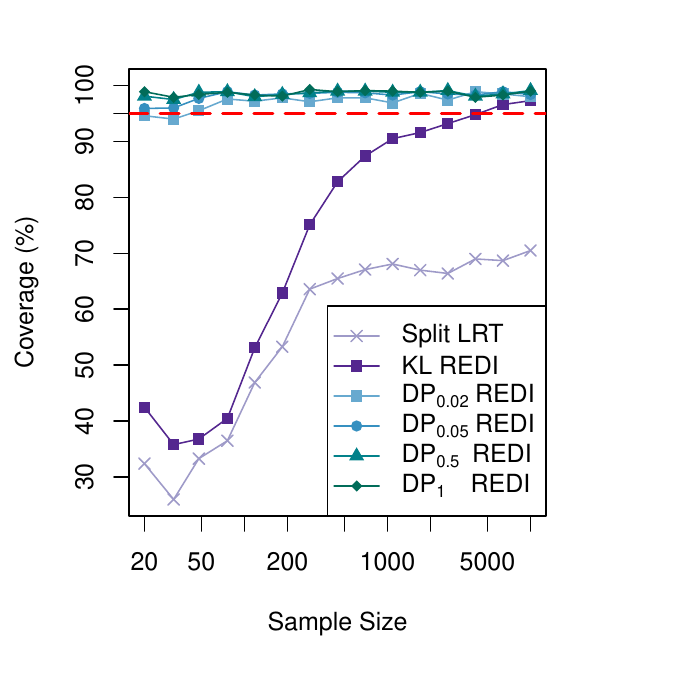}
\end{subfigure}
    \begin{subfigure}{.32\textwidth}
        \centering
       \caption{Case 2: Asymmetric}
        \includegraphics[trim={0 20 40 20},clip,width=\textwidth,height=1.55in]{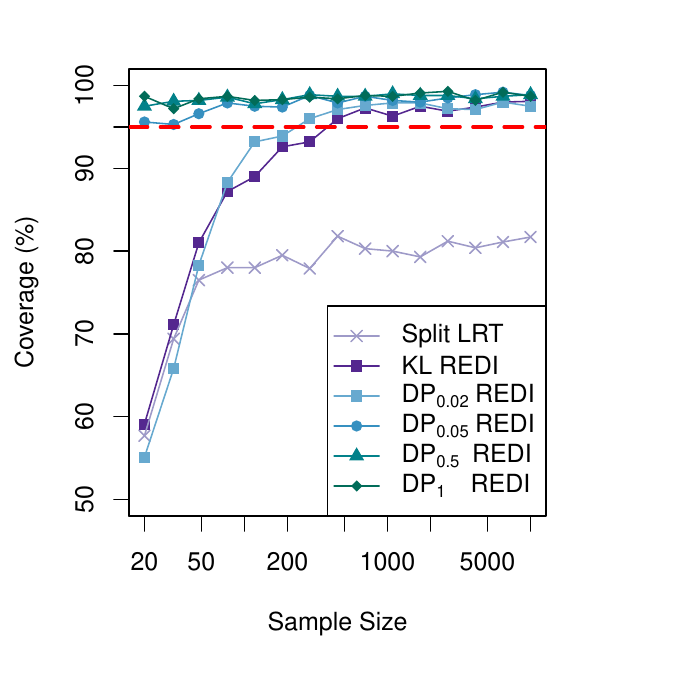}
    \end{subfigure}
    \begin{subfigure}{.34\textwidth}
        \centering
       \caption{Case 3: Heavily asymmetric}        
        \includegraphics[trim={0 20 40 20},clip,width=\textwidth,height=1.55in]{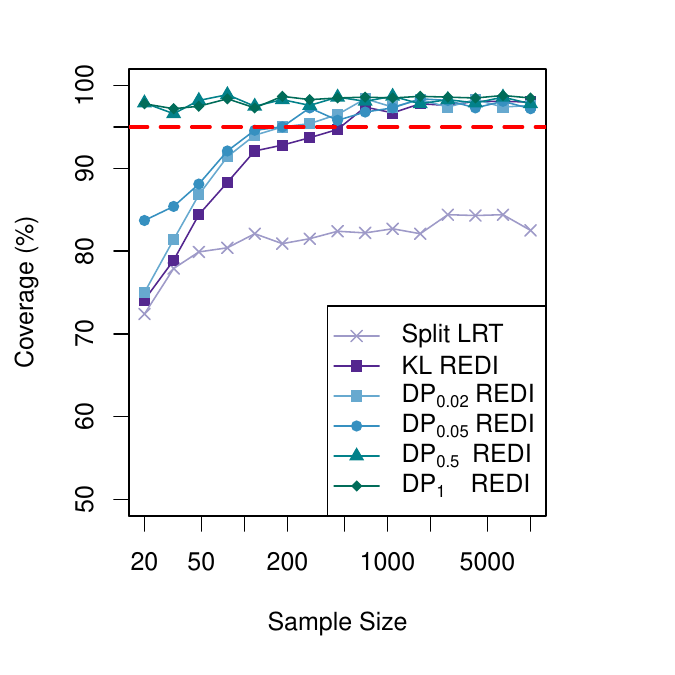}
    \end{subfigure}    
    \begin{subfigure}{.32\textwidth}
        \centering
        \includegraphics[trim={15 20 25 20},clip,width=\textwidth,height=1.55in]{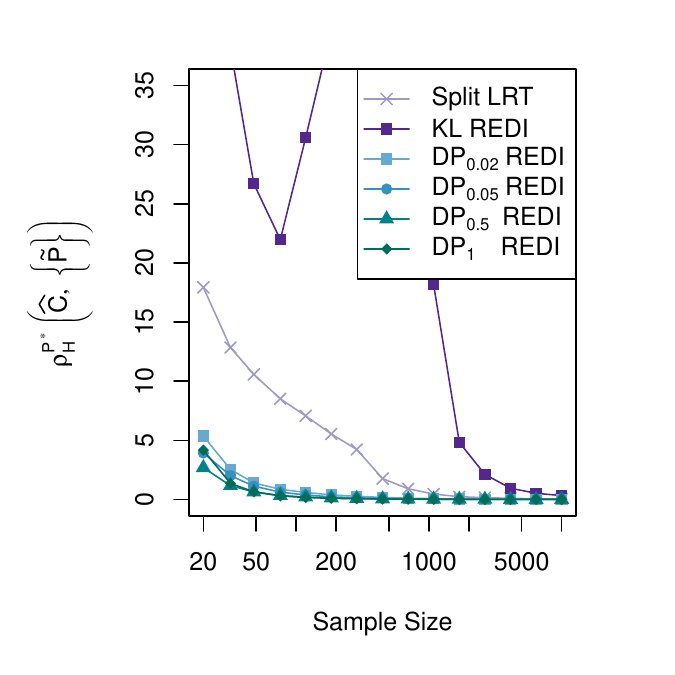}
    \end{subfigure}
    \begin{subfigure}{.32\textwidth}
        \centering
        \includegraphics[trim={15 20 25 20},clip,width=\textwidth,height=1.55in]{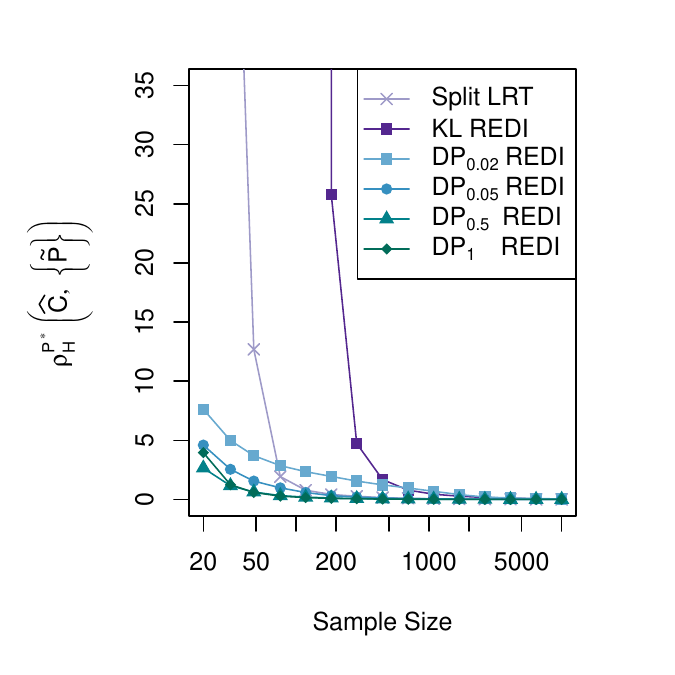}
    \end{subfigure}
    \begin{subfigure}{.34\textwidth}
        \centering
        \includegraphics[trim={15 20 25 20},clip,width=\textwidth,height=1.55in]{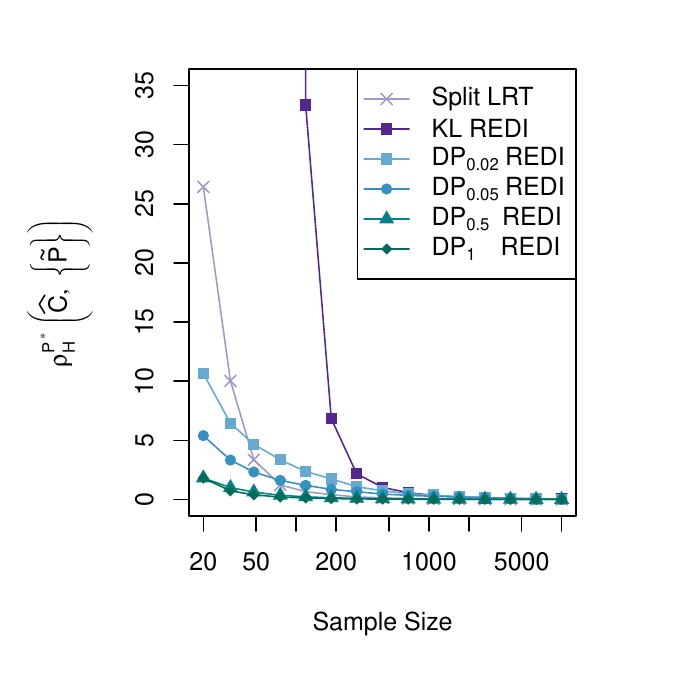}
    \end{subfigure}
    \caption{Performance of confidence sets for KL and DP projections on contamination examples over 1000 replications. (Top) Empirical coverage. (Bottom) Median $\rho$-diameter.}\label{fig: contam_sLRT_Rift}
    
\end{figure}

Figure~\ref{fig: contam_HRift} shows the maximal $\rho$-distance of $\HEL$\Redi and $\TV$\Redi set from $\truedist$ based on 1000 replications along with the $\rho(\truedist \| \projdist_\nu)$, a set of $\rho$-distance from $\truedist$ to approximate projection $\projdist_\nu$. $\rho(\truedist \| \projdist_\nu)$ illustrates the same phenomena as in Figure~\ref{fig: contam_proj} but with respect to each distance. Theoretically, we can only claim that the shrinkage of \Redi set up to $\projdist_\nu$. This can be seen in Figure~\ref{fig: contam_HRift} for both Hellinger and TV \Redi set as the maximum excess distance from $\truedist$ reaches $\nu\rho(\truedist \| \projdist_\nu)$ with large enough samples. \Redi sets shrink beyond $\projdist_\nu$ in this example: the $\HEL$\Redi set converges to a set quite close to $\projdist$ with large enough sample size, while the $\TV$\Redi set converges to a set around $\projdist$ which does not appear to shrink with sample size. 

\begin{figure}[!htb]
\centering
    \begin{subfigure}{.30\textwidth}
        \centering
      \caption{Case 1: Symmetric}
        \includegraphics[trim={15 20 25 20},clip,width=\textwidth,height=1.55in]{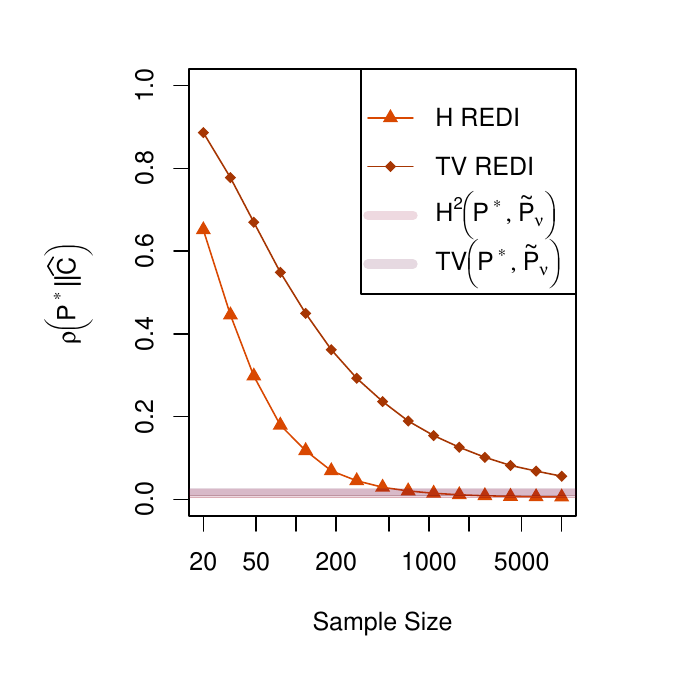}
    \end{subfigure}
    \begin{subfigure}{.30\textwidth}
        \centering
      \caption{Case 2: Asymmetric}
        \includegraphics[trim={15 20 25 20},clip,width=\textwidth,height=1.55in]{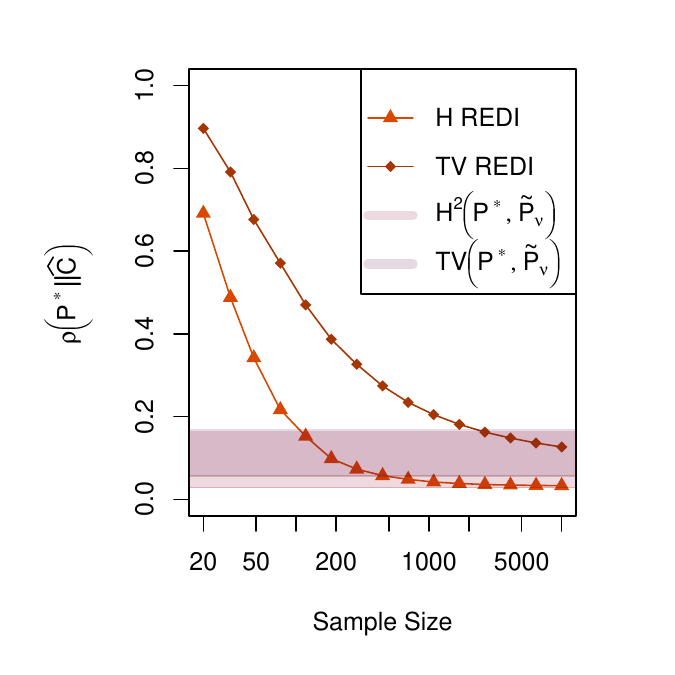}
    \end{subfigure}
    \begin{subfigure}{.32\textwidth}
        \centering
      \caption{Case 3: Heavily asymmetric}
        \includegraphics[trim={15 20 25 20},clip,width=\textwidth,height=1.55in]{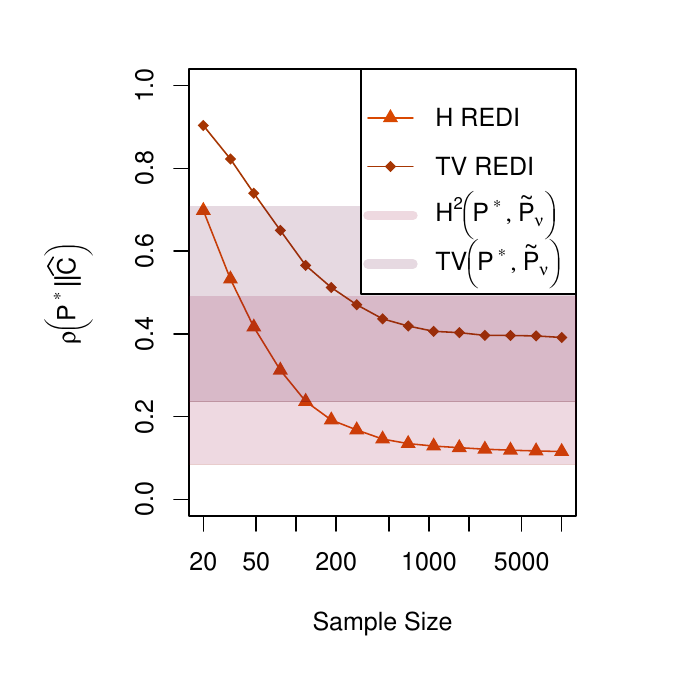}
    \end{subfigure}    
    \caption{Size of \Redi confidence sets of Hellinger and TV projection on contamination examples over 1000 replications superimposed with $\{\rho(\truedist \| P): P\in \projdist_\nu\}$.}\label{fig: contam_HRift}    
\end{figure}

\section{Application: Causal Discovery with Homoscedastic Linear Causal Model}\label{sec:LSEM}

In this section, we demonstrate how our proposed method can be applied to construct robust confidence sets for total causal effects using a bivariate Gaussian linear structural equation model (LSEM) with homoscedastic errors \citep{peters_identifiability_2014}. The homoscedastic bivariate Gaussian LSEM is formally defined as follows:
\begin{align}\label{eq:LSEM}
  \model = \left\{ \Normal \left( 0, \Sigma \right) :  \Sigma= \sigma^2 (I_2 - B)^{-1} (I_2 - B)^{-\top} , \quad
     B = \begin{pmatrix}
       0& \beta_{12}\\
       \beta_{21}& 0
     \end{pmatrix} \right\},
\end{align}
where $\sigma^2 \in \mathbb{R}_{+}$ is an unknown variance parameter and $\beta_{ij}$ are unknown parameters of direct causal effect from $X_i$ to $X_j$.
Under the model $\model$, the target of inference is the \emph{total causal effect} $C(i\to j) = (I_2 - B)^{-1}_{ji}$ for $(i,j) \in \{(1,2), (2,1)\}$, i.e.,
\begin{align*}
  C (i\to j) = \Sigma_{ij} / \Sigma_{ii} \I(\Sigma \in \mathcal{M}_{r_i}), \qquad
  \mathcal{M}_{r_i} := \left\{ \Sigma \succ 0 :  \Sigma^2_{ii} = |\Sigma|  \right\}.
  \end{align*}

Given iid samples $\bm{X}_1, \dots, \bm{X}_n \sim \truedist$, our goal is to construct 
a confidence set for the total causal effect $C(i\to j)$. When the model is well specified ($\truedist \in \calP$), the causal structure is identifiable \citep{peters_identifiability_2014}, and one can construct confidence sets by  inverting either (A) likelihood ratio tests that specify $C(i \to j)$, with the alternative relaxed to the positive definite cone instead of the union $\mathcal{M}{r_1} \cup \mathcal{M}{r_2}$, or (B) the split likelihood ratio test \citep{strieder_confidence_2021}. \citet{strieder_confidence_2021} showed that 
the LRT can be sensitive to model misspecification 
despite having narrower confidence set than the split LRT confidence set. Recall that split LRT is valid only when model is well specified (or convex) and its implicit target of inference---
the KL projection---can be sensitive to the model misspecification. Therefore we 
compare the LRT set, the split LRT set, the $\KL$\Redi, and the $\DP$\Redi sets 
including crossfit variants on datasets exhibiting varying degrees and types of model misspecification. See Supplementary Material~\ref{sec:LSEM_supp} for details on constructing \Redi sets under homoscedastic bivariate Gaussian LSEM, and for additional simulation results.

We used the \textsc{CausalEffectPairs} \citep{mooij_distinguishing_2016} dataset, which offers a collection of cause-effect pairs from diverse domains, with ground-truth causal directions established through domain expertise. For our application, we consider three representative pairs: pair 76 (the average annual growth of population $X_1$ and the food consumption $X_2$), pair 18 (age of child $X_1$ and chemical GAG concentration in urine $X_2$), and pair 12 (age $X_1$ and wage per hour $X_2$). These examples were selected to illustrate varying degrees and types of model misspecification. Pair 76 is considered nearly well-specified and was also analyzed in \citet{strieder_confidence_2021}. Pair 18 exhibits nonlinearity and possible deviations from homoscedasticity, while pair 12 involves both nonlinearity and heteroscedasticity, as well as notable outliers. See Figure~\ref{fig:CEP_data_scatterplot} for scatterplots of each pairs after centering.

\begin{figure}[!htb]
\centering
\includegraphics[width=0.7\textwidth]{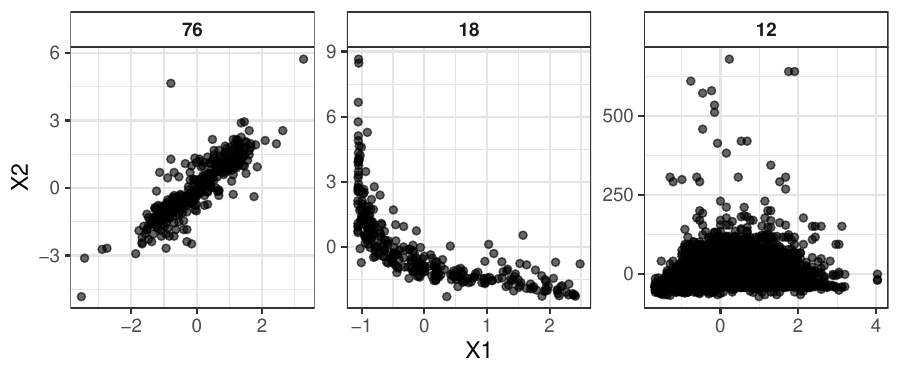}  
\caption{Scatterplots of selected \textsc{CausalEffectPairs} after centering.}\label{fig:CEP_data_scatterplot}
\end{figure}

\begin{figure}[!htb]
  \centering
  \includegraphics[width=\textwidth]{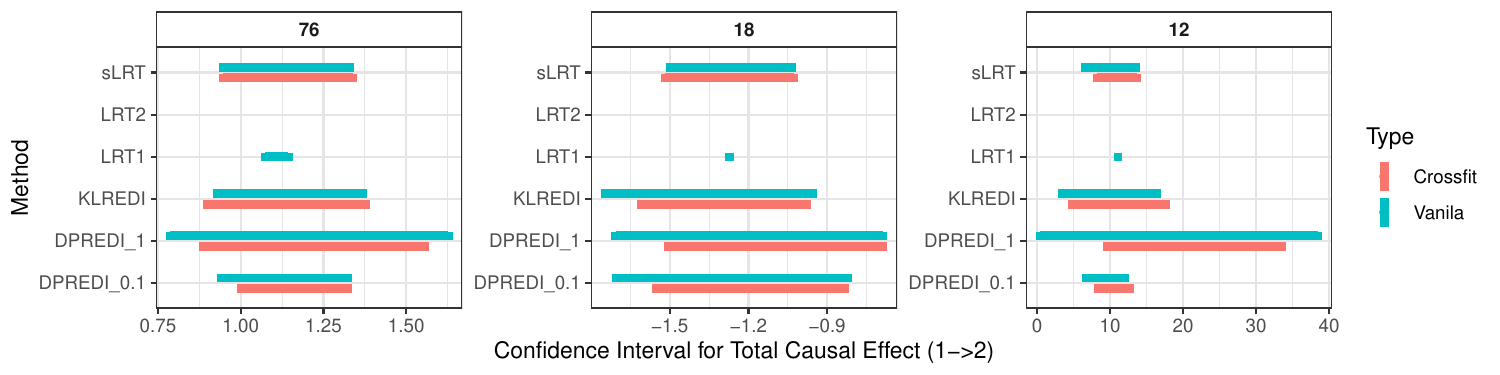}
  \includegraphics[width=\textwidth]{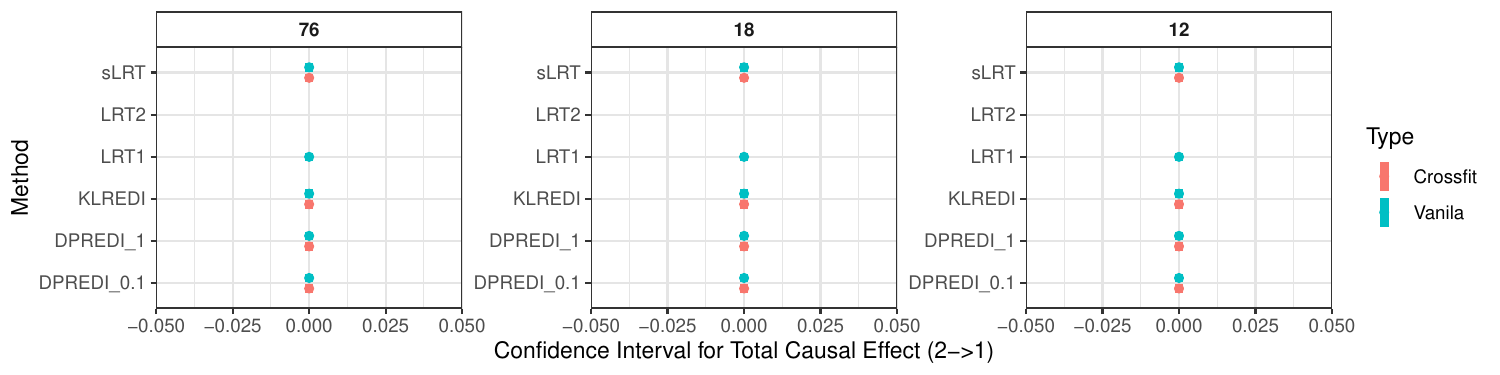}
  \caption{95\% confidence intervals for the total causal effect from $X_1$ to $X_2$ (Top) and from $X_2$ to $X_1$ (Bottom) for selected \textsc{CausalEffectPairs}.}\label{fig:CEP_CI}
\end{figure}

Figure~\ref{fig:CEP_CI} shows the 95\% confidence sets for the total causal effects $C(1\to2)$ and $C(2\to1)$, constructed using two LRT variants \citep[\texttt{LRT1} and \texttt{LRT2},][]{strieder_confidence_2021}, split LRT, KL\Redi, DP\Redi with $\beta= 0.1, 1$. For split LRT and \Redi sets, we provide both vanilla variant---based on equally split data---and the crossfit variant, which averages the split test statistics and its corrupted variances by swapping the role of split samples $\Data_0$ and $\Data_1$. This crossfitting approach can help reduce split-sample variability and improve efficiency. See Supplementary Material~\ref{sec:crossfit} for details. All methods correctly identified the ground-truth causal direction across all three pairs. \texttt{LRT2} fails to provide confidence set in all cases, despite its strong simulation performance, whereas \texttt{LRT1} yielded the narrowest confidence intervals. 

In the nearly well-specified example (e.g., pair 76), the split LRT and 
the KL\Redi produced comparable confidence sets. Among all methods, $\DP\Redi_1$ yielded the widest interval, consistent with the lower efficiency of minimum density power divergence methods compared to MLE \citep{basu_robust_1998}. Interestingly, $\DP\Redi_{0.1}$ resulted in a narrower confidence interval than both KL\Redi and split LRT. This suggests that $\DP\Redi$ may have effectively downweighted influential observations (e.g., $(X_1, X_2) = (-1.6, 9.5)$), achieving robustness with only a modest loss in efficiency.

For pair 18, where the model is moderately misspecified, the KL\Redi produced wider confidence sets than the split LRT, appropriately reflecting the uncertainty due to misspecification. In this case, $\DP\Redi$ (for both $\beta = 0.1$ and $\beta = 1$) yielded even wider sets than KL\Redi set. As $\beta$ increases, the $\DP\Redi$ confidence set shifts slightly toward larger values of $C(1 \to 2)$, though the shift remains minor. This shift reflects the effect of DP divergence downweighting influential observations near $X_1 = -5$,

For pair 12, where the data exhibit nonlinearity, heteroscedasticity, and contain outliers, the $\KL$\Redi produced a wider confidence set than the split LRT as in pair 18. The $\DP\Redi_{0.1}$ set was narrower than both split LRT and $\KL$\Redi sets, while remaining similarly centered, indicating a reduced variance from a small robustness gain. In contrast, the $\DP\Redi_{1}$ set shifted substantially toward larger $C(1\to 2)$ and was the widest interval. This is because the DP projection with a higher $\beta$ is a more stable inferential target as seen in Figure~\ref{fig: contam_proj}, though this robustness may come at the expense of efficiency as also shown in pair 76.

\section{Discussion}
\label{sec:discussion}
In this paper, 
we presented a general method for constructing uniformly valid exact and approximate confidence sets for 
various projection distributions
under weak regularity conditions in the presence of possible model misspecification. 
We demonstrated that the universal inference procedure~\citep{wasserman_universal_2020} can fail catastrophically
even in simple examples, under fairly benign model-misspecification. We then showed 
that the robust universal inference framework can address these failures, providing methods that are robust and can 
meaningfully target different projection distributions.

{Data splitting is useful in the design of assumption-light methods. The randomness due to sample splitting
can be ameliorated using the idea of \citet{gasparin2024}. Specifically, we repeat our method over $B$ random splits, yielding confidence sets $\{C_1,\ldots, C_B\}$ at level $1-2\alpha$. The set $C$, containing parameters present in at least half of these sets, is then a valid $1-\alpha$ confidence set by Markov's inequality.}
The conservativeness of sample-split methods can be partially addressed with \emph{crossfitting} where we average the split statistic with that after swapping the role of $\Data_0$ and $\Data_1$. In contrast to the well-specified setting where the validity of the crossfit set is immediate, extra care is needed under model-misspecification. We investigate the validity of the crossfit set in Supplementary Material~\ref{sec:crossfit} { and empirically show that crossfit set can be substantially smaller than the single-split set in Supplementary Material~\ref{sec:power_correct_regular_model}.}

The papers~\citep{wasserman_universal_2020,dunn2022gaussian} study many variants of universal inference (including constructing confidence sequences instead of confidence sets, to combining multiple sample-splits) and investigating these variants in the context of the robust universal inference framework of this paper would be interesting. One can apply $\KL\Redi$ to estimate the KL projection conditional distribution under model misspecification. See Supplementary Material~\ref{sec:KLRedi_regression} for details. Extending these results to other projections in the regression context would be another interesting future work.

Finally, our paper brings to the forefront the role of pairwise tests of fit (and relative fit) together with sample-splitting, in designing broadly applicable inference methods. We expect this basic insight to 
have further implications in other contexts, for instance, in designing universal inference procedures in other settings where likelihood-based methods are inapplicable. 

\section*{Acknowledgements}
{BP is currently at Google. This work was primarily conducted while BP was at Carnegie Mellon University.}
This work was partially supported by funding from the NSF grants DMS-1713003, DMS-2113684 and
CIF-1763734, as well as an Amazon AI and a Google Research Scholar Award to SB. The authors are grateful to Arun Kuchibhotla, Aaditya Ramdas, Tudor Manole, and Ian Waudby-Smith for
helpful discussions.

\bibliographystyle{plainnat}
\bibliography{references}

\newpage

\appendix

\section{Examples of IPM-\Redi set}
\label{sec:IPM_Redi_example}

\subsection{Total Variation Distance.}
\label{sec:test_TV}
Suppose $\rho(P_0 \| P_1) = \TV (P_0, P_1)$ where $\TV$ is the total variation distance. This is an IPM over the function class $\calF = \{f : \supnorm{f} \le 1\}$. An equivalent definition is $\TV (P_0, P_1) = \sup_{A} | P_0 (A) - P_1(A) | = P_0 (\calA) - P_1 (\calA)$ where $\calA = \{p_0 > p_1\}$ is the Yatracos set with maximal discrepancy between $P_0$ and $P_1$. The witness function is $f^*_{(P_0, P_1)} (x) = \I (x \in \calA) - 1/2$. An immediate advantage of targeting the
 TV projection comes from that $f^*$ is uniformly bounded.
 Given samples $\Data$, consider the following test statistic which referred to as the \emph{split Scheff{\'e} statistic}:
\begin{align*}
  \overline{T}_{n_0} (P,\pilot) = \frac{P (\calA) + \pilot (\calA)}{2} - \truedist_{n_0}(\calA), \qquad
  \truedist_{n_0} (\calA) = \frac{1}{n_0} \sum_{i\in\calI_0} \I (X_i \in \calA)
\end{align*}
where $\calA$ is redefined to be $\calA = \{p > \widehat{p}_1\}$. 
The split Scheff{\'e} statistic, as the name suggests, is a sample-split analogue of the Scheff{\'e} estimate that was originally proposed in \citet{devroye_combinatorial_2001} building upon the work of \citet{yatracos_rates_1985}. 

\subsection{Wasserstein Distance. }
\label{sec:test_W}

Suppose $\rho(P_0 \| P_1) = \Was_1 ({P_0}, {P_1})$ is the 1-Wasserstein distance (or Kantorovich metric). The associated function class is $\calF = \{f: \pnorm{L}{f} \le 1 \}$ where $\pnorm{L}{f} := \sup \{ |f(x) - f(y) | / \|x - y \| : x\ne y \}$ is the Lipschitz semi-norm. 
Although the ideas are much more general, we limit our discussion to univariate distributions on a compact support, i.e., $\calX = [0,b]$. In this case, the witness function is explicit and easy to describe~\citep{baraud_tests_2020}. 

Define $\sgntilde{t; {P_0}, {P_1}} = \I \left( F_{P_1}(t) > F_{P_0} (t) \right) - \I \left( F_{P_0} (t) > F_{P_1} (t) \right) \in \{0, \pm 1 \}$, 
where $F_P$ denotes the CDF of $P$. 
The witness function is 
$f^*_{(P_0, P_1)} (x) = \int_{0}^{x} \sgntilde{t; {P_0}, {P_1}} \d t$~\citep{baraud_tests_2020}.

A direct application of the split statistic~\eqref{eq: IPM_split_stat} yields
\begin{align*}
  \overline{T}_{n_0} (P,\pilot) = \frac{1}{2} \int \sgntilde{t; P, \pilot} \left( \truedist_{n_0} (t) - \frac{F_{P} (t) + F_{\pilot} (t)}{2} \right) \d t,
\end{align*}
where $\truedist_{n_0} (t) = \frac{1}{n_0}\sum_{i\in\calI_0} \I(X_i \le t)$ is the empirical distribution. This particular split statistic is a sample-split analogue of the $\ell$-estimator \citep[Proposition~4 of Baraud][]{baraud_tests_2020}. 

\subsection{Maximum Mean Discrepancy. }
\label{sec:test_MMD}

Suppose that  $\calF$ is a unit ball of the reproducing kernel Hilbert space (RKHS) $\calH$, 
with kernel $k(x,y)$, and RKHS norm $\|\cdot\|_{\mathcal{H}}$,
i.e., $\calF = \{f: \pnorm{\calH}{f} \le 1\}$. Then the corresponding IPM is called the Maximum Mean Discrepancy \citep[MMD,][]{gretton_kernel_2012}. It was shown by \citet{gretton_kernel_2012} that the analytic witness function $f^*_{(P, \pilot)} = \frac{\mu_{P} - \mu_{\pilot}}{\pnorm{\calH}{\mu_{P} - \mu_{\pilot}}}$ where $\mu_{P}(\cdot) := \mathbb{E}_P [k(X,\cdot)]$ is the mean embedding of $P$.

The split statistic $\overline{T}_{n_0} (P, \pilot)$ in this case 
reduces to an average of the (negative) witness function $- \truedist_{n_0} (f^*_{(P, \pilot)} )$ if the kernel $k(\cdot,\cdot)$ is symmetric. In this case, the sign of the split statistic captures, in expectation, whether the population $\truedist$ is closer to $P$ or $\pilot$ based on mean embeddings.

\section{Proofs from Section~\ref{sec:fail}}
\label{app:examples} 

\subsection{Example~\ref{ex: instability_lr}}
\begin{proof}
Note that the KL projection $\projdist = \Bern\left( 1/2 \right)$. Consider the event $E$ where all of the observed samples $X_1,\ldots,X_n$ are 0.
We can see that,
\begin{align*}
\bbP_{\truedist}(E) = 1 - \bbP_{\truedist} \left( \sum_{i=1}^n X_i > 0 \right) \geq 1 - \bbE_{\truedist} \left[ \sum_{i=1}^n X_i \right] = 1 - n \epsilon_n.
\end{align*}
Now, on the event $E$, it is clear that the MLE $\pilot = \Bern(0)$.
Let us denote the split-sample universal set by $C_{\alpha}(X_1,\ldots,X_n)$, where we assume for simplicity that $\mathcal{D}_0$ and $\mathcal{D}_1$ each have $n/2$ samples.
We then have,
\begin{align*}
\bbP_{\truedist}(\projdist \notin C_{\alpha}(X_1,\ldots,X_n)| E) 
= \bbP_{\truedist}(\mathcal{L}_0(\projdist)/\mathcal{L}_0(\pilot) \leq \alpha | E) 
= \bbP_{\truedist}(1/2^{n/2} \leq \alpha | E) = 1, 
\end{align*}
for $n \geq 2 \log_2(1/\alpha)$. As a consequence, we can upper bound the coverage of the universal set by,
\begin{align*}
\bbP_{\truedist}(\projdist \notin C_{\alpha}(X_1,\ldots,X_n)) 
\geq \bbP_{\truedist}(E) \bbP_{\truedist}(\projdist \notin C_{\alpha} | E) \geq 1 - n \epsilon_n.
\end{align*}
Thus, we see that if $0 < \epsilon_n \leq \beta/n$ for some $\beta > 0$, and $n \geq 2 \log_2(1/\alpha)$ then the universal set has coverage at most $ \beta$. Choosing $\beta < (1 - \alpha)$ 
we see that the universal set fails to have its advertised coverage.

\end{proof}

\subsection{Example~\ref{ex:normal_contam}}
\label{pf:ex:normal_contam}

{
\begin{proof}
Assume $a_n >0$ without loss of generality.
We can see that the KL projection $\projtheta = \epsilon_n a_n$ using the first-order optimality condition,
\begin{align*}
  \frac{\partial}{\partial \theta} \KL (\truedist \| P_\theta) 
    = \int \frac{\partial}{\partial \theta} \log \left( \frac{p^* (x)}{p_\theta (x)} \right) \d \truedist (x) 
    \propto \int  (x - \theta) \d \truedist = \epsilon_n a_n - \theta \overset{\text{set}}{=} 0
\end{align*}
where $\log p_\theta (x) = - \log (2\pi) / 2 - (x-\theta)^2/2$. 
For simplicity we suppose that $n$ is even, and that $\Data_0$ consists of the first $n/2$ samples and $\Data_1$ consists of the remaining samples. Following the similar argument above, we can see that MLE using $\Data_1$ is the sample mean $\thetahat_1 = \overline{X}_{n_1}$ under Gaussian location family.
For positive constants $K_0, K_1, K_2 >0$, let us consider the events $E_0, E_1$ defined as,
\begin{align*}
    E_0 &= \I \left( \sum_{i: \Data_0} X_i < K_0 \sqrt{n}/2 \right), \\ 
    E_1 &= \I \left( K_1 /2 < \sum_{i: \Data_1} X_i < \projtheta n /2 - K_2 \sqrt{n/2}  \right).
\end{align*}
We will show that the unversal set $C_\alpha (X_1,\ldots,X_n)$ fails to cover $\projdist$ when the events $E_0$ and $E_1$ hold.
On the event $E_0$ and $E_1$, $\overline{X}_{n_0} < K_0/\sqrt{n}$ and $K_1/n < \thetahat_1 < \projtheta - K_2 \sqrt{2/n}$ holds, respectively.
Using the log likelihood ratio of Gaussian location family,
\begin{align*}
   \log (\calL_0 (\thetahat_1) / \calL_0 (\projtheta) ) / n_0 
    =  (\thetahat_1 - \projtheta) \overline{X}_{n_0} - (\thetahat_1^2 - \projtheta^2 )/2,
    =  (\thetahat_1 - \projtheta) \left[\overline{X}_{n_0} - (\thetahat_1 + \projtheta)/2\right],    
\end{align*}
and $\thetahat_1 < \projtheta $, we have that
\begin{align*}
    \bbP_{\truedist} (\projdist \notin C_\alpha | E_0, E_1) 
    &= \bbP_{\truedist} \left( \log (\calL_0 (\thetahat_1) / \calL_0 (\projtheta)) > \log(1/\alpha)  \middle| E_0, E_1 \right) \\
    &= \bbP_{\truedist} \left( \overline{X}_{n_0}  < \frac{(\thetahat_1 + \projtheta)}{2} - \frac{2 \log(1/\alpha)}{n (\projtheta - \thetahat_1)} \middle| E_0, E_1 \right)\\    
    &\geq  \bbP_{\truedist} \left( \overline{X}_{n_0}  < \frac{(K_1/n + c_2 n)}{2} - \frac{\sqrt{2} \log(1/\alpha)}{K_2 \sqrt{n}} \middle| E_0, E_1 \right)
    = 1,
\end{align*}
provided that $n > C(K_0, K_1, c_2, \alpha)$ where $C(K_0, K_1, c_2, \alpha)$ is the largest solution of $K_1/ \sqrt{n} + c_2 n\sqrt{n} = 2(K_0 + \sqrt{2} \log(1/\alpha) / K_2) $.
Using the fact that the Total Variation distance between n-fold product measures,
\begin{align*}
    \TV (N(0,1)^n, {\truedist}^n) 
        \leq n \TV (N(0,1), \truedist) 
        &= n \epsilon \TV ( N(0,1), Q_a) \leq n \epsilon_n \leq c_1, 
\end{align*}
we can reason instead about the probability of the events $E_0$ and $E_1$ when drawing i.i.d. samples from standarad Gaussian distribution, and account for the difference using the Total Variation. In more detail, note that
\begin{align*}
    \bbP_{X_i \sim N(0,1)} (E_0)
        &=  \bbP( Z < K_0 / \sqrt{2} ),\\
    \bbP_{X_i \sim N(0,1)} (E_1) 
        &= \bbP ( K_1/ \sqrt{n}  < Z < \sqrt{n/2} \projtheta -  K_2 )\\
        &\geq \bbP (K_1/ \sqrt{n} < Z < c_2 n\sqrt{n/2} - K_2),
\end{align*}
where $Z\sim N(0,1)$.
Combining the results, we have that,
\begin{align*}
    \bbP_{\truedist} (E_0 \cup E_1) 
        &\geq \bbP (Z < K_0 /\sqrt{2}) \bbP (K_1/ \sqrt{n} < Z < c_2 n\sqrt{n/2} - K_2)  - c_1.
\end{align*}
Thus choosing $c_1, K_1, K_2$ to be sufficiently small positive constants, when $n$ is sufficiently large, we conclude that
\begin{align*}
    \bbP_{\truedist} (\projdist \notin C_\alpha ) = \bbP_{\truedist} (\projdist \notin C_\alpha  | E_0, E_1) \bbP_{\truedist} (E_0 \cup E_1) \geq 1/8. 
\end{align*}
\end{proof}
}

\subsection{Example~\ref{ex: sLRT_fail}}
\label{pf:Ex_sLRT_fail}

\begin{proof}

The KL projection $\projdist$ is $\Bern\left( 3/4 \right)$. For simplicity we suppose that $n$ is even, and that $\mathcal{D}_0$ consists of the first $n/2$ samples and $\mathcal{D}_1$ consists of the remaining samples.
For a constant $\beta > 0$, let us consider the events $E_0, E_1$ defined as,
\begin{align*}
E_0 &= \I \left( \sum_{i=1}^{n/2} X_i < n/4 - \beta \sqrt{n} \right) \\
E_1 &=\I \left( \sum_{i=n/2}^n X_i < n/4 \right).
\end{align*}
When events $E_0$ and $E_1$ hold we can see that the universal set $C_{\alpha}(X_1,\ldots,X_n)$ fails to cover $\projdist$. In more detail, on the event $E_1$ the MLE, $\pilot$ is $\Bern(1/4)$ and thus,
\begin{align*}
\bbP_{\truedist}(\projdist \notin C_{\alpha}(X_1,\ldots,X_n) | E_0, E_1) &= \bbP_{\truedist}(\mathcal{L}_0(\projdist)/\mathcal{L}_0(\pilot) \leq \alpha | E_0, E_1) \\
&\leq \bbP_{\truedist}(1/3^{2\beta \sqrt{n}} \leq \alpha) = 1,
\end{align*}
provided that $n \geq (\log_3(1/\alpha))^2/(4\beta^2)$. Thus, it suffices to show that $E_0$ and $E_1$ happen with sufficiently large probability. 

Using the fact that the Total Variation distance between the n-fold product measures, 
\begin{align*}
\TV (\Bern(1/2)^n, \Bern(1/2 + \epsilon_n)^n) \leq n \epsilon_n, 
\end{align*}
we can reason instead about the probability of the events $E_0$ and $E_1$ when drawing samples from $\Bern(1/2)$, and account for the difference using the Total Variation. Combining this fact with the standard
Berry-Esseen bound applied to Bernoulli sums, together with some simple algebraic manipulations, we obtain that for some universal constant $C > 0$,
\begin{align*}
P(E_0 \cup E_1) \geq P(Z < 0) \times P(Z < -2\sqrt{2}\beta) - \frac{2 C}{\sqrt{n}} - n \epsilon_n. 
\end{align*}
Thus, choosing $\epsilon_n \ll 1/n$, and $\beta$ to be a sufficiently small constant, when $n$ is sufficiently large, we obtain that,
\begin{align*}
P(E_0 \cup E_1) \geq \frac{1}{8},
\end{align*}
and thus that,
\begin{gather*}
P(\projdist \notin C_{\alpha}(X_1,\ldots,X_n)) \geq 1/8. \qedhere
\end{gather*}
\end{proof}

\section{Proofs from Section~\ref{sec:rui}}
\label{app:secfour}
In this section, we formally verify the claim that the universal set typically includes both the pilot and the projection distribution. 
We first define the $\rho$-diameter of the set $C$ as ${\textstyle\diam_{\rho}} (C) = \sup_{P_a, P_b \in C} \rho(P_a \| P_b)$.

\begin{proposition}
\label{prop:min_diam_exact}
    Let $\pilot \in \calP$ be any estimate of $\projdist$ based on $\Data_1$. Then, the exact robust universal confidence set $C_{\alpha,n}$ defined in \eqref{eq:exact_rob_uni_set} has diameter at least $\rho(\projdist \| \pilot)$ with 
  $\truedist$-probability at least $1 - 2\alpha$:
    \begin{align*}
    \inf_{\truedist\in\trueclass} \bbP_{\truedist} \left({\textstyle\diam_{\rho}} (C_{\alpha,n}) \ge \rho(\projdist \|\pilot) \right) \ge 1 - 2 \alpha.
  \end{align*}
\end{proposition}
\begin{proof}
   Following a similar argument to that in the proof of Theorem~\ref{thm:honest_exact}, notice that for any $\truedist\in\trueclass$, $\bbP_{\truedist} (\pilot \notin C_{\alpha,n} ~|~ \Data_1) \le \alpha$. Together with a union bound, we obtain that 
   both $\projdist$ and $\pilot$ are included in the set $C_{\alpha,n}$ with $\truedist$-probability at least $1- 2\alpha$ (conditionally on $\Data_1$), and on this event, the diameter of the set is at least $\rho(\projdist \| \pilot)$ in expectation. 
\end{proof}

\section{Proofs from Section~\ref{sec:tests}}
\label{sec:proof_section_test}

\subsection{Proof of Theorem~\ref{thm:Rift_validity}}
\begin{proof}
We first work conditional on the sample $\mathcal{D}_1$ used to construct the pilot estimate $\pilot$. 
Let us define 
\begin{align*}
M_{P,\delta,\xi} := \mathbb{E}_{\truedist} \left[ |T_i(P,\pilot) + Z_i \delta - \mathbb{E}_{\truedist} T(P,\pilot)|^{2+\xi} \;\middle|\; \mathcal{D}_1 \right].
\end{align*}
Due to the added Gaussian noise{,} the variance $M_{P,\delta,0}$ is always strictly positive (i.e., larger than $\delta^2$). 
By Minkowski's inequality, conditional on $\Data_1$, we have
\begin{align*}
    M_{P,\delta,\xi}
    \le \left[ \left(\bbE_{\truedist} \left[ \left| T_i (P, \pilot) - \bbE_{\truedist} T_i (P, \pilot)\right|^{2+\xi} \;\middle|\; \Data_1 \right]\right)^{\frac{1}{2+\xi}}
    + \delta \left(\bbE |Z|^{2+\xi}\right)^{\frac{1}{2+\xi}} \right]^{2+\xi}.
\end{align*}
This means that for the assumed $\trueclass$, there exists a universal constant $C_M$ such that (conditionally on $\Data_1$) the $2+\xi$ moment of corrupted statistic $T_i (P, \pilot) + \delta Z_i$ is uniformly bounded by $C_M$ for all $P\in \model$. {This in turn implies that that unconditional expectation of $M_{P,\delta,\xi}$ is also upper bounded as $\mathbb{E}[M_{P,\delta,\xi}] = \mathbb{E}[\mathbb{E}[M_{P,\delta,\xi} | \mathcal{D}_1]] \leq C_M.$}

Conditionally on $\Data_1$, the generalized Berry-Esseen bound
for the studentized statistic \citep{bentkus_berry-esseen_1996,petrov_sums_1975} yields that, {for a universal constant $C'$},
\begin{align*}
\sup_{t} \left| \mathbb{P}_{\truedist} \left(\sqrt{n_0} \frac{\left( \overline{T}_{n_0,\delta} (P,\pilot) - \bbE_{\truedist} T (P, \pilot) \right)}{\widehat{s}_{P,\delta}} \geq t ~\middle|~ \mathcal{D}_1\right) - P(Z \geq t)\right|
 &\leq \frac{C' M_{P,\delta,\xi}}{n_0^{\xi/2} \delta^{2+\xi}}, \\
 &\leq C n_0^{-\xi/2},
\end{align*}
where $C = C' C_M \delta^{-(2+\xi)}$.

This holds in particular for $\projdist \in \mathcal{P}$. 
Consequently, we see that,
\begin{align*}
\inf_{\truedist \in \trueclass} \bbP_{\truedist} (\projdist \in \Chat_{\KL\Redi, n}) &= \inf_{\truedist \in \trueclass} \mathbb{E}_{\truedist} [  \bbP_{\truedist} (\projdist \in \Chat_{\KL\Redi, n} ~|~ \mathcal{D}_1)] \\
&\geq 1 - \sup_{\truedist \in \trueclass} \mathbb{E}_{\truedist} [  \bbP_{\truedist} (\projdist \notin \Chat_{\KL\Redi, n} ~|~ \mathcal{D}_1)] \\
&\geq 1 - [\alpha - C n^{-\xi /2}],
\end{align*}
as claimed.
\end{proof}

\subsection{Proof of Proposition~\ref{prop:DP_ex1}}\label{sec:DP_ex1_proof}

\begin{proof}
Recall that $X_i \overset{iid}{\sim} \Bern(\epsilon_n)$ for $\epsilon_n \le (1-\alpha)/n$ and our hypothesized model is $\model=\{\Bern(p): p\in\{0, 1/2\}\}$. For a fixed $\beta >0$,
\begin{align*}
  \DP_\beta(\Bern(\epsilon_n) \| \Bern(p) )
  &= C + (p^{1+\beta} + (1-p)^{1+\beta}) - (1 + 1/\beta) [\epsilon_n p^\beta + (1-\epsilon_n) (1-p)^{\beta}]   
\end{align*}
where $C = \sum_{x\in{0,1}} \epsilon_n^{(1+\beta)x} (1-\epsilon_n)^{(1+\beta)(1-x)}$. The DP divergences from $\truedist$ to the elements of the working model are
\begin{align*}
  \DP_\beta(\Bern(\epsilon_n) \| \Bern(0) )
  &\propto 1 - (1 + 1/\beta) (1-\epsilon_n)
  = (1 + 1/\beta) \epsilon_n - 1/\beta\\
  \DP_\beta(\Bern(\epsilon_n) \| \Bern(1/2) )
  &\propto - (1/2)^\beta / \beta.
\end{align*}
Therefore, the DP projection is
\begin{align*}
  \projdist =
\begin{cases}
    \Bern(0),& \text{if } \epsilon_n \le (1 -(1/2)^{\beta}) / (1 + \beta),\\
   \Bern(1/2),& \text{otherwise}.
\end{cases}
\end{align*}

Since $\epsilon_n < (1-\alpha)/n$, the projection will be $\Bern(0)$ for any $\beta >0$, provided that $n \ge (1-\alpha) (1 + \beta) / (1- (1/2)^\beta)$.

Now we turn our attention to constructing the DP \Redi set. For any fixed $(P,Q) \in \model^2$, the split statistic is uniformly bounded, i.e., $|T_i (P,Q) - \bbE_{\truedist} T (P,Q)| \le 1 + 1/\beta$ since 
\begin{align*}
  T_i (P, Q) 
    =& \sum_{x\in \{0,1\}} \left[ (p^x (1-p)^{1-x})^{1+\beta} - (q^x (1-q)^x)^{1+\beta}  \right]\\
     &- \left( 1+\frac{1}{\beta} \right) \left[  (p^{X_i} (1-p)^{1-X_i})^{\beta} - q^{X_i} (1-q)^{1-X_i} \right] (X_i).
\end{align*}
By Hoeffding's inequality, $\Chat_{\DP\Redi,1+1/\beta,n}$ ensures nominal coverage for any estimator $\pilot$, since we have that:
\begin{align*}
  \bbP_{\truedist} \left(\projdist \notin \Chat_{\DP\Redi,1+1/\beta,n} \right)
  = \bbE_{\truedist} \left(\bbP_{\truedist} \left( \overline{T}_{n_0} (\projdist,\pilot) > \frac{\beta + 1}{\beta} \sqrt{\frac{\log(1/\alpha)}{2 n_0}} ~\middle|~\Data_1 \right)\right) \le \alpha.
\end{align*}
\end{proof}

\subsection{Proof of Proposition~\ref{prop:HEL_ex1} }
\label{sec:H_ex1_proof}

\begin{proof}
Note that Hellinger projection is $\Bern(0)$ for $n>6$ (as long as $\epsilon_n < 0.146$) since
\begin{align*}
  \HEL^2(\Bern(\epsilon_n), \Bern(0)) 
    &= 1 - \sqrt{1 - \epsilon_n}, \\
  \HEL^2(\Bern(\epsilon_n), \Bern(1/2)) 
    &= 1 - \sqrt{\epsilon_n / 2} - \sqrt{(1-\epsilon_n) / 2}.
\end{align*}

Similarly, $\nu$-approximate Hellinger projection is
$\Bern(0)$ if $\epsilon_n < 0.051$ or $\model$ otherwise. Hereafter we only consider $n > 20$ where $\projdist_\nu = \projdist$.

The minimum Hellinger Distance Estimator (MHDE) is
\begin{align*}
  \argmin_{p\in\{0, 1/2\}} \HEL^2 (\bbP_{n_1}, \Bern(p)) 
  =  \argmax_{p\in \{0, 1/2\}} \sqrt{p \overline{X}_{n_1} } + \sqrt{(1-p) (1 - \overline{X}_{n_1})}
\end{align*}
where $\overline{X}_{n_1} = \sum_{i\in\calI_1} X_i / n_1$.
Thus,
\begin{align*}
  \pilot = \begin{cases}
    \Bern(0), & \overline{X}_{n_1} < 0.5-1/(2\sqrt{2}) \approx 0.146\\
    \Bern(1/2), & \text{Otherwise.}\\    
  \end{cases}
\end{align*}

This implies that the advertised coverage is guaranteed when $\overline{X}_{n_1} < 0.146$. Otherwise, Corollary~\ref{thm:HELRedi_validity} ensures the asymptotic (approximate) validity.
\end{proof}

\subsection{Proof of Proposition~\ref{prop: IPM_gamma} }

\begin{proof} The proof follows directly by the triangle inequality.
\begin{align*}
2 \bbE_{\truedist} T (P_0, P_1)
    &=  \bbE_{P_0} f^*_{(P_0, P_1)} + \bbE_{P_1} f^*_{(P_0, P_1)} - 2 \bbE_{\truedist} f^*_{(P_0, P_1)}\\
    &= 2 \left[ \bbE_{\truedist} f^*_{(P_0, P_1)} - \bbE_{\truedist} f^*_{(P_0, P_1)} \right] - \bbE_{P} f^*_{(P_0, P_1)} - \bbE_{P_1} f^*_{(P_0, P_1)}\\ 
    &= 2 \left[ \bbE_{\truedist} f^*_{(P_1, P_0)} - \bbE_{P} f^*_{(P_1, P_0)} \right] - \IPM_\calF (P_0, P_1) \\     
    &\le 2 \IPM_\calF (\truedist, P_0) - \IPM_\calF (P_0, P_1)\\
    &\le 2 \IPM_\calF (\truedist, P_0) - \left[\IPM_\calF(\truedist, P_1) - \IPM_\calF(\truedist, P_0)\right]\\
    &= 3 \IPM_\calF(\truedist, P_0) - \IPM_\calF(\truedist, P_1) &
\end{align*}
where the last inequality is by the triangle inequality.
\end{proof}

\section{Proofs from Section~\ref{sec:size_set}}
\label{sec:proof_section_size}

\subsection{Proof of Theorem~\ref{thm:exact_optimal_shrink}}
\label{sec: proof_delta_bound}
Recall that the exact robust universal confidence set based on the Hoeffding bound is
\begin{align*}
    \Chat_{\rho\Redi,\sigma,n} = \left\{ P\in \model : \overline{T}_{n_0} (P,\pilot) \le B \sqrt{\frac{ \log (1/\alpha)}{2 n_0}  } \right\}.
\end{align*}
We denote $t_{\alpha,n} := B \sqrt{\frac{ \log (1/\alpha)}{2 n_0}}$ throughout the proof, and use $\widehat{C}$ to denote $\Chat_{\rho\Redi,B,n}$.  
Throughout the proof, we fix a projection distribution $\projdist$ and assume an equal split between $\Data_0$ and $\Data_1$.

\begin{proof} [Proof of Theorem~\ref{thm:exact_optimal_shrink}]
Denote $\delta_\nu (P, Q) = \rho(\truedist \| P) - \nu \rho(\truedist \| Q)$ for any $P, Q \in \model$. 
We want to show that, for fixed $\kappa > 0$, for some finite $M > 0$,
\begin{align*}
  \bbP_{\truedist} \left( \sup_{P \in \Chat} \delta_\nu(P, \projdist) > M  (\eta_n + \tilde{\epsilon}_n) \right) \le \kappa,
\end{align*}
where $\tilde{\epsilon}_n = \mathfrak{R}_n(\calF_{T,\mathcal{P}}) \vee t_{\alpha,n}$ and for all $n$ large enough. 

Let the event $E$ be $\delta_1 (\pilot, \projdist) \le (M/\nu) \eta_n$ which happens with probability at least $1-\kappa/2$. Then,
\begin{align}
  &\bbP_{\truedist} \left( \sup_{P \in \Chat} \delta_\nu(P, \projdist) > M  (\eta_n + \tilde{\epsilon}_n) \right)\nonumber\\
  &\le \bbP_{\truedist} \left( \sup_{P \in \Chat} \delta_\nu(P, \projdist) > M  (\eta_n + \tilde{\epsilon}_n) ~\middle|~ E \right) + \frac{\kappa}{2} \nonumber\\
  &= \bbP_{\truedist} \left( \sup_{P \in \Chat} \delta_\nu(P, \pilot) > M  (\eta_n + \tilde{\epsilon}_n) - \nu \delta_1(\pilot, \projdist) ~\middle|~ E \right) + \frac{\kappa}{2} \nonumber\\
  &\le \bbP_{\truedist} \left( \sup_{P \in \Chat} \delta_\nu(P, \pilot) > M  \tilde{\epsilon}_n ~\middle|~ E \right) + \frac{\kappa}{2}.  \label{eq:want}
\end{align}
Thus, it suffices to show that conditional on $E$, with $\truedist$-probability at most $\kappa/2$, all $P\in \model$ such that $\delta_\nu(P, \pilot) > M \tilde{\epsilon}_n$ are not included in $\Chat$. Hereafter we condition on event $E$.

Let $\calS_\epsilon := \{P \in \calP : \delta_\nu(P, \pilot) > {c_1} \epsilon\}$.
From Assumption~\ref{ass:T_antisym} and ~\ref{ass:T_ub}, we have that
\begin{align*}
     &\bbP_{\truedist} \left( \forall P\in \calS_\epsilon, P \in   \Chat_{\rho\Redi,B,n} ~\middle|~ \Data_1 \right) \\
     &= \bbP_{\truedist} \left( \sup_{P\in \calS_\epsilon} \overline{T}_{n_0}(\pilot, P)  \ge - t_{\alpha,n}  ~\middle|~ \Data_1 \right),\\
     &\le \bbP_{\truedist} \left( \sup_{P\in\calS_\epsilon} \left[\overline{T}_{n_0}(\pilot, P) - \bbE_{\truedist} T(\pilot, P)\right]  \ge \epsilon - t_{\alpha,n}  ~\middle|~ \Data_1 \right).
\end{align*}
where the inequality is by Assumption~\ref{ass:T_ub} noticing that conditional on $\Data_1$, we have that
{
\begin{align*}
    \sup_{P \in \calS_{\epsilon}} [ - \bbE_{\truedist} T(\pilot, P)]
        \geq c_1^{-1} \sup_{P \in \calS_{\epsilon}} \left( - \nu \rho (\truedist \| \pilot) + \rho (\truedist \| P) \right)  
        = c_1^{-1}\sup_{P \in \calS_{\epsilon}} \delta_\nu (P, \pilot) 
        \geq \epsilon.
\end{align*}
}
To ease the notation, denote the centered statistic as $\widetilde{T}_P := \overline{T}_{n_0}(\pilot, P) - \bbE_{\truedist} T(\pilot, P)$.
Since $|T(\pilot,P)| \le B$, any change in $X_i$ can change $\sup_{P\in \calS_\epsilon}\widetilde{T}_P$ at most $2B/ n_0$. By McDiarmid's inequality, we have that
\begin{align}\label{eq:mcdiarmid}
  \bbP_{\truedist} \left(\sup_{P\in \calS_\epsilon} \widetilde{T}_P \ge  \epsilon - t_{\alpha,n} \;\middle|\; \Data_1\right) 
  \le \exp\left( - \frac{n\left(\epsilon - t_{\alpha,n} -  \bbE_{\truedist} \left[\sup_{P\in \calS_\epsilon} \widetilde{T}_P\right] \right)^2}{2 B^2}\right).
\end{align}

Now we focus on bounding $\bbE_{\truedist} [\sup_{P\in \calS_\epsilon} |\widetilde{T}_P|]$ (which is greater than $\bbE_{\truedist} [\sup_{P\in \calS_\epsilon} \widetilde{T}_P ]$). 
Let $\calF_{T,\mathcal{P}} = \{T(\cdot; \pilot, P) : P \in \model\}$. The symmetrization lemma \citep[Lemma~2.3.1,][]{vaart_weak_1996} states that
\begin{align*}
  \bbE_{X} \sup_{f\in \calF_{T,\mathcal{P}}} \frac{1}{n_0} \left| \sum_{i=1}^{n_0}[f (X_i) - \bbE_\truedist f(X_i)] \right| 
  \le  2 \bbE_{X,\varepsilon} \sup_{f\in \calF_{T,\mathcal{P}}} \left|\frac{1}{n_0} \sum_{i=1}^{n_0} R_i f(X_i)\right| := 2 \Re_{n_0} (\calF_{T,\mathcal{P}})
\end{align*}
where $R_i$ are iid Rademacher random variables.
\end{proof}

\section{Finite-sample valid confidence set for bounded test statistic}
\label{sec:finite_set}
\label{sec: finite_sample_bound}

Suppose the split statistics are uniformly bounded, i.e., $|T_i (P)| \le B$ for all $i$. Classic Cram{\'e}r-Chernoff bounds yield finite-sample valid exact $(\nu = 1)$ or approximate $(\nu > 1)$ confidence sets.
\begin{proposition}[Hoeffding \citep{hoeffding_probability_1963}] 
\label{prop:fs}
$\Chat_{\textsc{HF}}$ is a uniformly valid $1-\alpha$ exact $(\nu = 1)$ or approximate $(\nu > 1)$ confidence set for $\projdist$ where 
\begin{align}\label{eq: ConfSet_HEL_Hoeffding}
  \Chat_{\textsc{HF}} = \left\{P \in \model : \overline{T}_{n_0} (P) \le \sqrt{\frac{B^2}{2n_0} \log \left(\frac{1}{\alpha}\right)} \right\}.
\end{align}
 
\end{proposition}
Typically, Hoeffding's bound does not scale with the variance which results in a conservative confidence set. Confidence set based on Bernstein's inequality is given as follows.
\begin{proposition}[Bernstein]
$\Chat_{\textsc{BS}}$ is a uniformly valid $1-\alpha$ exact $(\nu = 1)$ or approximate $(\nu > 1)$ confidence set for $\projdist$ where 
\begin{align}\label{eq: ConfSet_HEL_Bernstein}
  \Chat_{\textsc{BS}} = \left\{ P\in\model : \overline{T}_{n_0}(P) \le \sqrt{\frac{2 S^2 \log (1/\alpha)}{n_0} + \frac{B^2}{9} \left( \frac{\log (1/\alpha)}{n_0} \right)^2} + \frac{B \log (1/\alpha)}{3 n_0} \right\}
\end{align}
where $S^2 = S^2(P) =  (c_1 \nu)^2 [\rho (\truedist \| P) + \rho (\truedist \| \pilot)]$.
\end{proposition}

However, $\Chat_{\textsc{BS}}$ above requires knowledge of $\truedist$ to compute $S$. Empirical Bernstein bounds \citep{maurer_empirical_2009,waudby-smith_estimating_2023} address this issue.
\begin{proposition}[Empirical Bernstein \citep{waudby-smith_estimating_2023}] Denote $\tilde{T}_i (P,Q) = (T (X_i; P, Q)+ B) / (2B)$. $\Chat_{\textsc{EBS}}$ is a valid $1-\alpha$ confidence set for exact $(\nu = 1)$ or approximate $(\nu > 1)$ projection where
\begin{align}\label{eq: ConfSet_HEL_emp_Bernstein}
  \Chat_{\textsc{EBS}} = \left\{ P\in\model : \sum_{i=1}^{n_0} \lambda_i \tilde{T}_i (P, \pilot)\le \log(1/\alpha) + \sum_{i=1}^{n_0} v_i \psi_E (\lambda_i) \right\},
\end{align}  
$v_i = (\tilde{T}_i (P, \pilot) - \overline{\tilde{T}}_{i - 1} (P, \pilot))^2$, $\psi_E(\lambda) = - (\log(1-\lambda) - \lambda) $, and 
\begin{align*}
    \lambda_i = \sqrt{\frac{2\log(1/\alpha)}{n_0 \hat{S}_{i - 1}^2}} \wedge c, \qquad
    \hat{S}_{i}^2 = \frac{1/4 + \sum_{l=1}^i (\tilde{T}_l - \overline{\tilde{T}}_l)^2}{i + 1}, \qquad
    \overline{\tilde{T}}_i = \frac{1}{i+1} \sum_{l=1}^i \tilde{T}_l,
\end{align*}
for some $c \in (0,1).$
\end{proposition}

When the variance or an upper bound of the variance is known, Bentkus's bound \citep{bentkus_domination_2006} is sharper than any Cram{\'e}r-Chernoff type bounds. See \citet{pinelis_optimal_2013,kuchibhotla_near-optimal_2021} for details. Define a Bernoulli random variable $G = G(S^2, B)$ as
\begin{align*}
  \bbP \left( G = B  \right) = \frac{S^2}{S^2 + B^2} := p_{SB}, \qquad
  \bbP \left( G = - \frac{S^2}{B}  \right) = 1 - p_{SB}
\end{align*}

\begin{proposition}[Bentkus \citep{bentkus_domination_2006}] $\Chat_{\textsc{BK}}$ is a valid $1-\alpha$ confidence set for is a valid $1-\alpha$ confidence set for exact $(\nu = 1)$ or approximate $(\nu > 1)$ projection where
\begin{align}\label{eq: ConfSet_HEL_Bentkus}
  \Chat_{\textsc{BK}} = \left\{ P\in\model : \overline{T}_{n_0}(P{,\pilot}) \le q(\alpha) \right\}
\end{align}
where $q(\alpha)$ is the solution to 
\begin{align*}
  P_2 \left( u; \sum_{i\in \calI_0} G_i \right)
   := \inf_{t \le u} \frac{\bbE_{\truedist} \left( \sum_{i\in \calI_0} G_i - t  \right)_+^2}{(u -t )_+^2}  = \alpha,
\end{align*}
and $S^2 = S^2(P) =  (c_1 \nu)^2 [\rho (\truedist \| P) + \rho (\truedist \| \pilot)]$.
\end{proposition}

As in the case of Bernstein's bound~\eqref{eq: ConfSet_HEL_Bernstein}, Bentkus's bound~\eqref{eq: ConfSet_HEL_Bentkus} requires prior knowledge of $\truedist$ to compute the variance $S$. The empirical Bentkus's bound \citep{kuchibhotla_near-optimal_2021} addresses this by taking the union bound on variance over-estimation and the Bentkus's inequality. Following \citet[Lemma F.1,][]{kuchibhotla_near-optimal_2021} define the over-estimator of $S$ as, for $\delta\in[0,1]$,
\begin{align*}
  \widehat{S}_n (\delta) = \sqrt{\overline{S}_{n_0}^2 + g_{2,n_0} (\delta)} + g_{2,n_0}(\delta),\qquad
  \overline{S}_n^2 = \frac{1}{\lfloor n / 2 \rfloor} \sum_{i=1}^{\lfloor n / 2 \rfloor} \frac{(T_{2i} - T_{2i-1})^2}{2},
\end{align*}
where $g_{2,n}(\delta) := B (\sqrt{2} n)^{-1} \sqrt{\lfloor n / 2 \rfloor} \Phi^{-1} (1- 2 \delta / e^2)$ and $\Phi$ is the cdf of a standard Gaussian.

\begin{proposition}[Empirical Bentkus \citep{kuchibhotla_near-optimal_2021}] $\Chat_{\textsc{EBK}}$ is a valid $1-\alpha$ confidence set for is a valid $1-\alpha$ confidence set for exact $(\nu = 1)$ or approximate $(\nu > 1)$ projection where for some $\delta\in[0,1]$,
\begin{align}\label{eq: ConfSet_HEL_EBentkus}
  \Chat_{\textsc{EBK}} = \left\{ P\in\model : \overline{T}_{n_0}(P{,\pilot}) \le q(\alpha - \delta) \right\}
\end{align}
where $q(\alpha - \delta)$ is the solution to 
\begin{align*}
  P_2 \left( u; \sum_{i\in \calI_0} G_i \left( \widehat{S}^2_*(\delta), B \right) \right) = \alpha - \delta.
\end{align*}
with $\widehat{S}_* (\delta) := \min_{1 \le i \le n_0} \widehat{S}_i (\delta)$.
\end{proposition}
In Section~\ref{sec:overdisp_finite_comparision}, we choose $\delta = \alpha/3$ to construct the empirical Bentkus's bound-based TV \Redi set.

\section{Crossfit \Redi set}\label{sec:crossfit}
Despite universal inference holds for any $\pilot$, let us assume we choose $\pilot$ such that $\sup_{P \in \model} \|T(\cdot; P, \widehat{P}_1) - T(\cdot; P, \projdist)\|_{L_2(\truedist)} = o(1)$.
For any fixed $P \in \model$, consider the following decomposition:
\begin{align*}
  \bbP_{n_0} T(\cdot; P, \widehat{P}_1) - \bbP_{\truedist} T(\cdot; P, \projdist)
  =& (\bbP_{n_0} - \bbP_{\truedist}) [T(\cdot; P, \widehat{P}_1) - T(\cdot; P, \projdist)] \\
  &+ \bbP_{\truedist} \left[T(\cdot; P, \widehat{P}_1) -  T(\cdot; P, \projdist)\right] + (\bbP_{n_0} - \bbP_{\truedist}) T(\cdot; P, \projdist).
\end{align*}
The first term is the empirical process which is $o_{\truedist} (1/\sqrt{n_0})$ applying Lemma~2 of \citet{kennedy_sharp_2020}. The second part is the bias which is $o(1)$ from our choice of $\pilot$. The last term yields the CLT.

Now let $\overline{T}_{n_1} (P; \widehat{P}_0) := \sum_{i\in\calI_1} T(X_i; P, \widehat{P}_0) /n_1$ where we change the role of $\Data_0$ and $\Data_1$. Define a cross-fitting estimator as 
\begin{align*}
    \overline{T}_{n}^{\times} (P) 
        = \frac{n_1 \overline{T}_{n_1} (P; \widehat{P}_0) + n_0 \overline{T}_{n_0} (P; \widehat{P}_1)}{n}.
\end{align*}
The $n (\overline{T}^{\times} (P) -  \bbP_{\truedist} T(\cdot; P, \projdist))$ has the following decomposition:
\begin{align*}
  &n_0 (\bbP_{n_0} - \bbP_{\truedist}) [T(\cdot; P, \widehat{P}_1) - T(\cdot; P, \projdist)]
 + n_1 (\bbP_{n_1} - \bbP_{\truedist}) [T(\cdot; P, \widehat{P}_0) - T(\cdot; P, \projdist)]\\
 &+ n_0 \bbP_{\truedist} [T(\cdot; P, \widehat{P}_1) - T(\cdot; P, \projdist)] + n_1 \bbP_{\truedist} [T(\cdot; P, \widehat{P}_0)- T(\cdot; P, \projdist)]\\
 &+  n (\bbP_{n} - \bbP_{\truedist}) T(\cdot; P, \projdist).
\end{align*}
Similarly, both empirical process terms in the first line are $o_{\truedist} (1/\sqrt{n})$, and bias terms in the second line are $o (1)$. Thus, we left with the same CLT term. The decomposition implies that as long as one chooses a ``good'' candidate estimator, the cross-fit estimator also provides an asymptotically (uniformly) valid inference on $\projdist$.
Construct a \emph{cross-fit $\rho$-\Redi set} as follows:
\begin{align}\label{eq: ConfSet_cross}
  C^{\times}_{\rho\Redi,\alpha, n} = \left\{P\in \model : \overline{T}^{\times} (P) \le  z_{\alpha} \frac{\hat{s}_P^{\times}}{\sqrt{n}} \right\}.
\end{align}
where $\hat{s}_P^{\times 2} = [\hat{\bbV}_{\truedist} (T (X; P, \widehat{P}_1)) + \hat{\bbV}_{\truedist} (T (X; P, \widehat{P}_0) )] / 2$ is a consistent estimator of $\bbV_{\truedist} (T(X; P, \projdist))$.

{
\section{Power comparison in regular model and sampling distribution}\label{sec:power_correct_regular_model}

One of the concern is that the proposed method can be conservative. Indeed universal inference is known to be conservative when the model is regular \citep{tse_note_2022,strieder_choice_2022,dunn2022gaussian}. As our proposal extends universal inference to guard against potential misspecification, it is expected to be conservative compared to the existing methods which do not account for model misspecification and/or require strong regularity conditions.

This section empirically compares the power of our proposed methods to the existing methods in a regular model with and without model misspecification. We emphasize that this comparison is inherently unfavorable to our proposed method and we recommend using more powerful existing methods (LRT and the Wald test with robust standard errors) in settings where the statistician is certain that strong regularity conditions are met.

\subsection{Correctly Specified Case}

We compare the power of the proposed methods ($\KL$\Redi test~\eqref{eq: ConfSet_Rift} and the crossfit variant~\eqref{eq: ConfSet_cross}) to the classical LRT, split LRT, and crossfit split LRT under the Gaussian location model $\calP = \{ N(\theta, 1) : \theta\in \bbR \}$ without the model misspecification, i.e., $\truedist \in \calP$. We use Monte Carlo simulation to estimate the power: for a given value of $\theta$, we simulate $n=1000$ observations $X_1, \dots, X_n \sim \truedist = N(\theta,1)$, test whether $\theta = 0$ is in each confidence sets. We repeat this procedure 2000 times at each $\theta$ and estimate power at $\theta$ using the proportion of times that $\theta=0$ is not included in each confidence sets.

Figure~\ref{fig: normal_power} shows the estimated power of the LRTs against true $\theta$. As expected, the classical LRT has the highest power (and is optimal), followed by crossfit split LRT, crossfit $\KL$\Redi, split LRT, and $\KL$ \Redi in order. Interestingly $\KL$\Redi has higher power compared to the other split variant methods when the $\theta \le 0.1$ and lower power in high separation regime, $\theta > 0.1$. While crossfit procedure improves power for both split LRT and $\KL$\Redi in the high separation regime, $\KL$\Redi benefits even more that the split LRT and brings the power as close to split LRT. Thus whenever one can find accurate pilot estimate, it is recommended to apply crossfit to improve power.

\begin{figure}[!htb]
\centering
    \begin{subfigure}{.49\textwidth}
        \centering
        \includegraphics[trim={0 20 0 19},clip,scale=0.35]{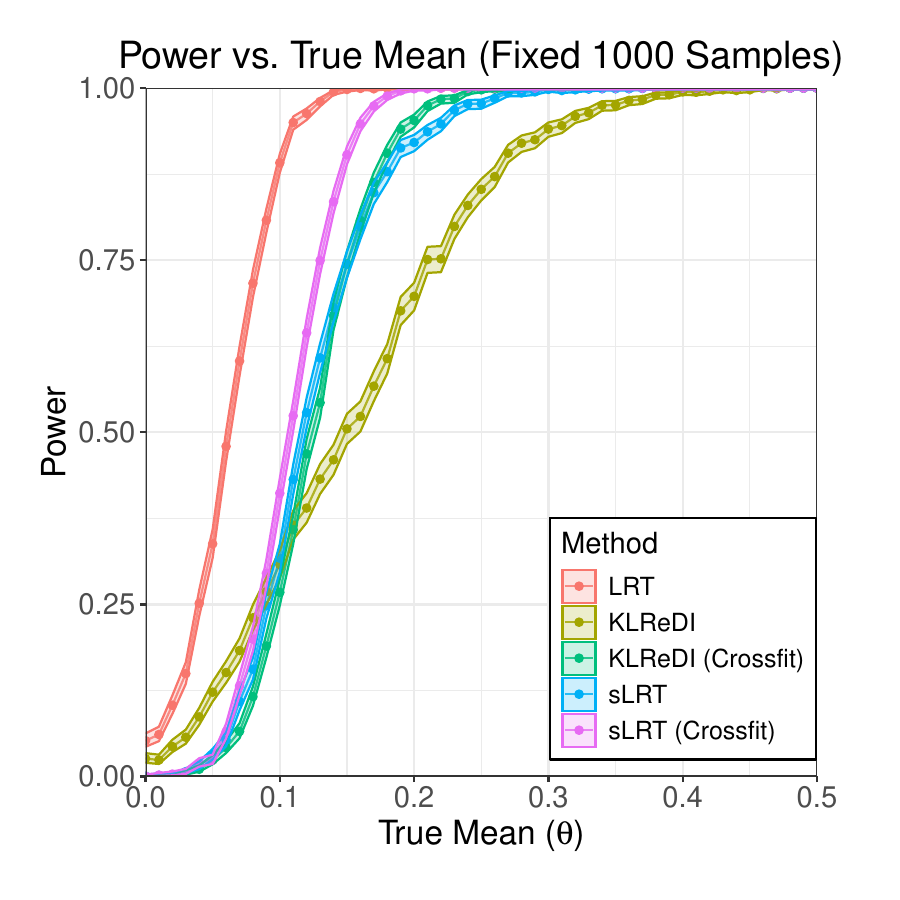}
    \end{subfigure}
    \begin{subfigure}{.49\textwidth}
        \centering
        \includegraphics[trim={0 20 0 19},clip,scale=0.35]{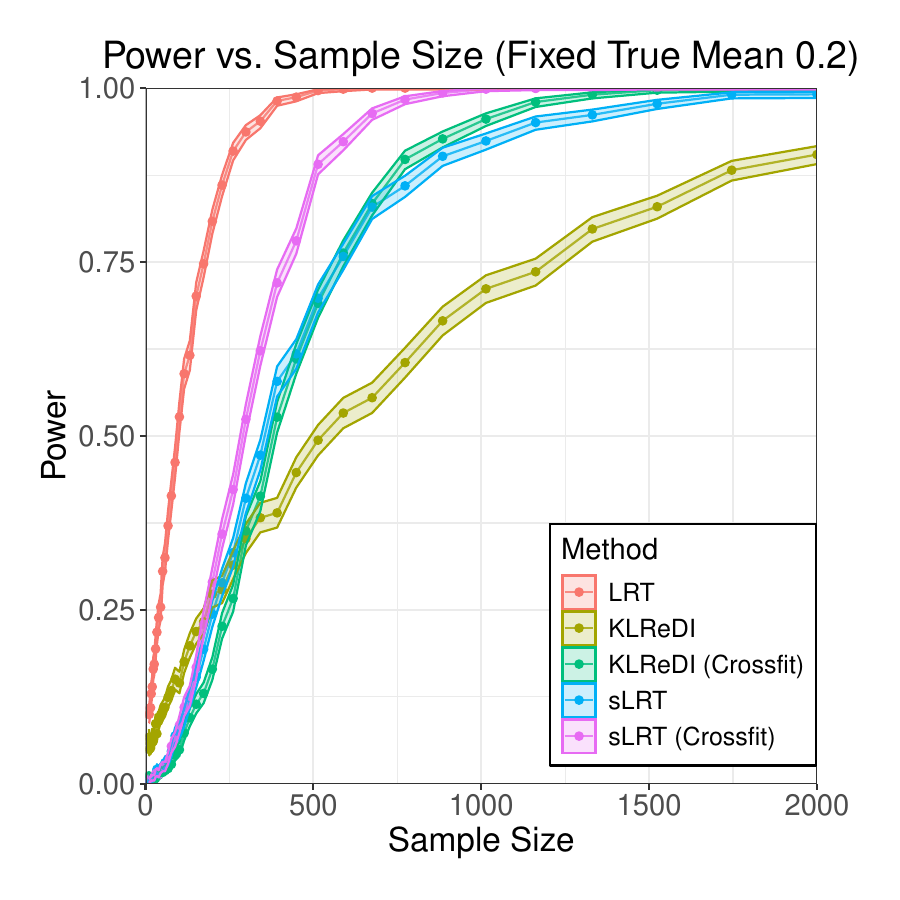}
    \end{subfigure}
    \caption{Estimated power of classical LRT, $\KL$\Redi, crossfit $\KL$\Redi, split LRT, and crossfit split LRT under Gaussian location model $N(\theta, 1)$ without model misspecification. We are testing $H_0: \theta = 0$ versus $H_1: \theta \ne 0$.}\label{fig: normal_power}    
\end{figure}

\subsection{Misspecified Case}

In this section, we compare the power of the proposed methods ($\KL$\Redi test~\eqref{eq: ConfSet_Rift} and the crossfit variant~\eqref{eq: ConfSet_cross}) to the Wald confidence set with Huber sandwich standard error (we denote as robust Wald confidence set hereafter) where the regularity conditions hold. The robust Wald confidence set is a commonly choice for statistical inference when models may be misspecified. This is because under suitable regularity conditions, MLE $\thetahat_{\text{MLE}}$ converges to KL projection $\projtheta_{\KL}$ and $\thetahat_{\text{MLE}} - \projtheta_{\KL}$ is asymptotically normal with mean 0 and Huber Sandwich covariance estimate \citep{freedman_so-called_2006,huber_behavior_1967}. Although our method offers clear advantages over the robust Wald confidence set when regularity conditions are violated (leading to singular Fisher information, as in Example~\ref{ex:mix_unident}), this comes at the cost of losing power.

Suppose the sampling distribution is contaminated distribution $\truedist = 0.95 N(\theta,1) + 0.05 N(-2, 10^2)$ for $\theta > 0$ and we consider a location family of Gaussian distribution $\model = \{ N(\mu, 1): \mu \in \bbR \}$. Then $\projtheta_{\KL} = 0.9\theta - 0.1$ following the proof of Example~3 in Section~\ref{pf:ex:normal_contam}. Consider following hypotheses:
\begin{align}
H_0: &~\projtheta_{\KL} = - 0.1, \quad \text{versus} \quad
H_1: ~\projtheta_{\KL} \ne - 0.1.
&\implies&&
H_0: ~\theta = 0, \quad \text{versus} \quad
H_1: ~\theta \ne 0.
\end{align}

We use Monte Carlo simulation to estimate the coverage and power: for a given value of $\theta$, we simulate $n=1000$ observations $X_1,\ldots,X_n \sim \truedist$, test whether $\projtheta_{\KL} = - 0.1$ is in each confidence sets.  We repeat this procedure 2000 times at each $\theta$ and estimate power at $\theta$ using the proportion of times that $\projtheta_{\KL} = - 0.1$ is not included in each confidence sets.

Figure~\ref{fig: contam_power} shows the estimated power of the LRTs against $\theta$ and their empirical coverage when the nominal coverage is 95\%. Under this particular contamination, robust Wald confidence set, $\KL$\Redi, and crossfit $\KL$\Redi attains the nominal coverage, while LRT, split LRT and crossfit split LRT do not. This highlights the importance of choosing methods robust to misspecification. Power comparisons for these non-robust methods are meaningless because these methods---LRT, sLRT and crossfit variant---do not provide valid inference.  However, we still include them with dashed lines for reference. Among the valid methods, the robust Wald confidence set (in yellow) has the highest power, while both $\KL$\Redi (in green) and its crossfit variant (in blue) exhibit conservativeness. As it was the case for the well-specified case, crossfit $\KL$\Redi improves power compared to vanilla method in high separation regime (i.e., $\theta > 0.25$) whereas the vanilla $\KL$\Redi has higher power in the low separation regime, $\theta \le 0.25$. When an accurate pilot estimate is available, crossfitting is recommended to improve power (in the high-separation regime) and yield smaller, valid confidence sets.

\begin{figure}[!htb]
  \begin{minipage}[b]{0.49\textwidth}
    \centering
    \includegraphics[trim={0 20 0 0},clip,scale=0.35]{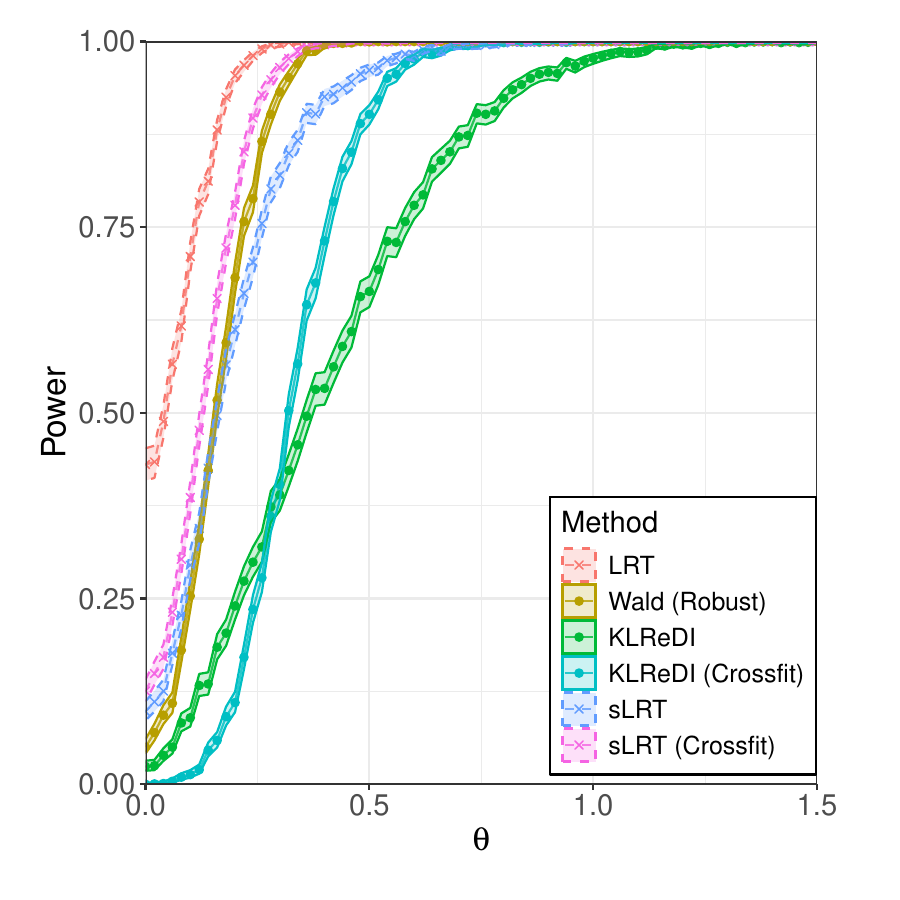}
  \end{minipage}
  \hfill
  \begin{minipage}[b]{0.49\textwidth}
    \centering
    \begin{tabular}[b]{lc}
  \hline
 & Coverage (95\% CI) \\ 
  \hline
LRT & 59\% (57\%, 61\%) \\ 
  Wald (Robust) & 94\% (92\%, 95\%) \\ 
  $\KL$\Redi & 98\% (98\%, 99\%) \\ 
  $\KL$\Redi (Crossfit) & 100\% (100\%, 100\%) \\ 
  sLRT & 90\% (89\%, 91\%) \\ 
  sLRT (Crossfit) & 87\% (86\%, 89\%) \\ 
   \hline
    \end{tabular}
    \vspace{3em}
  \end{minipage}
  \caption{Estimated power (left) and coverage (right) of classical LRT, $\KL$\Redi, crossfit $\KL$\Redi, split LRT, and crossfit split LRT under Gaussian location model $N(\theta, 1)$ obtained from 1,000 i.i.d. samples from the corrupted distribution $\truedist = 0.95 N(\theta,1) + 0.05 N(-2, 10^2)$. We are testing $H_0: \theta = 0$ versus $H_1: \theta \ne 0$. Classical LRT, sLRT, and crossfit sLRT fails to attain nominal 95\% coverage, and thus are dashed in the power curve.} \label{fig: contam_power}  
\end{figure}
}

{
\section{Empirical Verification of Irregular Bernoulli (Example~1 and 2)}
\label{sec:empirical_bernoulli}

In this section, we demonstrate the failures of the universal confidence set and the benefit of the proposed robust universal confidence sets in the irregular Bernoulli examples considered in Section~\ref{sec:fail} using simulation results. For both Example 1 and 2, we consider a more general parameter space of the Bernoulli working model $\model$ than what we considered in Section~\ref{sec:fail} (which are two disjoint points) so that we can compare the empirical width of the confidence intervals.

\subsection{Example 1}

Suppose we observe $X_1,\ldots,X_n \sim \Bern(\epsilon_n)$ such that $\epsilon_n = 0.8 / n$. Let the working model $\model = \{\Bern(p) : p \in \{0\} \cup [0.5, 1] \}$. We only consider $n \ge 20$ to let the approximate Hellinger projection be the singleton.

\noindent\textbf{KL Projection} With some algebra, we can identify that KL projection $\projdist_{\KL} = \Bern(0.5)$ and that the quasi-MLE as follows:
\begin{align*}
  \hat{p}_{n_1} = 
  \begin{cases}
    0, & \forall X_i = 0, \; i \in \calI_1,\\
    0.5 \vee \overline{X}_{n,1}, & o.w.
  \end{cases}
\end{align*}
where $\overline{X}_{n_1} = \sum_{i\in\calI_1} X_i / n_1$.
The universal confidence sets are
\begin{align*}
  C_{\alpha,n} = \left\{P \in \model : \left(2\overline{X}_{n_0} - 1\right) \left(\log (\hat{p}_{n_1}) - \log(p)\right) \le t_{\alpha} \right\}
\end{align*}
where $t_\alpha = \log 1/\alpha$ when sLRT set is used and
$z_{\alpha} \hat{s}_{p} / \sqrt{n_0}$ when $\KL$\Redi set is used. The critical issue with (asymptotic) $\KL$\Redi set is that the variance $s_{0}$ is infinite, i.e., $4 (\log \hat{p}_{n_1} / p)^2 (\bbV_{\truedist} X) = \infty$ for $p = 0$.
Nevertheless it is clear that sLRT failes to cover $\projdist_{\KL}$ at nominal level since whenever $\hat{p}_{n_1} = 0$ (which happens at least 20\%), $C_{\alpha,n}$ is $\{0\}$ so long as $\overline{X}_{n_0} \le 0.5$ (which happens with high probability).

The aforementioned issues can easily be verified in the simulation study. In particular, we obtain $\hat{p}_{n_1} = 0$  about 45\% of the simulation replication, and the empirical coverage of 95\% sLRT set and $\KL$\Redi set are around 80\% for sample sizes ranging from 40 to 10,000. This coincides with our conclusion in Section~\ref{sec:KL} that constructing a uniformly valid confidence set of KL projection is infeasible even with the proposed robust universal confidence set due to the instability of the KL projection.

\noindent\textbf{DP Projection} Recall that DP projection is $\Bern(0)$ for small enough $\epsilon_n$ Following the proof of Proposition~\ref{prop:DP_ex1}. We only consider DP parameter $\beta < 2$ here. 

Given minimum DP estimator $\pilot$, the (asymptotic) DP\Redi set is given as follows:
\begin{align*}
  C_{\alpha,n} = \bigg\{P \in \model : &\left[ p^{1+\beta} + (1-p)^{1+\beta} - \left(\widehat{p}_1^{1+\beta} +  (1 - \widehat{p}_1)^{1+\beta}\right)  \right]\\
     &- \left( 1+\frac{1}{\beta} \right) \left( \overline{X}_{n_0} \left[p^{\beta} - \widehat{p}_1^{\beta} \right] + (1-\overline{X}_{n_0}) \left[(1-p)^{\beta} - (1-\widehat{p}_1)^{\beta} \right]   \right) 
     \le z_{\alpha} \frac{\hat{s}_{p}}{\sqrt{n_0}} \bigg\}.
\end{align*}
Unlike KL \Redi set, the split statistic is uniformly bounded and thus we do not suffer from infinite variance. Figure~\ref{fig: bern_ex1_DP} summarizes the empirical performance of 95\% $\DP$\Redi sets by varying the parameter $\beta$. The results are averages over 1,000 replications. In all choices of $\beta$, the $\DP$\Redi set achieves the nominal coverage for sample sizes ranging from 40 to 10,000. Regarding the size of the confidence set, larger $\beta$ corresponds to a smaller confidence set as shown in the right panel.

\begin{figure}[!htb]
\centering
    \begin{subfigure}{.47\textwidth}
        \centering
        \includegraphics[trim={0 10 0 50},clip,width=\textwidth,height=2in]{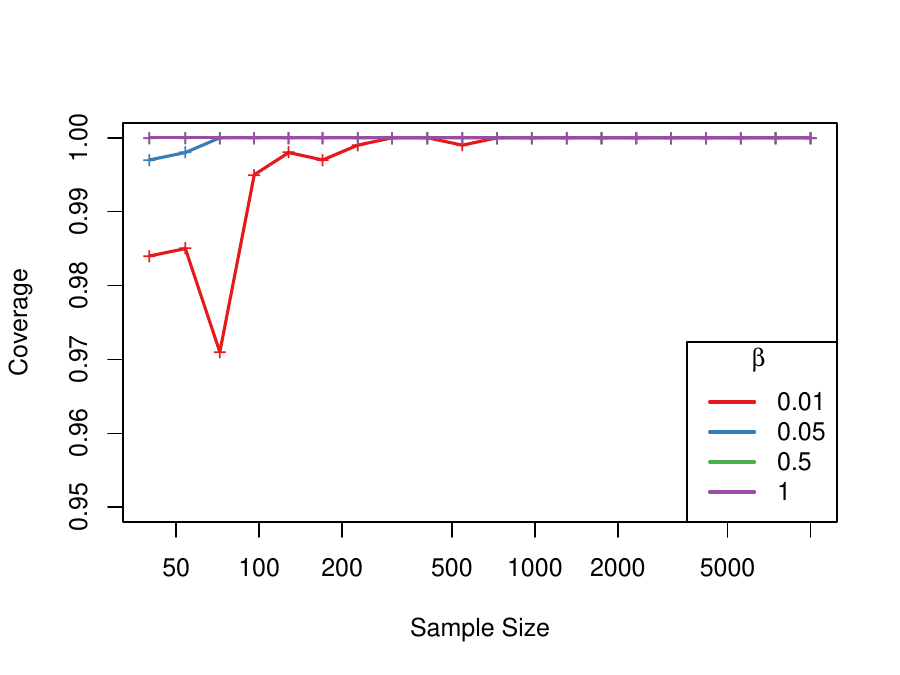}
    \end{subfigure}
    \begin{subfigure}{.47\textwidth}
        \centering
        \includegraphics[trim={0 10 0 50},clip,width=\textwidth,height=2in]{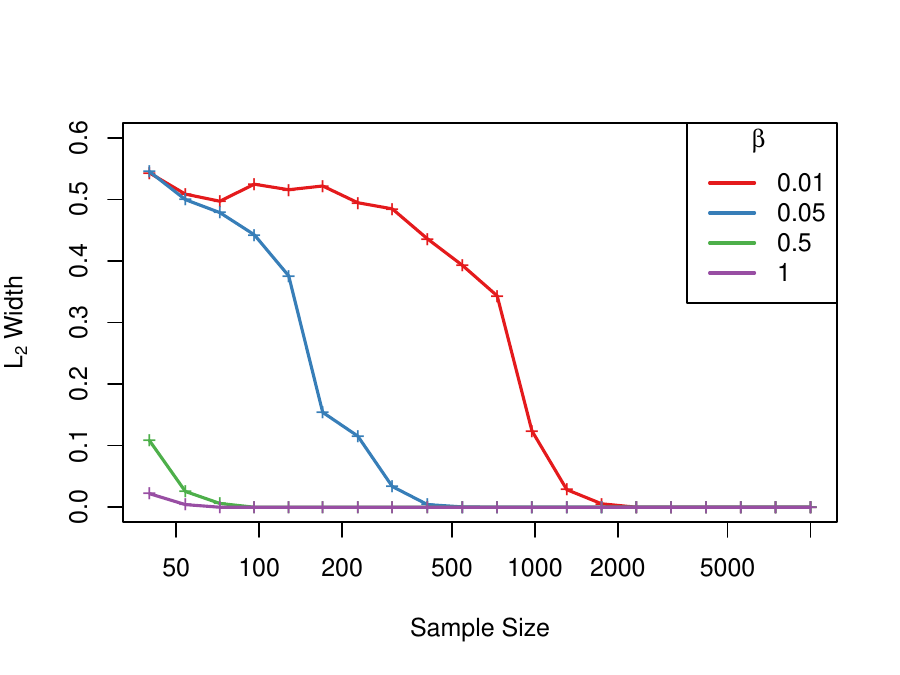}
    \end{subfigure}
    \caption{Performance of $\DP$\Redi set}\label{fig: bern_ex1_DP}    
\end{figure}

\noindent\textbf{Hellinger Projection} Given $\epsilon \le 0.04$, (approximate) Hellinger projection $\projdist_\HEL$ is $\Bern(0)$.

Empirically, 95 \% $\HEL$\Redi set covered the Hellinger projection 100\% of the replications for all sample sizes ranging from 40 to 10,000. Figure~\ref{fig: bern_ex1_HEL} shows the width of the $\HEL$\Redi set averaged over 1,000 replicates. This shows that $\HEL$\Redi set shrank to $\projdist_\HEL$ with a sample size of more than 100.
\begin{figure}[!htb]
\centering
    \begin{subfigure}{.47\textwidth}
        \centering
        \includegraphics[trim={0 10 0 50},clip,width=\textwidth,height=2in]{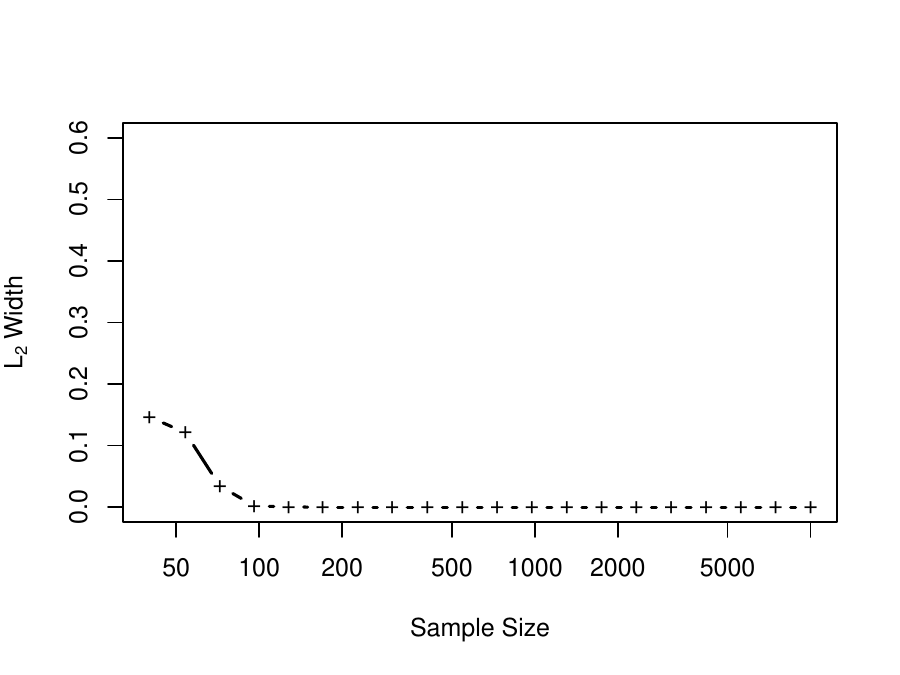}
    \end{subfigure}
    \caption{Width of $\HEL$\Redi set}\label{fig: bern_ex1_HEL}
\end{figure}

\subsection{Example 2}

Suppose we observe $X_1,\ldots,X_n \sim \Bern(0.5+\epsilon_n)$ such that $\epsilon_n = 0.5 / n$. Consider the working model $\model = \{\Bern(p) : p \in (0, 0.4]\cup [0.6,1) \}$.

\noindent\textbf{KL Divergence}
With some algebra, we can identify that KL projection $\projdist_{\KL} = \Bern(0.6)$ and that the quasi-MLE as follows:
\begin{align*}
  \hat{p}_{n_1} = 
  \begin{cases}
    \overline{X}_{n_1} \wedge 0.4, & \overline{X}_{n_1} < 0.5 \\
    \overline{X}_{n_1} \vee 0.6, & \overline{X}_{n_1} > 0.5\\
  \end{cases}
\end{align*}
where $\overline{X}_{n_1} = \sum_{i\in\calI_1} X_i / n_1$.
The universal confidence sets are
\begin{align*}
  C_{\alpha,n} = \left\{P \in \model : \left(2\overline{X}_{n_0} - 1\right) \left(\log (\hat{p}_{n_1}) - \log(p)\right) \le t_{\alpha} \right\}
\end{align*}
where $t_\alpha = \log 1/\alpha$ when sLRT set is used and
$z_{\alpha} \hat{s}_{p} / \sqrt{n_0}$ when $\KL$\Redi set is used. Figure~\ref{fig: bern_ex2_DP} shows the empirical coverage and $L_2$ width of sLRT and $\KL$\Redi set, respectively. This result empirically verifies our claim that sLRT fails to attain nominal coverage whereas the proposed $\KL$\Redi set asymptotically ensures uniform validity by properly enlarging the size of the confidence set.
\begin{figure}[!htb]
\centering
    \begin{subfigure}{.47\textwidth}
        \centering
        \caption{Coverage}
        \includegraphics[trim={0 10 0 50},clip,width=\textwidth,height=2in]{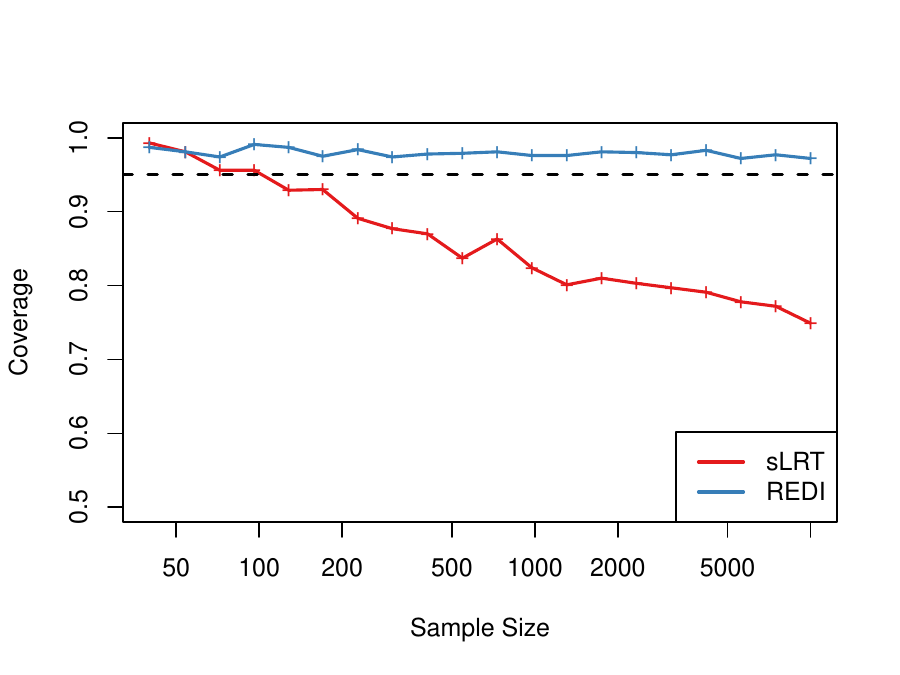}
    \end{subfigure}
    \begin{subfigure}{.47\textwidth}
        \centering
        \caption{Width}
        \includegraphics[trim={0 10 0 50},clip,width=\textwidth,height=2in]{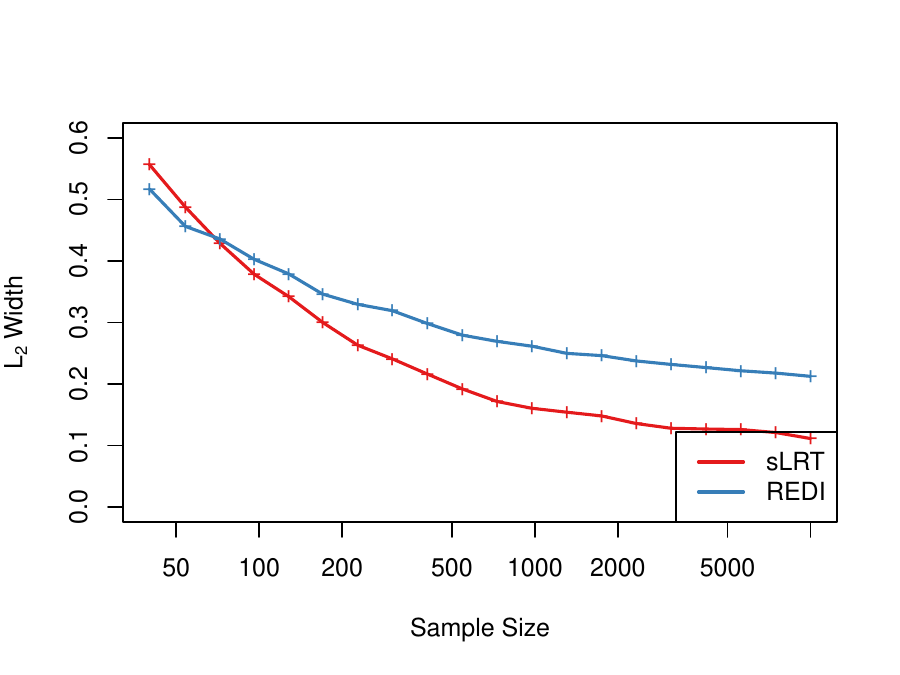}
    \end{subfigure}
    \caption{Performance of sLRT and $\KL$\Redi set}\label{fig: bern_ex2_DP}    
\end{figure}

\section{Additional Illustration: Overdispersion}
\label{sec:overdisp}

Overdispersion is a classic example of model misspecification where the true distribution has larger variance than what can be represented by the hypothesized model. Specifically, consider a case of count data generated from the negative binomial distribution $\truedist$ with mean $\mathbb{E}_{\truedist} (X):= \theta^*$ and variance $\mathbb{V}_{\truedist} (X) =  \kappa \theta^*$ where the positive constant $\kappa$ represents the \emph{dispersion ratio}. Suppose a statistician hypothesized a Poisson model $\mathcal{P}_{\Theta} = \{\textrm{Poi}(\theta) : \theta \in \mathbb{R}_{+} \}$ to best describe $\truedist$. Since the mean and the variance are the same for the Poisson distribution (implicitly assuming $\kappa=1$), the dispersion ratio $\kappa$ captures the severity of the model misspecification. Figure~\ref{fig: Poi_Pop} shows $\rho (\truedist\| \textrm{Poi}(\theta))$ with $\rho = \KL, \HEL, \TV$ across the dispersion ratio. Notice that KL projection is the true mean $\theta^* (= 10)$ regardless of the dispersion ratio whereas Hellinger and TV projection gets smaller as the true variance is more inflated. 
{
We remark that, in reality, the true distribution is unknown to the statistician, and thus these projections cannot be known a priori. It is the choice of inferential tool that implicitly determines the target of inference.
}

\begin{figure}[!htb]
\centering
    \begin{subfigure}{.32\textwidth}
        \centering
        \caption{KL}
        \includegraphics[trim={25 20 25 20},clip,width=\textwidth,height=1.55in]{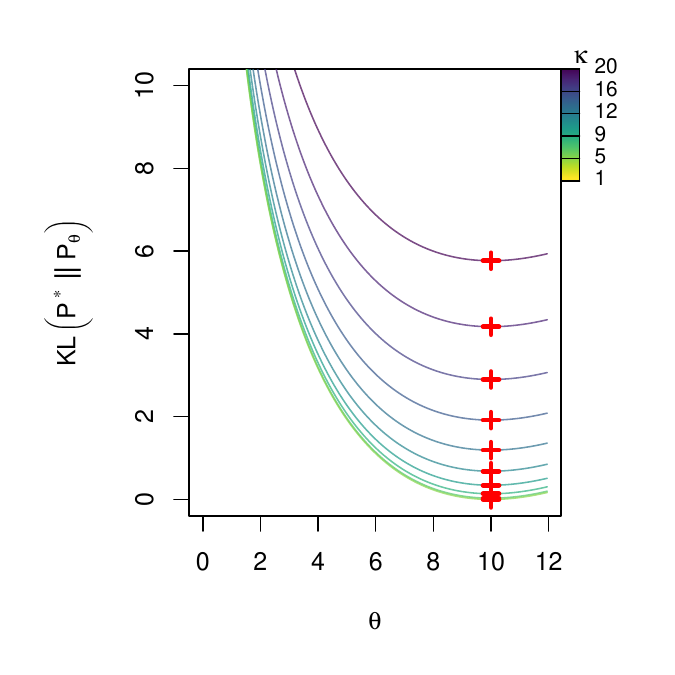}
    \end{subfigure}
    \begin{subfigure}{.32\textwidth}
        \centering
        \caption{Hellinger}
        \includegraphics[trim={25 20 25 20},clip,width=\textwidth,height=1.55in]{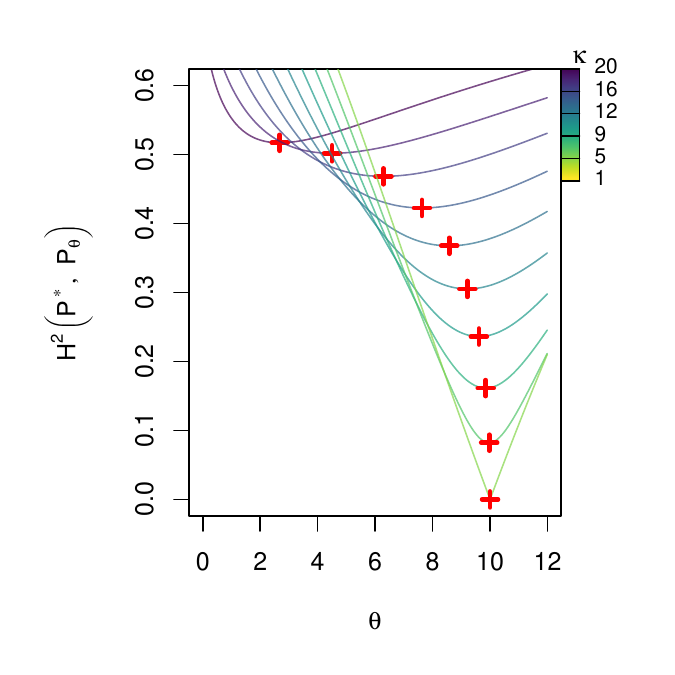}
    \end{subfigure}
    \begin{subfigure}{.32\textwidth}
        \centering
        \caption{TV}
        \includegraphics[trim={25 20 25 20},clip,width=\textwidth,height=1.55in]{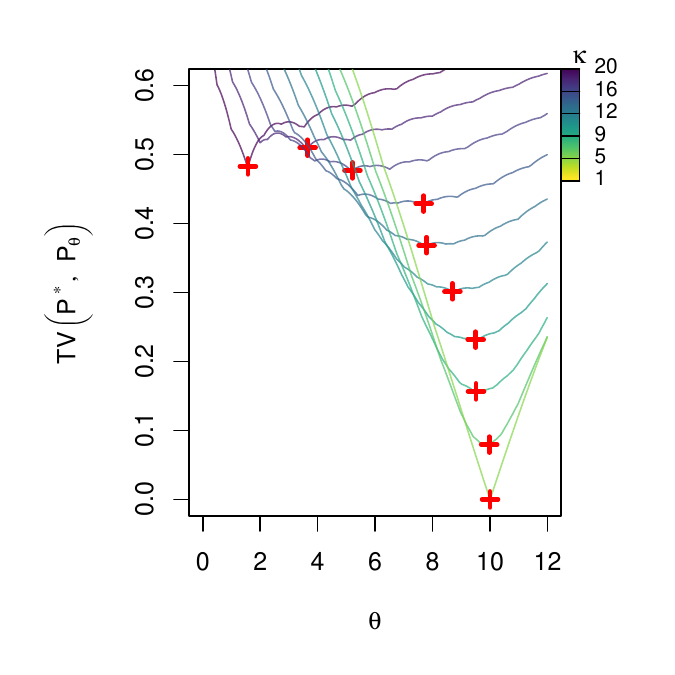}
    \end{subfigure}
    \caption{$\rho (\truedist\| \textrm{Poi}(\theta))$ across the varying dispersion ratio $\kappa$. Red + indicates $\rho (\truedist\| \projdist)$.}\label{fig: Poi_Pop}    
\end{figure}

\subsection{The split LRT is sensitive to the misspecification.}
As highlighted in Section~\ref{sec:KL}, the split LRT confidence set $(\Chat_{sLRT})$ may fail to cover the KL projection unlike the KL \Redi set $(\Chat_{\KL\Redi})$ even with the same choice of $\widehat{\theta}_1$ and the same log split likelihood-ratio statistic. Figure~\ref{fig: Poi_sLRT_Rift} contrasts the performance of $\Chat_{sLRT}$ and $\Chat_{\KL\Redi}$ based on 1000 replicates of 200 simulated observations. In computing the confidence sets, the observations are equally split in half and we choose $\widehat{\theta}_1$ to be the sample mean (which is the MLE) of the first half samples. As the misspecification gets more severe (larger $\kappa$), the empirical coverage of KL projection parameter $(\tilde{\theta})$ decreases for $\Chat_{sLRT}$. When the dispersion ratio becomes larger than 3, $\Chat_{sLRT}$ fails to achieve the nominal 95\% coverage whereas $\Chat_{\KL\Redi}$ maintains the validity regardless of how severe the misspecification is. Both the center and the right panel depict the size of the estimated confidence set varying over the dispersion ratio but from a different perspective. The former is based on the maximal excess KL divergence from the KL projection (which can be at most twice the KL-diameter of the set) whereas the latter is based on the $L_2$ distance over the parameter space. It is not surprising that compared to $\Chat_{\KL\Redi}$, $\Chat_{sLRT}$ is smaller in the $L_2$ sense and is closer to $\projdist$ in an excess divergence sense.
\begin{figure}[!htb]
\centering
    \begin{subfigure}{.32\textwidth}
        \centering
        \includegraphics[trim={0 20 40 20},clip,width=\textwidth,height=1.55in]{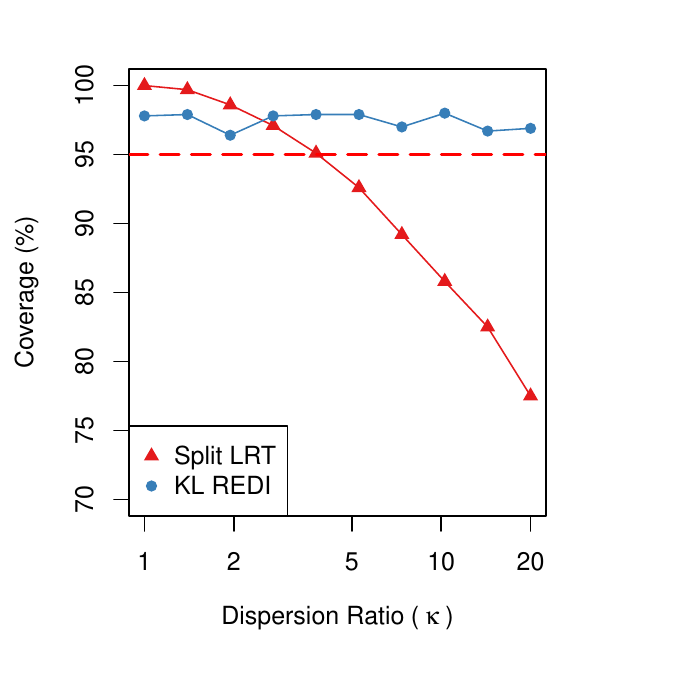}
    \end{subfigure}
    \begin{subfigure}{.32\textwidth}
        \centering
        \includegraphics[trim={0 20 40 20},clip,width=\textwidth,height=1.55in]{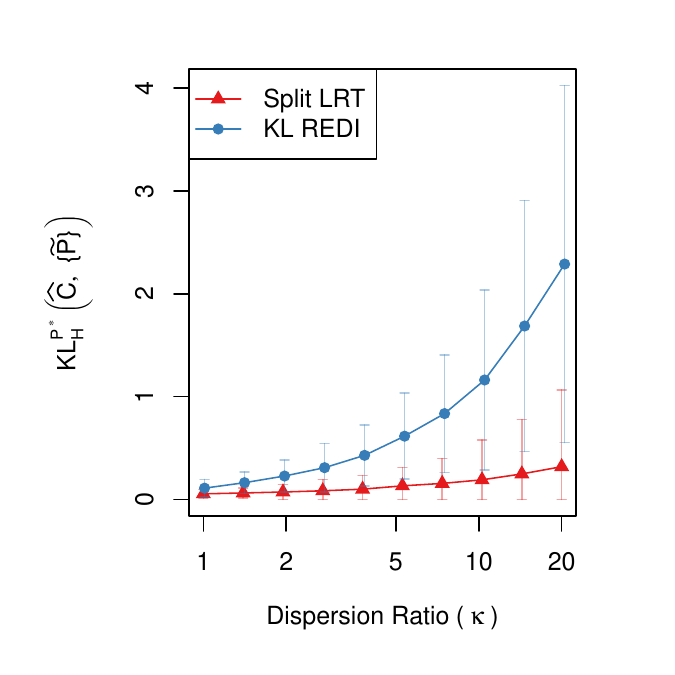}
    \end{subfigure}
    \begin{subfigure}{.32\textwidth}
        \centering
        \includegraphics[trim={0 20 40 20},clip,width=\textwidth,height=1.55in]{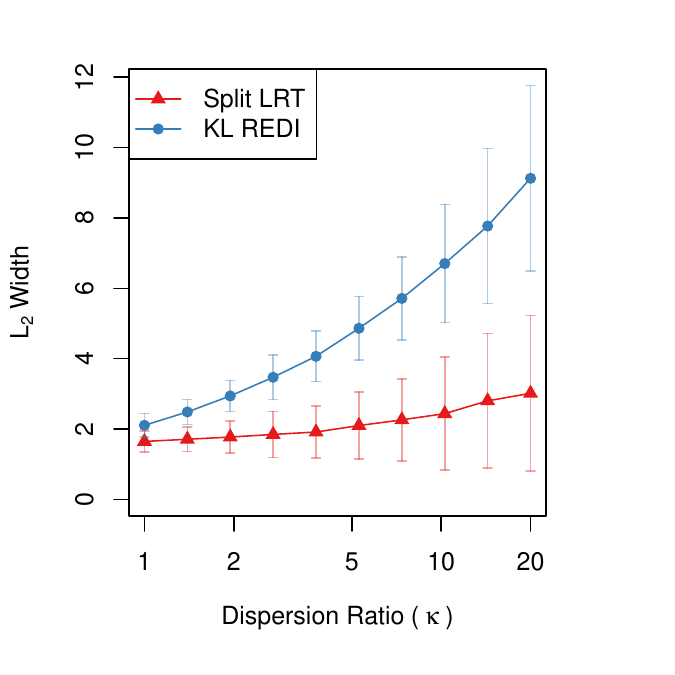}
    \end{subfigure}    
    \caption{Inference on KL projection: (Left) Empirical coverage (Center) $\KL_{\scH}^{\truedist} \left( \Chat, \{\projdist\} \right)$ (Right) $L_2$ width of $\Chat$ in the parameter space. Point estimates are median over 1000 replicates and the error bar indicates $\pm$ 1 standard error.}\label{fig: Poi_sLRT_Rift}    
\end{figure}

\subsection{Beyond KL projection}

Unlike the KL projection, the Hellinger and TV projections are different for different degrees of overdispersion. Our target of inference regarding Hellinger and TV distance is $\nu$-approximate projection rather than the projection $\projdist$ as seen in the left panel of Figure~\ref{fig: Poi_HEL_TVRift}. When the factor $\kappa \geq 6$ the $\nu$-approximate target for both Hellinger and TV distance includes any $\theta \in \mathbb{R}_+$. For values of dispersion ratio $\kappa \geq 6$, the $\nu$-approximate projection for both the Hellinger and TV distances becomes $\model$ and thus the approximate coverages are trivially 100\%. Once again
this highlights that the approximate projection is a meaningful target only when the model misspecification is not too severe.

Figure~\ref{fig: Poi_HEL_TVRift} summarizes the performance of approximate \Redi sets regarding Hellinger ($\Chat_{\HEL\Redi}$) and TV distance ($\Chat_{\TV\Redi}$) based on 1000 replicates of 200 simulated observations. We choose the minimum distance estimator for $\widehat{\theta}_1$ for both $\Chat_{\HEL\Redi}$ and $\Chat_{\TV\Redi}$. Both $\Chat_{\HEL\Redi}$ and $\Chat_{\TV\Redi}$ yield 100\% empirical coverage---defined as a proportion of the confidence set that intersects $\projdist_\nu$---across all dispersion ratios except almost well-specified case (0.01\% dispersion) with 97.4\% and 99.1\% coverage, respectively. This conservativeness is expected because for these divergences we have relaxed our target of inference to be the set
of $\nu$-approximate projections. 

Nevertheless, this does not mean that the Hellinger and TV \Redi sets are vacuously large. The center and right panel of Figure~\ref{fig: Poi_HEL_TVRift} show the diameter of the \Redi set in Hellinger or TV distance sense, or Euclidean sense. The size of the \Redi set increases as the misspecification exacerbates regardless of distance measure. In general, $\Chat_{\TV\Redi}$ is larger than $\Chat_{\HEL\Redi}$. $\Chat_{\HEL\Redi}$ behaves closer to $\Chat_{\KL\Redi}$ under slight to moderate overdispersion and to $\Chat_{\TV\Redi}$ as the overdispersion becomes severe.

\begin{figure}[!htb]
\centering
    \begin{subfigure}{.32\textwidth}
        \centering
    \end{subfigure}
    \begin{subfigure}{.32\textwidth}
        \centering
        \includegraphics[trim={0 20 40 20},clip,width=\textwidth,height=1.55in]{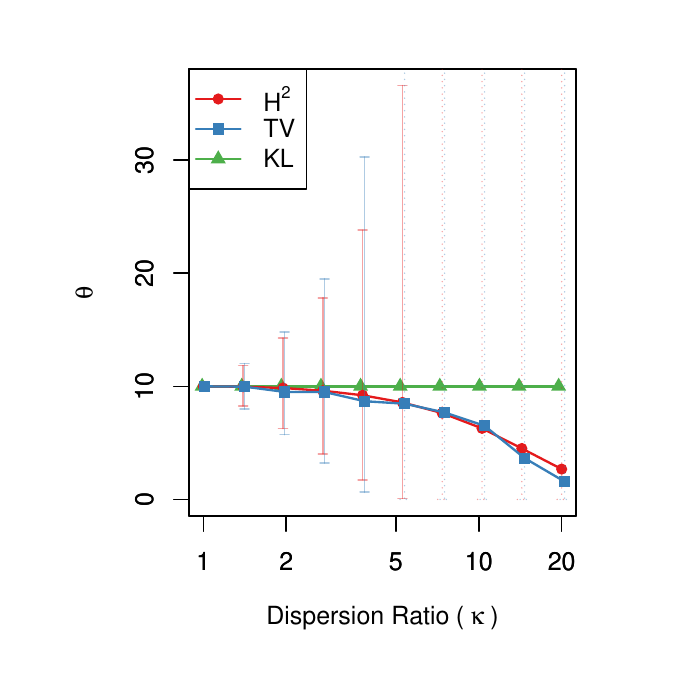}
    \end{subfigure}
    \begin{subfigure}{.32\textwidth}
        \centering
        \includegraphics[trim={0 20 40 20},clip,width=\textwidth,height=1.55in]{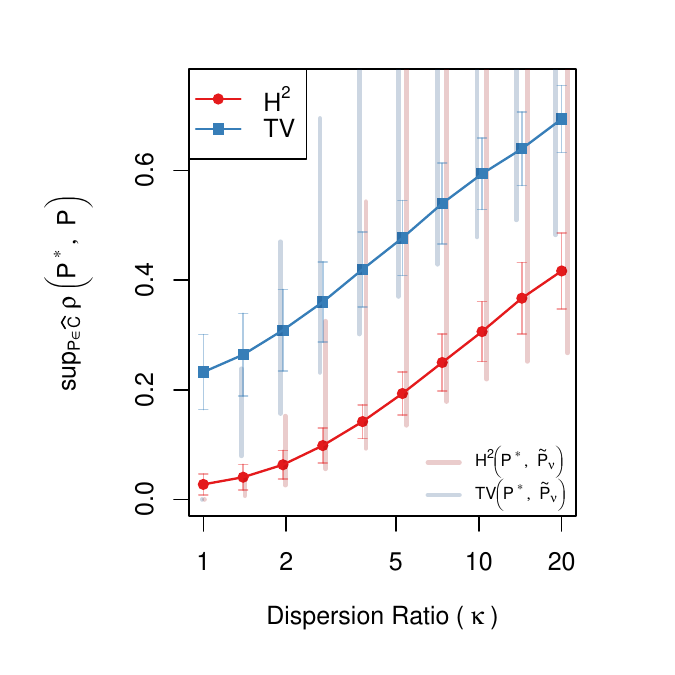}
    \end{subfigure}    
    \begin{subfigure}{.32\textwidth}
        \centering
        \includegraphics[trim={0 20 40 20},clip,width=\textwidth,height=1.55in]{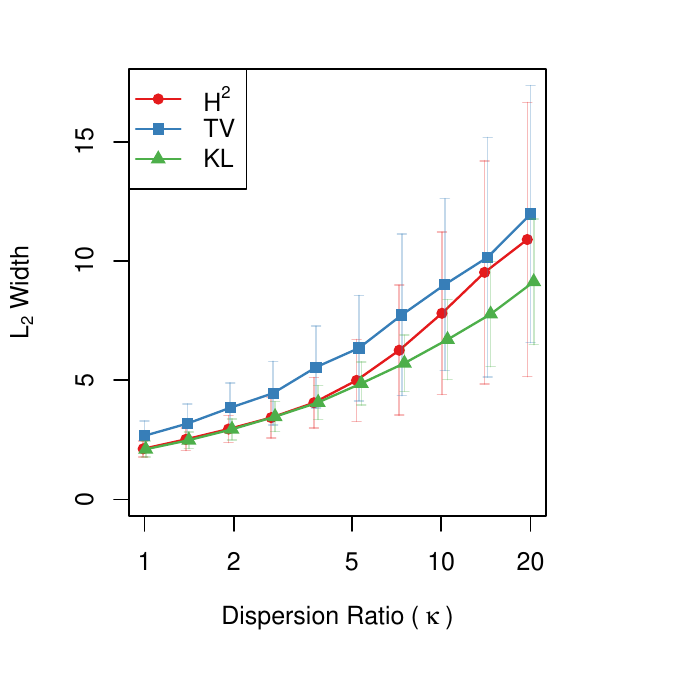}
    \end{subfigure}    
    \caption{Inference on $\nu$-approximate Hellinger and TV projection: (Left) $\tilde{\theta}$ (in points) and $\tilde{\theta}_\nu$ (in intervals) where the dotted intervals indicate $\bbR_{+}$ (Center) $\sup_{P_0 \in \Chat} \rho (\truedist\| P_0)$ superimposed with $\rho (\truedist\| \projdist_\nu) := \{\rho(\truedist\| P_0) : P_0 \in \projdist_\nu\}$ (Right) $L_2$ Width of (KL, Hellinger, and TV) \Redi set in the parameter space. Point estimates in the center and right panels are median over 1000 replicates and the error bar indicates $\pm$ 1 standard error. }\label{fig: Poi_HEL_TVRift}    
 \end{figure}

\subsection{Comparison between asymptotic and finite sample valid \Redi sets}
\label{sec:overdisp_finite_comparision}

Figure~\ref{fig: Poi_compare_overdispersion} compares the various TV \Redi set when the $\truedist$ is a $32\%$ variance inflated negative binomial---Berry-Esseen ($\Chat_{\TV\Redi}$), Hoeffding bound ($\Chat_{\TV\textsc{HF}}$), empirical Bernstein bound \citep[$\Chat_{\TV\textsc{EBS}}$;][]{waudby-smith_estimating_2023}, and empirical Bentkus bound \citep[$\Chat_{\TV\textsc{EBK}}$;][]{bentkus_domination_2006,kuchibhotla_near-optimal_2021}. See Section~\ref{sec:finite_set} for explicit forms of each confidence set. In all cases, we choose the same minimum TV distance estimator $\widehat{\theta}_1$. The KL \Redi set dominates all finite sample valid confidence sets considered in this section, despite its validity relying on asymptotics. The finite sample valid \Redi sets are too conservative (and yield a meaningless set $\Chat = \model$) when only a few observations are available ($n \le 50$).
Although our paper does not primarily focus on obtaining the tightest finite-sample valid confidence set, leveraging the variance $\bbV_{\truedist}(X)$ can often be beneficial when constructing the confidence set. In this example, $\Chat_{\TV\textsc{EBS}}$ and $\Chat_{\TV\textsc{EBK}}$ outperform $\Chat_{\TV\textsc{HF}}$ since the Bernstein and Bentkus bounds are more sensitive to the variance.

\begin{figure}[!htb]
\centering
   \begin{subfigure}{.32\textwidth}
       \centering
        \includegraphics[trim={20 20 40 30},clip,width=\textwidth]{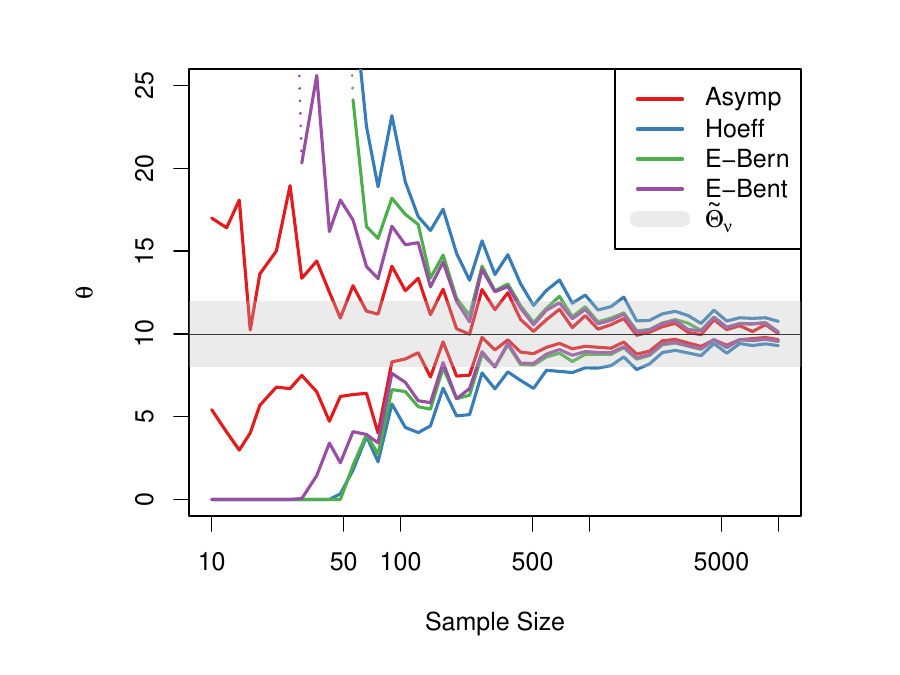}
   \end{subfigure}
   \begin{subfigure}{.32\textwidth}
       \centering
        \includegraphics[trim={20 20 40 30},clip,width=\textwidth]{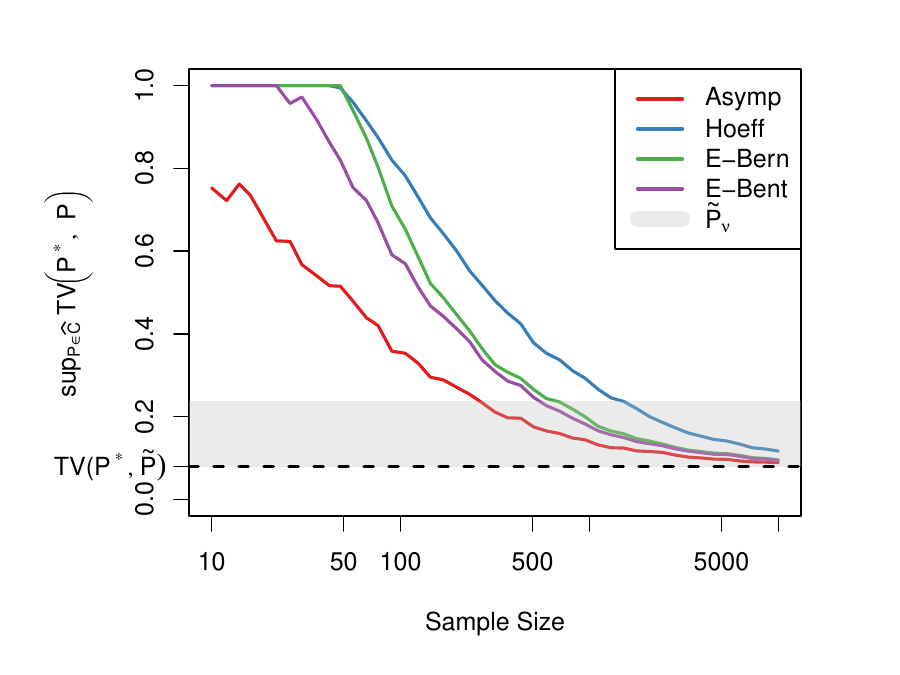}
   \end{subfigure}
   \begin{subfigure}{.32\textwidth}
       \centering
        \includegraphics[trim={20 20 40 30},clip,width=\textwidth]{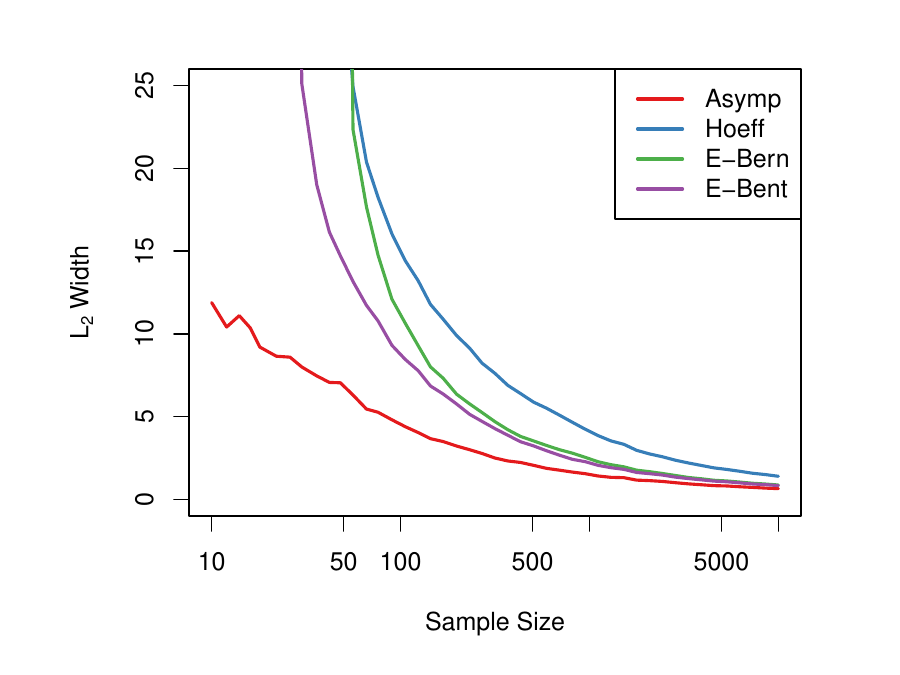}
   \end{subfigure}
    \caption{TV \Redi set based on Berry-Esseen, Hoeffding, empirical Bernstein and empirical Bentkus bounds regarding $32\%$ variance inflation example. (Left) Confidence set over $\Theta$ from a single replication. Gray area indicates the $\nu$-approximate TV projection parameter set $\{\theta : P_\theta \in \projdist_\nu \}$. (Center) Median $\sup_{P \in \Chat} \TV (\truedist, P)$ superimposed with $\{\TV (\truedist, P): P\in \projdist_\nu\}$. (Right) Median $L_2$ widths of TV \Redi parameter sets.}\label{fig: Poi_compare_overdispersion}   
\end{figure}
}

\section{Additional results and technical details of Section~\ref{sec:Empirical_analysis}}\label{app: numerical}

\subsection{Visualization of Data Generating Process}\label{sec:vis_dgp}

\begin{figure}[!htb]
\centering
    \begin{subfigure}{.47\textwidth}
        \centering
        \caption{Probability Density Function}
        \includegraphics[trim={0 0 0 0},clip,width=\textwidth,height=2in]{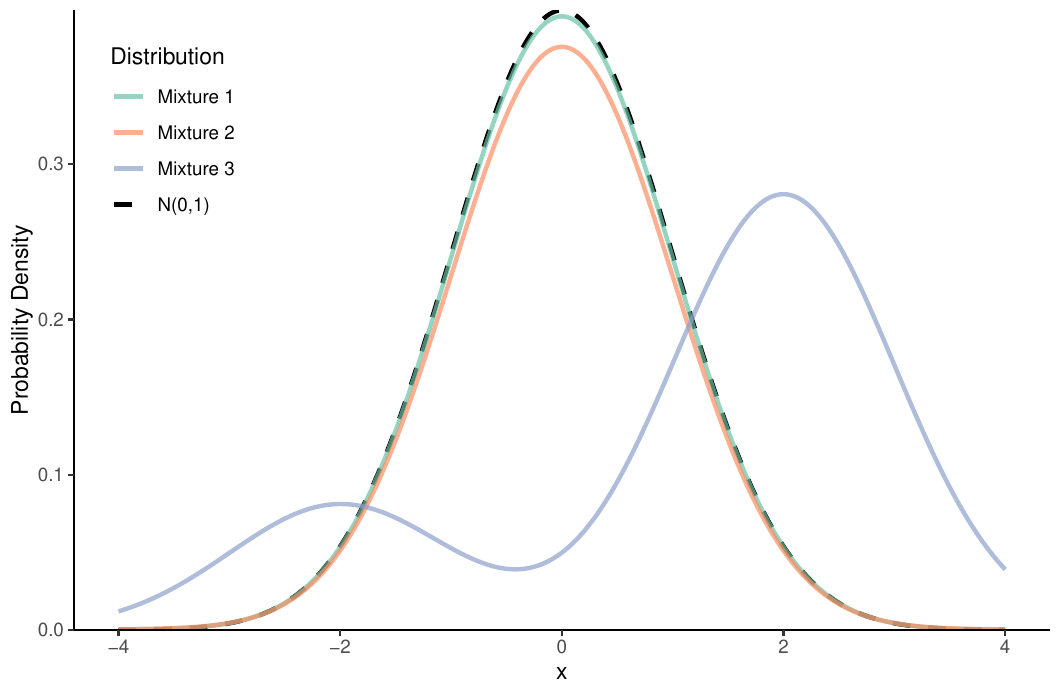}
    \end{subfigure}
    \begin{subfigure}{.47\textwidth}
        \centering
        \caption{Cumulative Density Function}
        \includegraphics[trim={0 0 0 0},clip,width=\textwidth,height=2in]{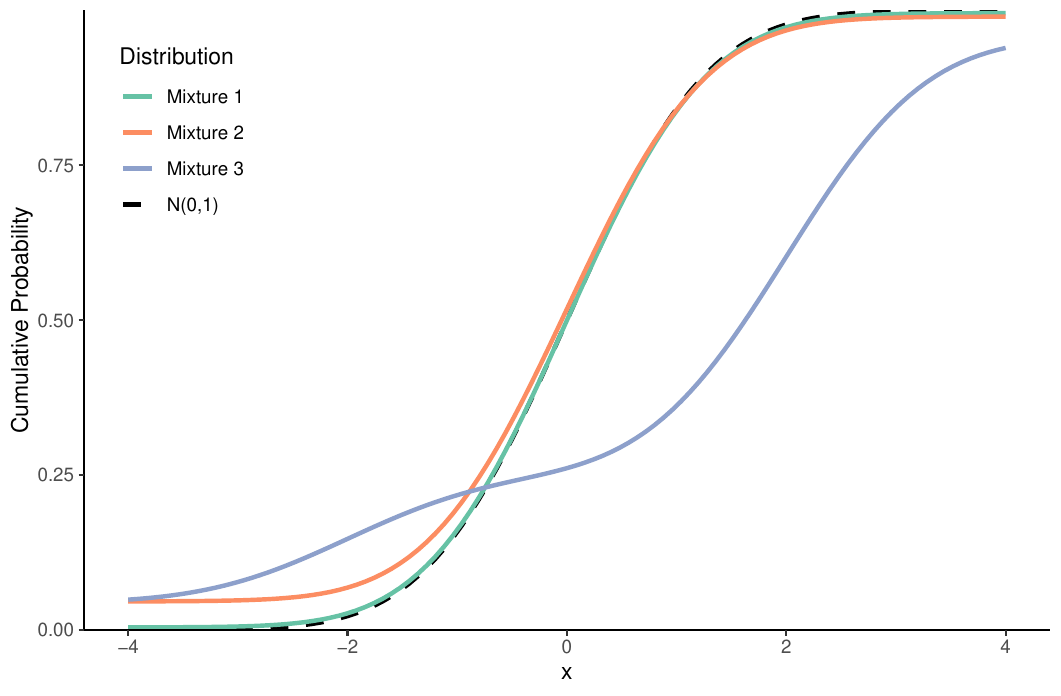}
    \end{subfigure}
    \caption{PDF and CDF of contaminated mixture distribution and standard normal distribution}\label{fig: sim_setting}    
\end{figure}

\subsection{Computational detail}

We adopt a heuristic search method for finding a confidence set in multivariate parameter space. For brevity, we explain the procedure in 2 dimensions, but the procedure can be straightforwardly extended to higher dimensions. From the observation that when $\thetahat_1$ is close to $\tilde{\theta}$, i.e., $\rho(\truedist\| P_{\thetahat_1}) \le \nu \rho(\truedist \| \projdist)$ as seen in the proof of Theorem~\ref{thm:honest_approx}, $\overline{T}_{n_0} (\thetahat_1) = 0$ for split statistics that satisfies Assumption~\ref{ass:T_antisym} and \ref{ass:T_ub}. Therefore, we construct a star-convex confidence set that always includes $\thetahat_1$. We construct the rays originate from $\thetahat_1$, i.e., $R_\omega = \{ \theta \in \Theta : r_\omega^\top (\theta - \thetahat_1) = 0, r \ge 0 \}$ where $r_\omega = (r \sin \omega, - r \cos \omega)$ for angle $\omega \in [- \pi, \pi]$. For each $\omega$, we find a root of an evidence function $\evidence(\theta) = \overline{T}_{n_0} (\theta) - t_\alpha (\theta)$ using Brent's method \citep{brent_algorithms_1972} on $R_\omega$ constrained with radius $r$ varying from 0 (corresponding $\theta=\thetahat_1$) to some $r_0 > 0$ such that the corresponding $\projtheta$ satisfies $\evidence(\projtheta) > 0$.

\subsection{Gaussian contamination - Location family} \label{app: contam_loc}

Consider a Gaussian location family $\model = \{ \Normal(\theta, 1) : \theta \in \bbR\}$ where the variance is fixed to that of uncontaminated distributions. Figure~\ref{fig: contam_proj_1d} shows the projection parameters along with those of contaminated and uncontaminated distributions. The mean of contaminated distribution and that of uncontaminated distributions are the same for Cases 1 and 3 but not for Case 2. This leads to the interesting observation that forward KL projection is the closest to the uncontaminated distribution in Case 3 unlike location-scale family in Figure~\ref{fig: contam_proj}, Section~\ref{sec:Empirical_analysis}. 
\begin{figure}[!htb]
\centering
    \begin{subfigure}{.32\textwidth}
        \centering
        \caption{Case 1: Symmetric}
        \includegraphics[trim={15 10 15 10},clip,width=\textwidth,height=1.2in]{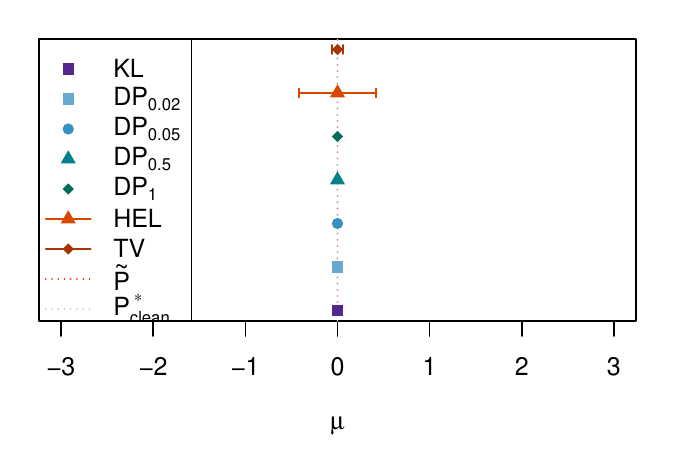}
    \end{subfigure}
    \begin{subfigure}{.32\textwidth}
        \centering
        \caption{Case 2: Asymmetric}
        \includegraphics[trim={15 10 15 10},clip,width=\textwidth,height=1.2in]{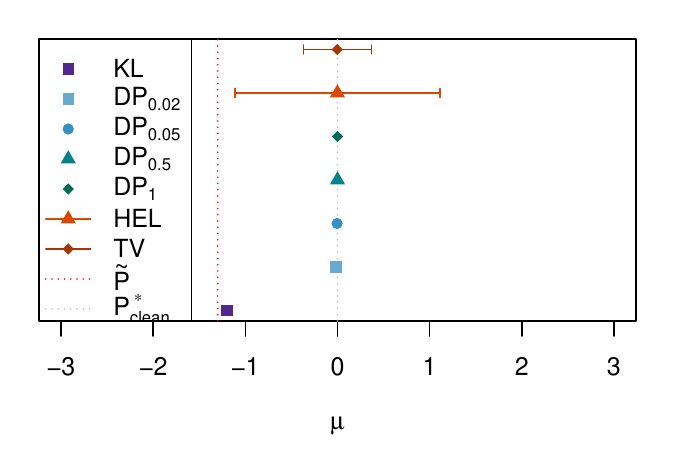}
    \end{subfigure}
    \begin{subfigure}{.32\textwidth}
        \centering
        \caption{Case 3: Heavily asymmetric}
        \includegraphics[trim={15 10 15 10},clip,width=\textwidth,height=1.2in]{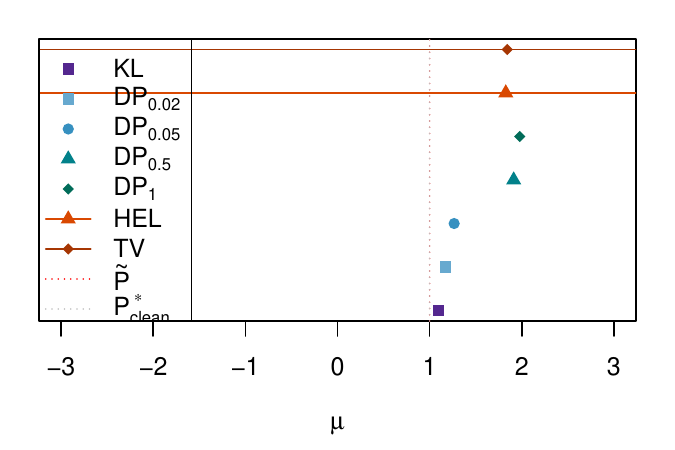}
    \end{subfigure}    
    \caption{Parameters of projections ($\KL$, $\DP_\beta$, $\HEL^2$, $\TV$), contaminated distribution ($\projdist$), and uncontaminated distribution $(P^*_{\Truth})$. Parameters corresponding to $\projdist_\nu$ for Hellinger and TV distance are given in intervals.}\label{fig: contam_proj_1d}
\end{figure}

Figure~\ref{fig: contam_1d_KLDP} summarizes the performance of confidence sets targeting the forward KL or DP projection over 1000 replications. Clearly, split LRT fails to attain the nominal coverage even for a large enough sample size. All other \Redi sets achieve the nominal coverage for moderate to large sample size. $\Chat_{\DP\Redi}$ are shorter than $\Chat_{\KL\Redi}$ and even than the invalid split LRT set for Cases 2 and 3. 

\begin{figure}[!htb]
\centering
    \begin{subfigure}{.32\textwidth}
        \centering
       \caption{Case 1: Symmetric}
        \includegraphics[trim={0 20 40 20},clip,width=\textwidth,height=1.60in]{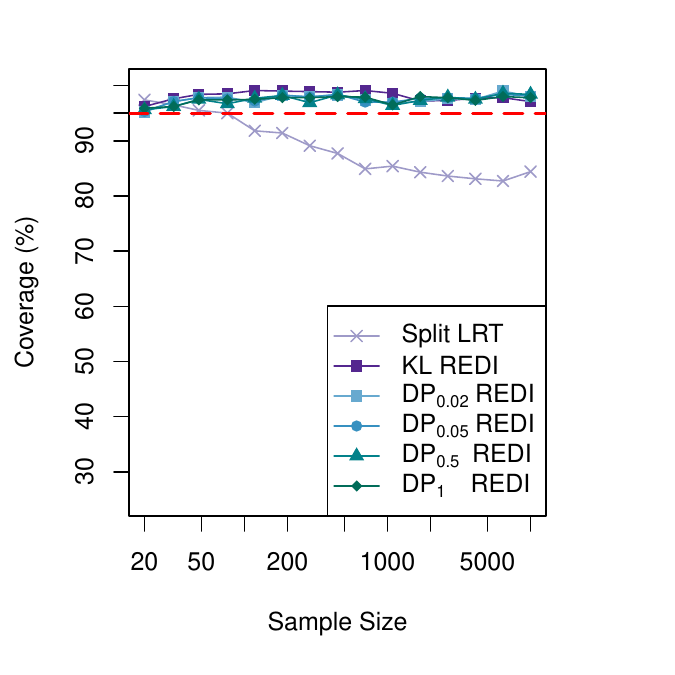}
    \end{subfigure}
    \begin{subfigure}{.32\textwidth}
        \centering
       \caption{Case 2: Asymmetric}
        \includegraphics[trim={0 20 40 20},clip,width=\textwidth,height=1.60in]{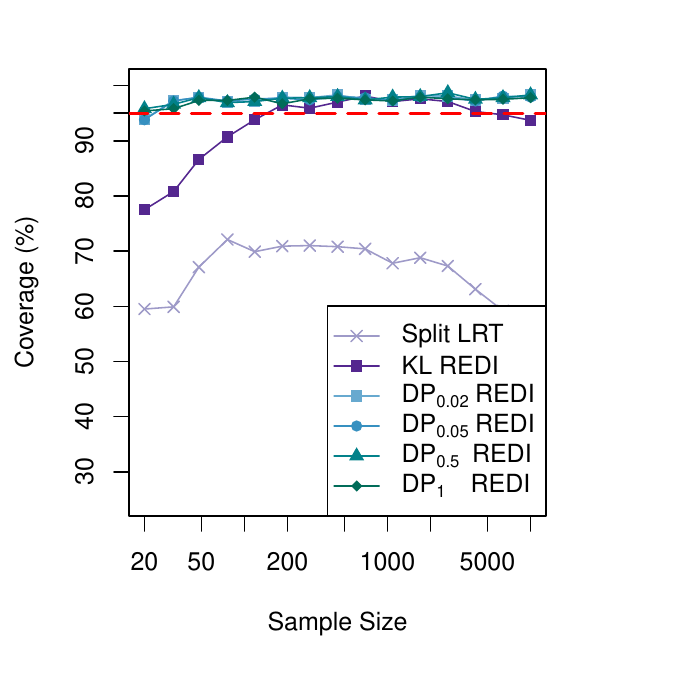}
    \end{subfigure}
    \begin{subfigure}{.32\textwidth}
        \centering
       \caption{Case 3: Heavily asymmetric}
        \includegraphics[trim={0 20 40 20},clip,width=\textwidth,height=1.60in]{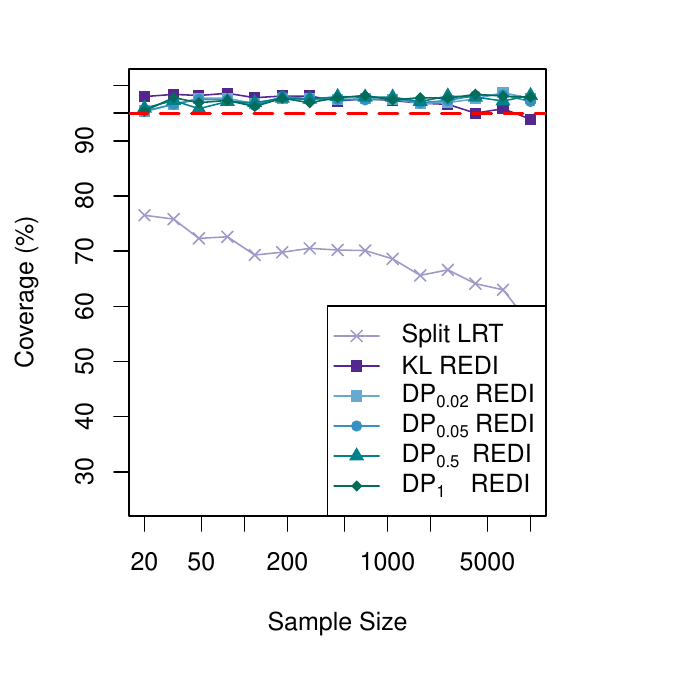}
    \end{subfigure}    
    \begin{subfigure}{.32\textwidth}
        \centering
        \includegraphics[trim={15 20 25 20},clip,width=\textwidth,height=1.60in]{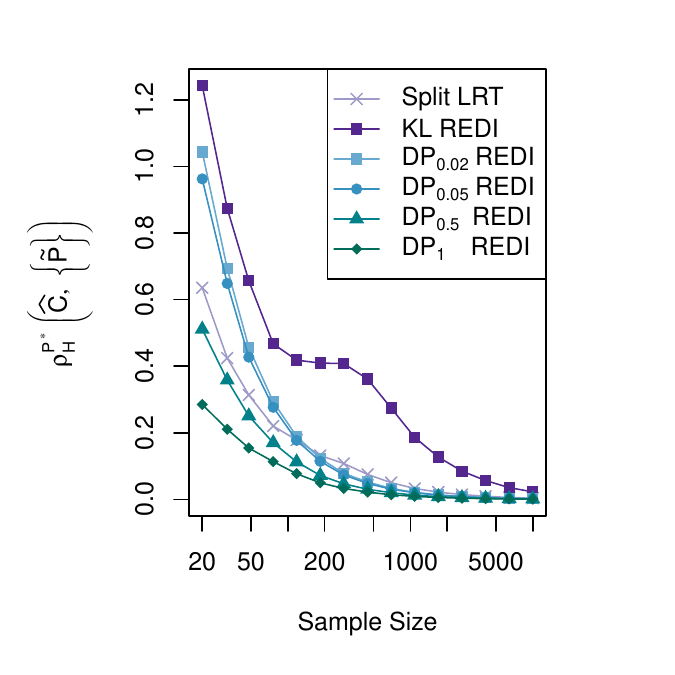}
    \end{subfigure}
    \begin{subfigure}{.32\textwidth}
        \centering
        \includegraphics[trim={15 20 25 20},clip,width=\textwidth,height=1.60in]{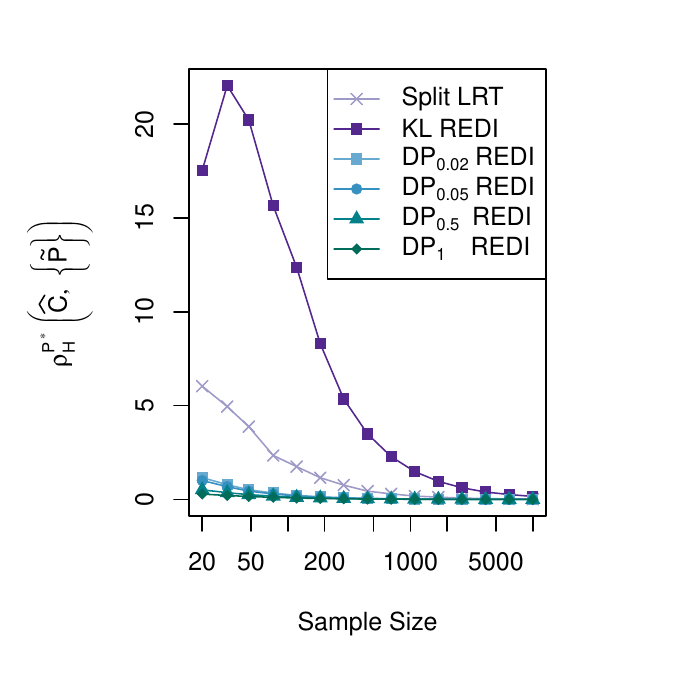}
    \end{subfigure}
    \begin{subfigure}{.32\textwidth}
        \centering
        \includegraphics[trim={15 20 25 20},clip,width=\textwidth,height=1.60in]{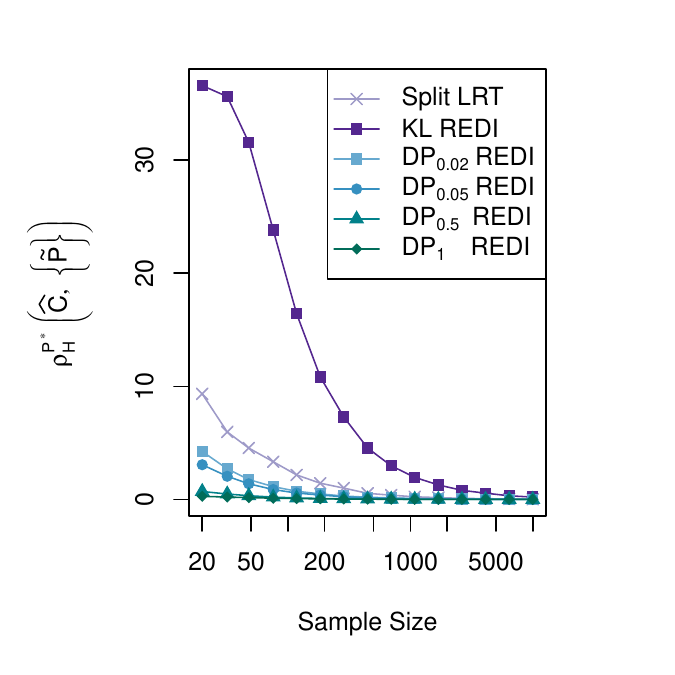}
    \end{subfigure}
    \begin{subfigure}{.32\textwidth}
        \centering
        \includegraphics[trim={15 20 25 20},clip,width=\textwidth,height=1.60in]{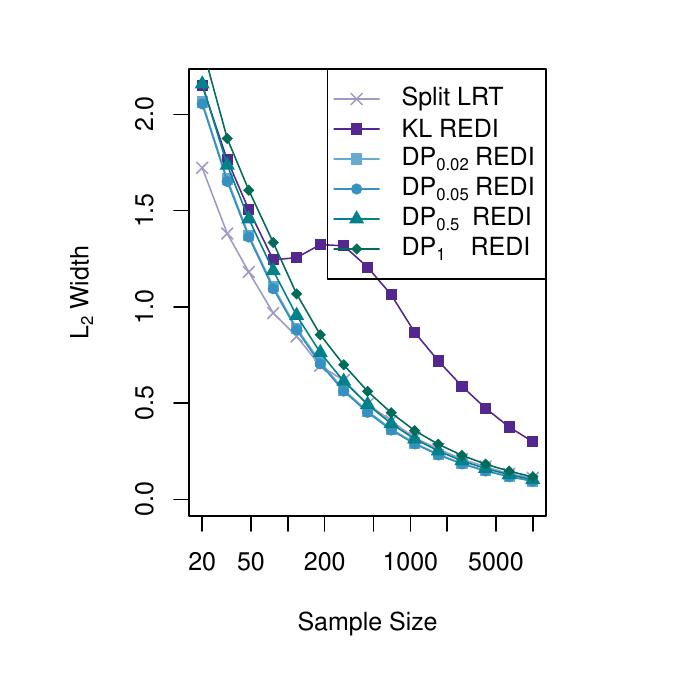}
    \end{subfigure}
    \begin{subfigure}{.32\textwidth}
        \centering
        \includegraphics[trim={15 20 25 20},clip,width=\textwidth,height=1.60in]{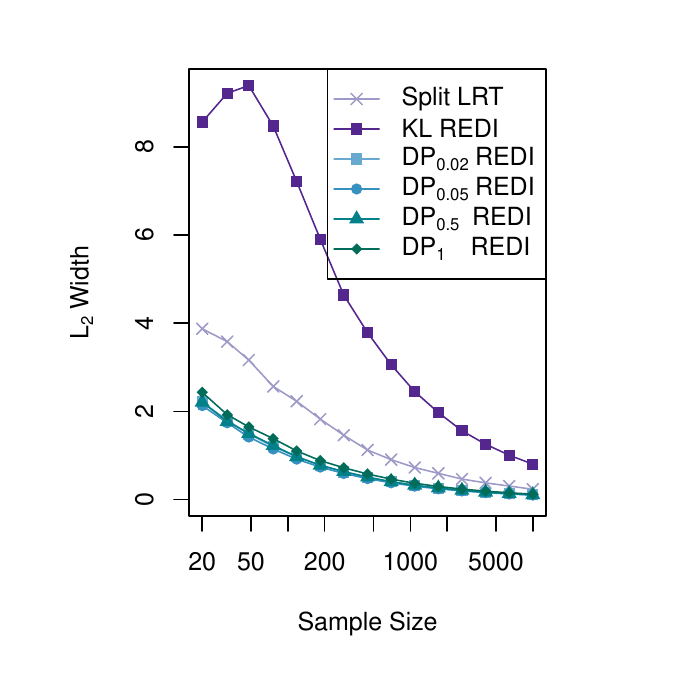}
    \end{subfigure}
    \begin{subfigure}{.32\textwidth}
        \centering
        \includegraphics[trim={15 20 25 20},clip,width=\textwidth,height=1.60in]{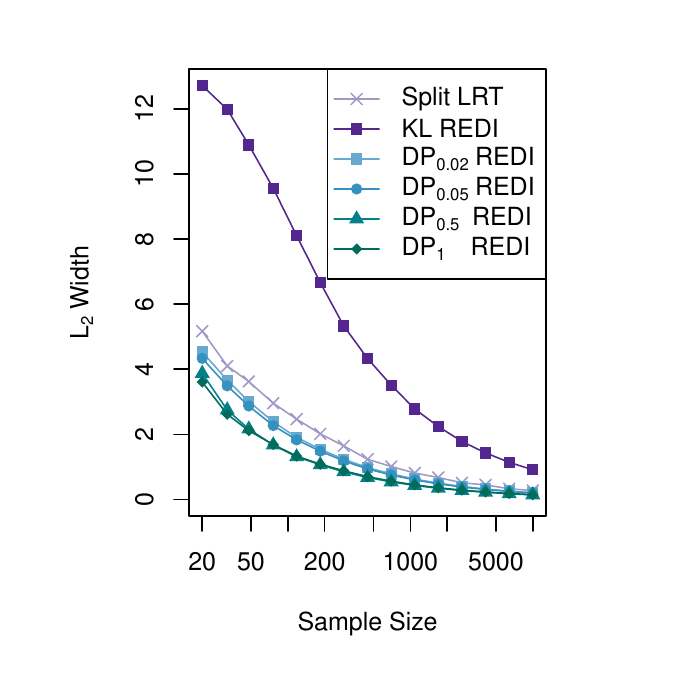}
    \end{subfigure}    
    \caption{Performance of confidence sets ($\Chat_{sLRT}$, $\Chat_{\KL\Redi}$, $\Chat_{\DP\Redi}$, $\Chat_{\DP\Wald}$) of forward KL and DP projections on contamination examples over 1000 replications. (Top) Empirical coverage. (Middle) Median $\rho$-size. (Bottom) Median $L_2$ interval width} \label{fig: contam_1d_KLDP}    
\end{figure}

\begin{figure}[!htb]
\centering
    \begin{subfigure}{.32\textwidth}
        \centering
     \caption{Case 1: Symmetric}
        \includegraphics[trim={15 20 25 20},clip,width=\textwidth,height=1.7in]{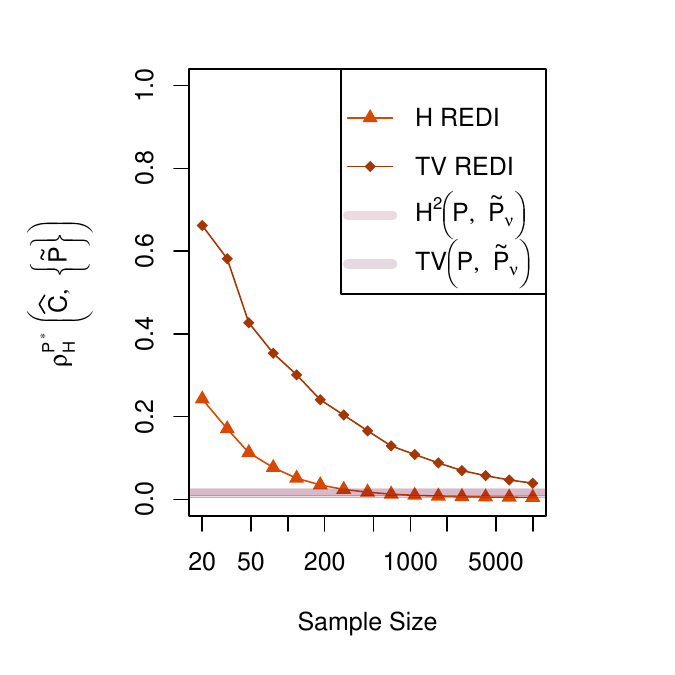}
    \end{subfigure}
    \begin{subfigure}{.32\textwidth}
        \centering
      \caption{Case 2: Asymmetric}
        \includegraphics[trim={15 20 25 20},clip,width=\textwidth,height=1.7in]{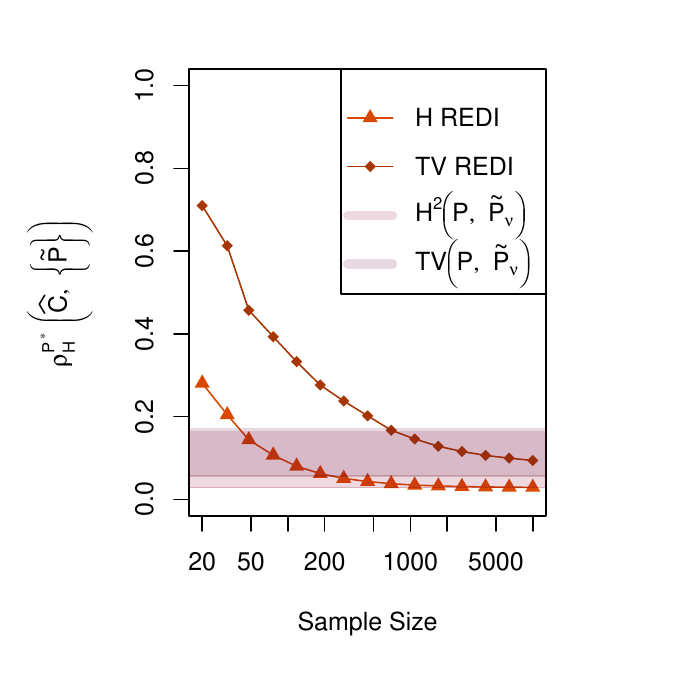}
    \end{subfigure}
    \begin{subfigure}{.32\textwidth}
        \centering
      \caption{Case 3: Heavily asymmetric}
        \includegraphics[trim={15 20 25 20},clip,width=\textwidth,height=1.7in]{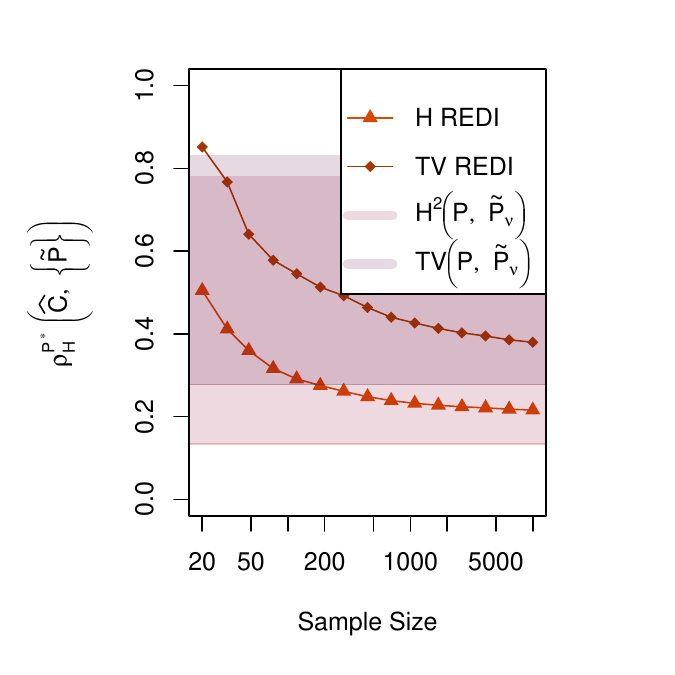}
    \end{subfigure}
    \begin{subfigure}{.32\textwidth}
        \centering
        \includegraphics[trim={15 20 25 20},clip,width=\textwidth,height=1.7in]{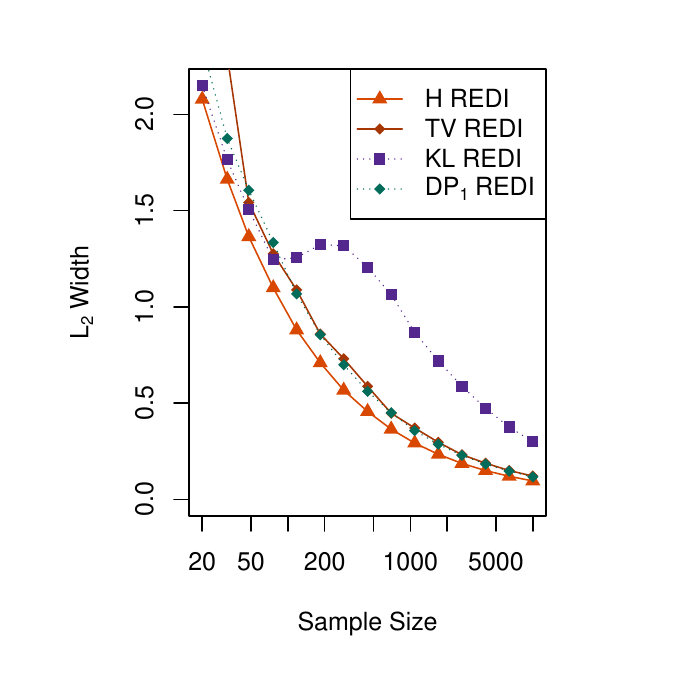}
    \end{subfigure}
    \begin{subfigure}{.32\textwidth}
        \centering
        \includegraphics[trim={15 20 25 20},clip,width=\textwidth,height=1.7in]{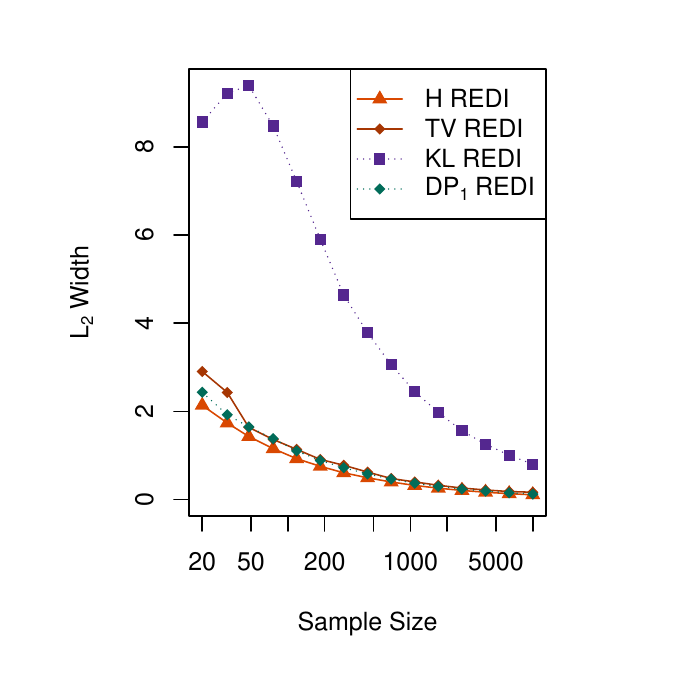}
    \end{subfigure}
    \begin{subfigure}{.32\textwidth}
        \centering
        \includegraphics[trim={15 20 25 20},clip,width=\textwidth,height=1.7in]{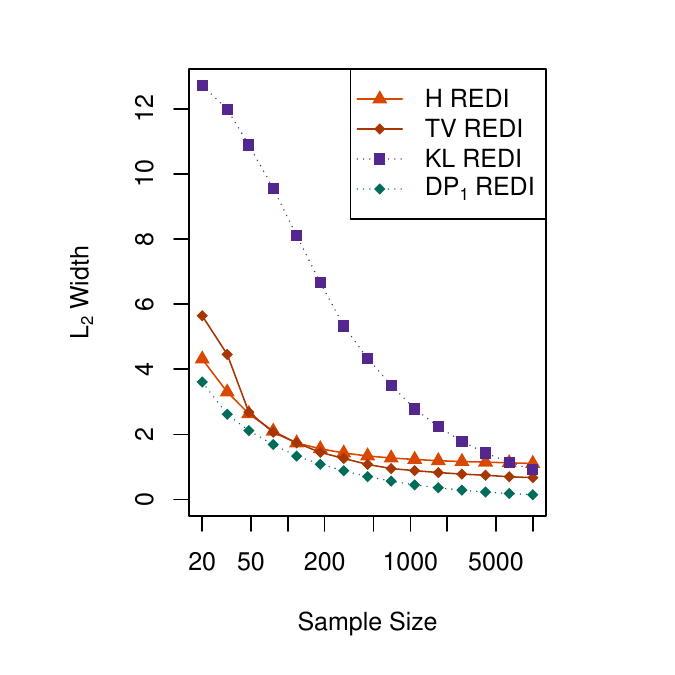}
    \end{subfigure}    
    \caption{Performance of confidence sets of Hellinger and TV projections on contamination examples over 1000 replications. (Top) Median $\rho$-diameter. (Bottom) Median interval width}    
\end{figure}

\section{Additional Results and Technical Details of Section~\ref{sec:LSEM}} \label{sec:LSEM_supp}

\subsection{\Redi sets for LSEM}

\textbf{KL \Redi}
Recall that KL\Redi utilizes log likelihood ratios as in sLRT. We follow the same strategy of profile sLRT with profile likelihood ratio and use the (corrupted) estimated variance of profile log likelihood ratios to set the critical value. Formally, suppose we observe iid $n$ observations $\Data := \{\bm{X}_1, \ldots \bm{X}_n \} \sim {\truedist}^{\otimes n}$. Randomly split $\Data$ into two equally sized data $\Data_0$ and $\Data_1$. Denote their sample sizes $n_0$ and $n_1$, respectively.
By LSEM construction~\ref{eq:LSEM}, we have the Gaussian log likelihood as follows:
\begin{align*}
  \ell ( \Sigma;  \bm{X}_i) = - \log (2\pi) - \frac{1}{2} \log |\Sigma| - \frac{1}{2} \bm{X}_i^\top \Sigma^{-1} \bm{X}_i.
\end{align*}

Since the functional of interest is $\psi := C(i\to j)$ and thus $\sigma^2$ is a nuisance parameter, we profile out the $\sigma$ for a given $\psi$, i.e.,
\begin{align*}
  \ell_\dagger (\psi; \bm{X}_i) := \ell(\hat{\sigma}^2_{0,\psi} \Lambda; X_i) = - \log (2\pi) - \frac{1}{2} \log \hat{\sigma}^2_{0,\psi} | \Lambda| - \frac{1}{2\hat{\sigma}^2_{0,\psi}} \bm{X}_i^\top \Lambda^{-1} \bm{X}_i,\qquad
  \hat{\sigma}^2_{0,\psi} = \argmax_{C(i\to j) = \psi } \ell (\Sigma; \Data_0),
\end{align*}
where $\Lambda \equiv \Lambda (\psi) $ is a structured correlation matrix depends on $\psi$. Conveniently, $\hat{\sigma}^2_{0,\psi}$ has an analytic solution for Gaussian likelihood:
\begin{align*}
  \hat{\sigma}^2_{0,\psi} = \tr{\Lambda(\psi)^{-1} S_0 } / 2, \quad S_0 = \frac{1}{n_0} \sum_{\bm{X}_i \in \Data_0} \bm{X}_i \bm{X}_i^\top.
\end{align*}
Let us detail $\Lambda(\psi)$. When $\psi \ne 0$, $\Sigma \in \mathcal{M}_{r1}$ and thus
\begin{align}\label{eq:Lambda_M1}
  \Lambda (\psi) = 
  \begin{pmatrix}
    1 & \psi \\
    \psi & 1 + \psi^2
  \end{pmatrix}.
\end{align}
When $\psi = 0$, $\Sigma$ can be any matrix in $\mathcal{M}_{r2}$ and thus
\begin{align}\label{eq:Lambda_M1_0}
  \Lambda (0) = B(\tilde{\xi}) =
  \begin{pmatrix}
    1 + \tilde{\xi}^2 & \tilde{\xi} \\
    \tilde{\xi} & 1
  \end{pmatrix} \text{ s.t. }
  \tilde{\xi} = \argmin_{\xi}\sum_{\bm{X}_i \in \Data_0} \bm{X}_i^\top B(\xi)^{-1} \bm{X}_i = (S_0)_{(i,j)} / (S_0)_{(j,j)},
\end{align}
since $|B(\xi)|  = 1$ for any $\xi $.
Then the test statistic of \emph{profile} KL \Redi is as follows:
\begin{align*}
  T_i (P_\psi, \widehat{P}_1) 
    = \ell_\dagger(\psi; \bm{X}_i) - \ell (\widehat{\Sigma}_1; \bm{X}_i) 
    = - \frac{1}{2} \left( \log \hat{\sigma}^2_{0,\psi} | \Lambda(\psi)| - \log |\widehat{\Sigma}_1| + \bm{X}_i^\top (\Lambda^{-1}(\psi) /  \hat{\sigma}^2_{0,\psi} - \widehat{\Sigma}_1^{-1}) \bm{X}_i\right),
\end{align*}
where $\widehat{\Sigma}_1$ is the pilot estimator of $\Sigma$ with $\Data_1$. A sensible choice of pilot estimator is MLE constrained in $\mathcal{M} = \mathcal{M}_{r1} \cup \mathcal{M}_{r2}$:
\begin{align*}
  \widehat{\Sigma}_1 = \argmax_{\Sigma \in \mathcal{M}} \ell (\Sigma; \Data_1) 
  =   \argmin_{\Sigma \in \mathcal{M}} 
  \left( \log |\Sigma| + \tr{\Sigma^{-1} S_1} \right), \quad S_1 = \frac{1}{n_1} \sum_{\bm{X}_i \in \Data_1} \bm{X}_i \bm{X}_i^\top.
\end{align*}
Finally we define the corresponding profile KL \Redi set for $C(i\to j)$ as
\begin{align}\label{eq:KLREDI_LSEM}
  \Chat_{\KL\Redi, n} = \left\{\psi \in \bbR : \overline{T}_{n_0, \delta}(P_\psi, \pilot) \le  \frac{z_{\alpha}  \hat{s}_{P,\delta}}{\sqrt{n_0}} \right\}.
\end{align}

\textbf{DP \Redi}
Given that DP \Redi statistic is the plug-in estimator of the difference in DP divergences, we adopt the similar profiling as we did for profile KL \Redi set~\ref{eq:KLREDI_LSEM}. Formally, the empirical DP divergence is proportional to the following:
\begin{align*}
  \DP_\beta (\bbP_{n_0} \| P) 
    &\propto \int p^{1+\beta} \d \lambda - \left(1 + \frac{1}{\beta}\right) \frac{1}{n_0} \sum_{\bm{X}_i \in \Data_0} p^\beta (\bm{X}_i)\\
    &= \frac{(2\pi)^{-\beta}}{|\Sigma|^{\beta/2}} \left[ \frac{1}{1+\beta} - \left(1 + \frac{1}{\beta}\right) \frac{1}{n_0} \sum_{\bm{X}_i \in \Data_0}  \exp \left( - \frac{\beta}{2} \bm{X}_i^\top \Sigma^{-1} \bm{X}_i\right)\right] := \frac{1}{n_0} \sum_{\bm{X}_i \in \Data_0} e\DP_\beta (\Sigma; \bm{X}_i),    
\end{align*}
where the last equality is from the fact that $P = \Normal(0, \Sigma)$.

Given that $\Sigma = \sigma^2 \Lambda(\psi)$ s.t. $|\Lambda| = 1$,
\begin{align*}
  e\DP_\beta (\sigma^2 \Lambda(\psi); \Data_0) = \frac{1}{(2\pi)^{\beta} \sigma^{2\beta}} \left[ \frac{1}{1+\beta} - \left(1 + \frac{1}{\beta}\right) \frac{1}{n_0} \sum_{\bm{X}_i \in \Data_0}  \exp \left( - \frac{\beta}{2 \sigma^2} \bm{X}_i^\top \Lambda^{-1} \bm{X}_i\right)\right] 
\end{align*}

Let us profile out the $\sigma$ for a given $\psi$, i.e., 
\begin{align*}
  e\DP_{\beta,\dagger} (\psi; \bm{X}_i) = e\DP_\beta \left(\hat{\sigma}^2_{0,\psi} \Lambda(\psi); \bm{X}_i\right),
  \qquad
  \hat{\sigma}^2_{0,\psi} = \argmin_{\sigma : C(i\to j) = \psi } e\DP_\beta (\sigma^2 \Lambda(\psi); \Data_0).
\end{align*}
When $\psi \ne 0$, $\Lambda$ is the same as \eqref{eq:Lambda_M1} and thus solving above minimization suffices.
When $\psi = 0$, $\Sigma$ can be any matrix in $\mathcal{M}_{r2}$ and thus has the same form as \eqref{eq:Lambda_M1_0} except that
\begin{align*}
  \tilde{\xi} = \argmin_{\xi } e\DP_\beta (\hat{\sigma}^2_{0,0} B(\xi); \Data_0).
\end{align*}
Unlike $\KL$\Redi, $\hat{\sigma}^2_{0,\psi}$ is analytically intractable and thus we iteratively update $\hat{\sigma}^2_{0,0}$ and $\tilde{\xi}$ using numerical optimization (i.e., fixed-point iteration).


Then we define the profile $\DP$ \Redi test statistic as follows:
\begin{align*}
  T_i (P_\psi, \pilot) = e\DP_{\beta, \dagger} (\psi; \bm{X}_i) - e\DP_\beta (\widehat{\Sigma}_1; \bm{X}_i),
\end{align*}
and construct profile DP \Redi set for $C(i\to j)$ analogous to \eqref{eq:KLREDI_LSEM}.

For the pilot estimator, it is sensible to choose minimum DP estimator constrained in $\mathcal{M}$:
\begin{align*}
  &\argmin_{\Sigma \in \mathcal{M}} e\DP_\beta (\Sigma; \Data_1).
\end{align*}

\subsection{Simulation}

We generate $n \in \{100, \ldots, 1000\}$ i.i.d samples following the model \eqref{eq:LSEM} with $\sigma^2 = 1, \beta_{21} = 0$ and thus $C(1 \to 2) = \beta_{12} \in \{0, 0.1, 0.2, 0.5\}$. We then construct 95\% confidence intervals for $C (1\to 2)$ using 6 different methods: \texttt{LRT1}, \texttt{LRT2}, \texttt{sLRT}, \texttt{KLREDI}, and \texttt{DPREDI} (with $\beta = 0.1$ and $1$). We also construct both crossfit and vanila \texttt{sLRT}, \texttt{KLREDI}, and \texttt{DPREDI} set, respectively. For each methods, we assess the empirical coverage by the average number of times that $C(1 \to 2)$ is included in the confidence interval and average maximum width out of 500 replications.

Figure~\ref{fig:coverage_well} shows the empirical coverage of each method. All methods yield valid coverage except the LRT1 slightly under covering in $n=500, C(1\to2) = 0.1$. Compared to \texttt{LRT1}, \texttt{LRT2}, both universal inference and robust universal inference methods shows conservative coverage. 
\begin{figure}[!htb]
\centering
\includegraphics[width=0.9\textwidth, trim={0 0 0 0},clip]{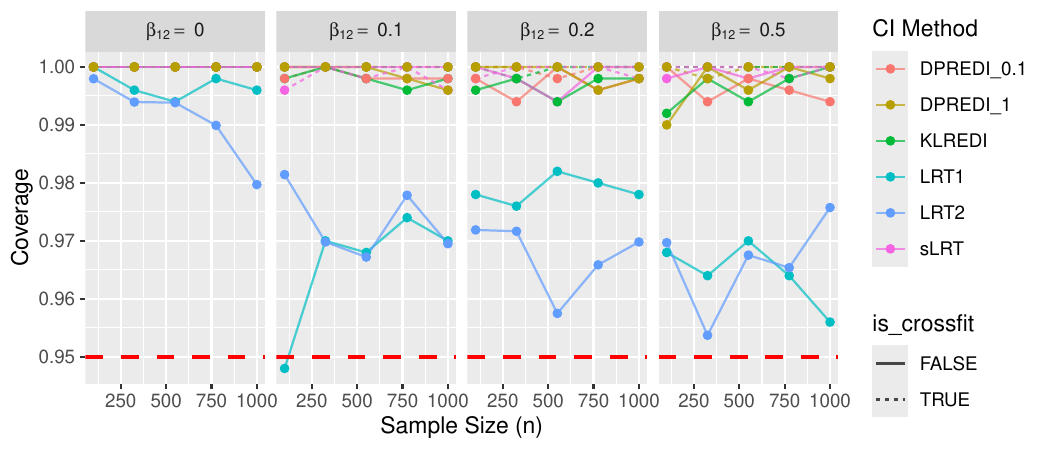}  
\caption{Empirical coverage by varying $C(1\to 2)$ and sample sizes in well-specified setting}\label{fig:coverage_well}
\end{figure}

Figure~\ref{fig:width_well} shows the average maximum width of the confidence set. Here we denote maximum width since the resulting confidence set may be the union of ${0}$ and some interval not including $0$. The width of confidence set the largest for DP\Redi with $\beta=1$, followed by DP\Redi with $\beta=0.1$, KL\Redi, sLRT, LRT1 and LRT2. This ordering underscores the conservatism of (robust) universal inference method as well as of choosing robust target of inference---The higher the $\beta$ is, the more robust the DP projection becomes---when the model is well specified. In all (robust) universal inference methods, crossfit reduced the size of the confidence set without sacrificing the validity. Particularly, the width of the crossfit KL\Redi is comparable to vanilla sLRT. This result precisely matches with the power comparison under Gaussian location model in the Supplementary Material S8.1.
\begin{figure}[!htb]
\centering
\includegraphics[width=0.9\textwidth, trim={0 0 0 0},clip]{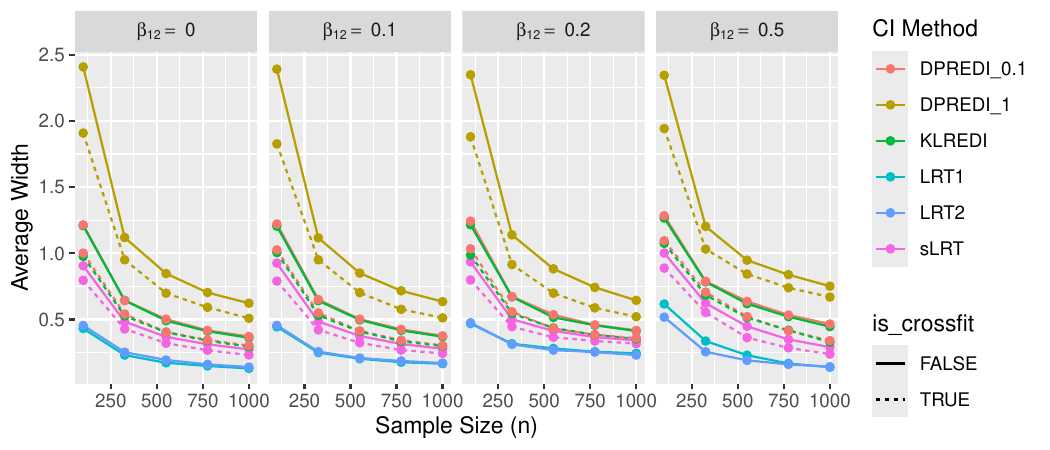}  
\caption{Average maximum width by varying $C(1\to 2)$ and sample sizes in well-specified setting}\label{fig:width_well}
\end{figure}

\section{KL\Redi for Conditional Estimation} \label{sec:KLRedi_regression}

Suppose we observe i.i.d. paired samples $(X_1, Y_1), \dots (X_n,Y_n) \in \calX \times \calY$ drawn from a joint distribution $\truedist_{X,Y} \in \trueclass$ where $\truedist_{Y|X} = \truedist_{Y|X} \otimes \truedist_X$, $\truedist_{Y|X} \in \trueclass_{Y|X}$ is the true conditional response distribution, and $\truedist_X \in \trueclass_X$ is a true marginal regressor distribution following the notation from \citet{buja_models_2019-2}. We assume that $\trueclass_{Y|X}$ is a class of regular conditional probability distributions. Suppose the working model $\model_{Y|X}$ is a collection of (regular) conditional response distributions parametrized by $\theta$, i.e., $\model_{Y|X} = \{P_{Y|X; \theta}: \theta \in \Theta\}$. Let the marginal regressor distribution $P_X \in \calM(\calX)$ be any probability distribution on $\calX$. Then the collection of the joint probability distribution from $\model_{Y|X; \Theta}$ and $\calM(\calX)$ is given by $\model_{\Theta} = \{P_{Y|X; \theta}\otimes P_X : \theta\in\Theta, P_X\in \calM(\calX) \}$.

Suppose our target of inference is the KL projection parameter defined as
\begin{align*}
  \projtheta 
    = \arginf_{\theta\in\Theta, P_X \in \calM(\calX)} \KL \left(\truedist_{X,Y} \| P_{Y|X;\theta} \otimes P_X \right)
    = \arginf_{\theta\in\Theta} \KL \left( \truedist_{Y|X} \| P_{Y|X; \theta} \right),
\end{align*}
where the last equality is due to the fact that the infimum over $\calM(\calX)$ is attained at $P_X = \truedist_X$ irrespective of the choice $\theta$:
\begin{align*}
  \inf_{P_X \in \calM(\calX)} \KL(\truedist_{X,Y} \| P_{Y|X} \otimes P_X )
  &= \inf_{P_X} \int \log \frac{\truedist_{X,Y}}{P_{Y|X} \otimes P_X} \d \truedist_{X,Y}\\
  &= \int \log \frac{\truedist_{Y|X}}{P_{Y|X} } \d \truedist_{Y|X} 
    + \inf_{P_X} \int \log \frac{ \truedist_{X}}{ P_X } \d \truedist_{X}\\
  &= \KL (\truedist_{Y|X} \| P_{Y|X}).
\end{align*}

To design a $\KL\Redi$ set of $\projtheta$, we follow steps analogous to those in Section~5.1 using the natural plug-in estimator of the difference in KL divergence. In particular, we can write the difference in the KL divergence as follows:
\begin{align*}
  \KL(\truedist_{X,Y} \| P_{Y|X;\theta} \otimes P_X ) - \KL (\truedist_{X,Y} \| P_{Y|X;\thetahat_1} \otimes P_X ) 
  &= \int \log \frac{P_{Y|X;\thetahat_1}}{P_{Y|X;\theta} } \d \truedist_{X,Y}.
\end{align*}
Therefore, the natural (split) plug-in estimator of the difference in KL divergence is 
\begin{align*}
  \overline{T}_{n_0} = \frac{1}{n_0} \sum_{i\in\Data_0} T_i (\theta, \thetahat_1), \qquad  T_i (\theta, \thetahat_1) \equiv T\left((X_i,Y_i); \theta, \thetahat_1 \right) = \log \frac{p_{Y|X_i; \thetahat_1} (Y_i)}{p_{Y|X_i; \theta} (Y_i)},
\end{align*}
where $\thetahat_1$ is a pilot estimate constructed from the split sample $\Data_1$. We can now construct the $\KL\Redi$ set $\Chat_{\KL\Redi, n}$ with $\delta$-corrupted samples analogous to (17).

\end{document}